\documentclass[12pt]{article}
\usepackage[T1]{fontenc}
\usepackage[
left=0.6in,
right=0.6in,
top=0.6in,
bottom=0.6in
]{geometry}   % Reduced margins (change to 0.75in if required)
\usepackage{setspace}
\usepackage{amsmath,amssymb} 
\usepackage{amsthm}
\newtheorem{remark}{Remark}
\usepackage{algorithm} 
\usepackage{algpseudocode}
\usepackage{float}
\usepackage{placeins}

\usepackage{parskip}
\usepackage{mathrsfs}
\usepackage{graphicx}
\usepackage{xcolor}
\usepackage{subcaption}
\usepackage{forest}
\usepackage{amsmath}
\usepackage{amssymb}
\usepackage{hyperref}
 \usepackage{booktabs}
\usepackage{longtable}
\usepackage{array}
\usepackage{tabularx}

\newtheorem{definition}{Definition}
\newtheorem{lemma}{Lemma}
\newtheorem{theorem}{Theorem}
\newtheorem{corollary}{Corollary}

\title{\bfseries The Surplus Parking Gathering Problem in Infinite Grids}

\author{
Animesh Maiti$^{1}$ \and
Abhinav Chakraborty$^{2}$ \and
Subhash Bhagat$^{1}$
}

\date{}

\begin{document}

\maketitle

\begin{center}
{\small
$^{1}$Department of Mathematics\\
Indian Institute of Technology Jodhpur, Rajasthan, India\\
Email: \texttt{animesh7.iitj@gmail.com, sbhagat@iitj.ac.in} \\[1em]

$^{2}$Department of Mathematics\\
Birla Institute of Technology Mesra, Ranchi, Jharkhand, India\\
Email: \texttt{abhinavchakraborty@bitmesra.ac.in}
}
\end{center}

\begin{abstract}
% Your abstract goes here.
In this paper, we introduce the \emph{Surplus Parking Gathering Problem} ($\mathcal{SPG}$), a new coordination problem for robots deployed on an infinite grid. The input consists of a set of designated parking nodes, each associated with a prescribed capacity, while the total number of robots exceeds the total parking capacity. The objective is to saturate every parking node exactly according to its capacity while gathering all remaining surplus robots at a common grid node that is not specified a priori. The robots are assumed to be autonomous, anonymous, oblivious, identical, disoriented, and homogeneous. We consider the asynchronous (\textsc{async}) model with global visibility and global strong multiplicity detection. 

We first establish necessary conditions for the solvability of $\mathcal{SPG}$ by characterizing the initial configurations that admit no deterministic distributed algorithm. For all the remaining solvable configurations, we present a deterministic distributed algorithm that correctly solves the problem. The proposed algorithm proceeds in several phases and avoids collisions throughout its execution. We prove that the algorithm terminates in finite time and, upon termination, every parking node is saturated according to its prescribed capacity while all surplus robots are gathered at a uniquely determined gathering node. We further analyze the move complexity of the proposed algorithm, obtaining an $O(n(a+b)+n^2)$ upper bound together with an $\Omega(n(a+b))$ worst-case lower bound for the SPG problem.

\end{abstract}

              % typeset the header of the contribution

%
%
\section{Introduction}
The coordination of large collections of autonomous mobile robots has been one of the central research topics in distributed computing and swarm robotics over the past three decades. Fundamental tasks such as gathering \cite{10.1007/3-540-45061-0_90,doi:10.1137/S009753979628292X,doi:10.1137/S0097539704446475,Flocchini2005}, pattern formation \cite{FLOCCHINI2008412,bose2020arbitrary,CICERONE20231}, exploration \cite{10.1007/978-3-540-77096-1_8, baldoni:inria-00333179} and mutual visibility \cite{9732663,ADHIKARY2022218,DBLP:journals/iandc/LunaFCPSV17} have been extensively studied under various computational models. These problems are motivated by various applications including automated parking systems, environmental monitoring, search-and-rescue operations, and precision agriculture \cite{inproceedings,inproceedings1}. The primary objective is to design distributed algorithms that enable robots to accomplish a common task despite their limited computational capabilities and the absence of explicit communication.

In the traditional robot model, the robots are assumed to be \textit{autonomous} (there is no central coordinator), \textit{anonymous} (they have no unique identifiers), \textit{homogeneous} (all robots execute the same deterministic algorithm), and \textit{identical} (they are physically indistinguishable). Depending on the underlying computational model, robots may operate either in a continuous Euclidean space or on the nodes of a graph. Robots are generally modeled as dimensionless points in Euclidean space in most existing work. In recent years, graph-based models have attracted considerable attention because they are able to capture many structured environments where robot movements are constrained by the topology of the underlying network. In this work, we consider robots deployed on the nodes of an infinite grid, where robot movements are restricted to adjacent grid nodes \cite{bhagat2022gathering,DISTEFANO2017377}.

The robots do not share a common global coordinate system; instead, each robot is equipped with its own local coordinate system whose origin is located at its current position. The directions and orientations of the coordinate axes may differ from one robot to another. However, since the robots are deployed on the nodes of an infinite square grid, they share a common unit of distance, namely the distance between two adjacent grid nodes, commonly referred to as the \emph{one-hop distance}.

In some computational models, the robots are assumed to have agreement on the directions and orientations of one or both coordinate axes. Furthermore, the robots may share a common sense of \textit{chirality}, i.e., a common notion of clockwise orientation. However, in this paper, we assume that the robots are fully disoriented. That is, they do not agree on either the directions or the orientations of their local coordinate axes, nor do they possess a common chirality.

Each robot is equipped with a \textit{visibility sensor}, that enables it to observe the positions of other robots. Depending on the model, the visibility range may be limited or global. In this paper, we assume \emph{global visibility}; hence, every robot can observe the positions of all other robots in the system. The robots are assumed to be \emph{silent}; that is, they do not possess any explicit means of communication. The grid also contains a set of designated nodes, called \emph{parking nodes}, each associated with a prescribed parking capacity.

At any given time, a robot is either active or inactive. If a robot is active, it executes a Look–Compute- Move (LCM) cycle. During the \textit{Look} phase, a robot obtains a snapshot of the current robot configuration and the locations of the parking nodes. In the \textit{Compute} phase, it computes a destination node based on the snapshot obtained during the Look phase. The computed destination node may coincide
with its current position. During the \textit{Move} phase, it moves toward the computed destination. If the computed position is the current position, then a robot performs a \textit{null movement}. After completing the cycle, the robot becomes inactive until it is activated again by the scheduler. In this paper, we assume instantaneous movement. Consequently, robots are observed only at grid nodes and never while traversing the edges of the grid graph. At any instant of time, a robot may move toward one of its adjacent nodes, and the movement must be only along the edges of the input grid graph. The robots are \textit{oblivious}; that is, they possess no persistent memory and cannot recall any information from previous Look–Compute–Move (LCM) cycles, including past observations, computations, or movements. 

The activation of the robots and the execution of the phases in their computational cycles are governed by a scheduler. Depending on the synchronization assumptions, three standard computational models are commonly considered in the literature: \textit{Fully Synchronous} (\textsc{fsync}), \textit{Semi-Synchronous} (\textsc{ssync}), and \textit{Asynchronous} (\textsc{async}). In the \textsc{fsync} model, all robots are activated simultaneously and execute their Look--Compute--Move cycles. The \textsc{ssync} model relaxes this assumption by allowing only a subset of the robots to be activated in each round; however, all activated robots complete their computational cycles synchronously within that round. In contrast, the \textsc{async} model assumes no global notion of time. Robots are activated independently by an adversarial scheduler, and the duration of each phase of the Look--Compute--Move cycle is finite but unpredictable. Consequently, the Look, Compute, and Move phases of different robots may overlap arbitrarily. Furthermore, the robots are not equipped with motion-detection capabilities and therefore cannot distinguish between moving and stationary robots. As a result, a robot may perform its computation based on an outdated snapshot of the configuration. Throughout this paper, we consider the \textsc{async} model. The scheduler is assumed to be fair, ensuring that every robot is activated infinitely often.

The robots may be equipped with multiplicity detection, which enables them to distinguish nodes occupied by multiple robots. Depending on the computational model, multiplicity detection may be either weak or strong. In this paper, we assume that the robots possess a global strong multiplicity detection. That is, every robot can determine the exact number of robots located at each occupied node in the configuration. Under local strong multiplicity detection, a robot can determine only the number of robots occupying its own node, rather than the multiplicities at all occupied nodes. Global strong multiplicity detection is essential for solving our problem, since otherwise the robots cannot determine whether the desired terminal configuration has been reached.

We allow robots to be initially deployed on parking nodes. Each parking node is assigned a fixed capacity that determines the maximum number of robots that may occupy the node at any time during the execution. The capacities of the parking nodes are part of the input and are known to all robots. We assume that the total number of robots is strictly greater than the sum of the capacities of all parking nodes. 

The $\mathcal {SPG}$ problem combines two fundamental coordination tasks that have traditionally been studied independently. The first is the \emph{parking} problem, in which robots are required to occupy a prescribed set of target locations. In our setting, this task is further constrained by the requirement that each designated parking node be occupied by exactly the number of robots specified by its capacity \cite{chakraborty2025parking}. The second is the classical \emph{gathering} problem, in which all robots must converge to a single node that is not known in advance.
Formally, the objective of this problem is to transform any initial configuration into a final configuration such that, at some finite time ($t>0$), every parking node contains exactly as many robots as its capacity. The remaining robots, referred to as \textit{surplus robots}, must gather at a common grid node that is not specified in advance and must be determined autonomously by the robots.
Consequently, the robots must simultaneously determine which robots should occupy the parking nodes and which robots should eventually form the surplus gathering, while preserving the correctness of both tasks throughout the execution.

\subsection{Motivation}

The $\mathcal {SPG}$ problem is motivated by both algorithmic and practical considerations. From an algorithmic perspective, $\mathcal {SPG}$ integrates two fundamental coordination tasks that have traditionally been studied independently. The parking problem assumes that every robot can eventually be assigned to a target location \cite{chakraborty2025parking}, whereas classical gathering algorithms require all robots to converge to a single location. In contrast, $\mathcal {SPG}$ requires both objectives to be achieved simultaneously under capacity constraints, i.e., the designated parking nodes must be saturated exactly according to their prescribed capacities, while all remaining surplus robots must gather at a common node that is not specified a priori.

The problem is also motivated by practical applications involving capacity-constrained service locations, such as autonomous parking facilities, charging stations in robotic warehouses, and docking hubs in logistics systems. In these environments, a limited number of robots can be accommodated at designated service stations, while the remaining robots must wait in an organized manner until resources become available. 

Furthermore, most classical robot coordination models assume that robots move freely as dimensionless points in a continuous Euclidean space, where infinitesimally small movements with arbitrary precision can be performed to avoid collisions. In some models, robots are even allowed to execute guided movements, that is, to follow prescribed curves during motion \cite{10.1007/s00446-018-0325-7,PATTANAYAK2019145}. Such assumptions, however, are often unrealistic in practical settings. Many real-world robotic systems operate in structured environments, such as warehouses, automated parking facilities, and agricultural fields, where robot movements are constrained by aisles, tracks, or other predefined pathways. Motivated by these applications, we model the environment as an infinite grid and restrict robot movements to its edges.

\subsection{Related Work}
The coordination of autonomous mobile robots has been extensively studied in distributed computing under a variety of computational models and environmental settings. Among the most fundamental coordination problems are gathering
\cite{10.1007/3-540-45061-0_90,doi:10.1137/S009753979628292X,
doi:10.1137/S0097539704446475,Flocchini2005},
pattern formation
\cite{FLOCCHINI2008412,bose2020arbitrary,CICERONE20231} and
mutual visibility
\cite{9732663,ADHIKARY2022218,DBLP:journals/iandc/LunaFCPSV17}, each aiming to enable a collection of autonomous robots to accomplish a common task despite severe computational limitations.

Gathering is one of the most extensively studied coordination problems for autonomous mobile robots. The objective is to bring all robots to a single location that is not specified a priori. Depending on the computational model, gathering has been investigated under various assumptions. In recent years, considerable attention has been devoted to gathering in discrete environments, where robots are deployed on the nodes of anonymous graphs \cite{d2016gathering,DISTEFANO2017377,klasing2008gathering,GOSWAMI2025114930}. Klasing et al. \cite{klasing2008gathering} investigated the gathering problem on anonymous rings and established that deterministic gathering is impossible in the absence of weak multiplicity detection. D’Angelo et al. \cite{d2016gathering} addressed deterministic gathering on finite grids. Di Stefano et al. \cite{DISTEFANO2017377} investigated the optimal gathering problem on infinite grids and proposed a distributed algorithm that gathers all robots at a Weber point while minimizing the total travel distance of the swarm. Their algorithm assumes global strong multiplicity detection, where the Weber point is defined as the graph node minimizing the sum of the shortest-path distances from all robot positions. Shibata et al. \cite{DBLP:conf/ipps/ShibataOS00K21} considered the gathering problem for a system of seven autonomous mobile robots deployed on a triangular grid. Goswami et al. \cite{GOSWAMI2025114930} considered the gathering problem of $n \geq 2$ mobile robots on an infinite triangular grid, assuming that the robots possess limited visibility. They further proved that gathering on a triangular grid with $1$-hop vision of robots is not possible even under a \textsc{fsync} scheduler if the robots do not agree on any axis.

Arbitrary Pattern Formation (\textsc{apf}) is one of the fundamental coordination problems in swarm robotics, in which a collection of autonomous mobile robots is required to form an arbitrary geometric pattern specified as input. The problem was first introduced in \cite{Sugihara1996DistributedAF}, where the authors characterized the class of formable patterns using the notion of symmetricity. Subsequently, Flocchini et al. \cite{FLOCCHINI2008412} investigated \textsc{apf} in the asynchronous (\textsc{async}) model for oblivious robots. Bose et al. \cite{bose2020arbitrary} further investigated \textsc{apf} in the Euclidean plane under the opaque robot model, where the visibility of a robot may be obstructed by other robots. In contrast to \textsc{apf}, which requires robots to realize a prescribed geometric pattern, $\mathcal {SPG}$  requires robots to satisfy prescribed parking capacities while simultaneously gathering the remaining surplus robots.

The \textit{parking} problem \cite{chakraborty2025parking} is another important coordination problem in swarm robotics, where robots are required to occupy a prescribed set of designated parking nodes while satisfying application-specific constraints. Unlike gathering, the objective is not to bring all robots to a common location but to distribute them among designated parking nodes. The parking problem may be regarded as a variant of the partitioning problem \cite{inproceedings3}, where robots partition themselves into $m$ groups and converge to separate regions.

\begin{table}[htbp]
\centering
\renewcommand{\arraystretch}{1.10}
\small
\resizebox{0.80\linewidth}{!}{%
\begin{tabular}{|p{0.18\linewidth}|p{0.22\linewidth}|p{0.24\linewidth}|p{0.21\linewidth}|}
\hline
\textbf{Problem} & \textbf{Domain} & \textbf{Objective} & \textbf{Reference} \\
\hline

Gathering &
Plane, rings, grids &
Gather all robots at one location &
\cite{klasing2008gathering,d2016gathering,DISTEFANO2017377} \\
\hline

Triangular-grid gathering &
Triangular grid &
Gather under limited visibility &
\cite{DBLP:conf/ipps/ShibataOS00K21,GOSWAMI2025114930} \\
\hline

Pattern formation &
Plane or regular grids &
Form a prescribed pattern &
\cite{Sugihara1996DistributedAF,FLOCCHINI2008412,bose2020arbitrary,CICERONE20231} \\
\hline

Parking &
Infinite grid &
Occupy designated parking nodes &
\cite{chakraborty2025parking} \\
\hline

Fixed-point / meeting-node problems &
Plane or infinite grid &
Occupy fixed points or meet at designated nodes &
\cite{10.1007/978-3-642-17653-1_1,DBLP:journals/dc/CiceroneSN19a,bhagat2022gathering} \\
\hline

Mutual visibility / exploration &
Grids and graphs &
Ensure visibility or explore nodes &
\cite{9732663,ADHIKARY2022218,10.1007/978-3-540-77096-1_8} \\
\hline

$\mathcal {SPG}$ &
\textbf{Infinite grid with capacities} &
\textbf{Saturate parking nodes and gather surplus robots} &
\textbf{This paper} \\
\hline

\end{tabular}%
}
\caption{Comparison with related coordination problems.}
\label{tab:related-work}
\end{table}

Several related coordination problems also require robots to occupy prescribed target locations. Fujinaga et al. \cite{10.1007/978-3-642-17653-1_1} introduced the notion of fixed points while studying the \textit{landmark covering problem} in the Euclidean plane. The problem requires the robots to reach a configuration in which they all occupy a common fixed point or landmark. Building on this concept, Cicerone et al. \cite{DBLP:journals/dc/CiceroneSN19a} investigated the \textit{Embedded Pattern Formation} problem and proposed a distributed algorithm that guarantees each robot occupies a unique fixed point within finite time without assuming common chirality. Similarly, the $k$-circle formation problem \cite{DBLP:journals/tcs/DasCBM22,DBLP:journals/algorithms/BhagatDCM21} requires robots to form disjoint circles centered at predefined locations, with each circle containing exactly $k$ robots positioned at distinct points on its circumference. Bhagat et al. \cite{bhagat2022gathering,DBLP:journals/ijfcs/BhagatCDM23} subsequently extended this line of research by studying the \textit{Gathering over Meeting Nodes} problem on infinite square grids, where the robots are deployed on grid nodes and a subset of the nodes is designated as meeting nodes. Unlike these problems, $\mathcal {SPG}$  requires robots to simultaneously satisfy parking-capacity constraints while gathering the remaining surplus robots at a common location.
Other coordination problem investigated on infinite grids include the \textit{Mutual Visibility}
\cite{9732663,ADHIKARY2022218,DBLP:journals/iandc/LunaFCPSV17}. The Mutual Visibility problem requires the design of a distributed algorithm that enables robots to relocate to distinct positions such that no three of them become collinear. A comparison of $\mathcal {SPG}$  with closely related coordination problems is presented in
Table~\ref{tab:related-work}.

Although the above coordination problems have been extensively investigated, they have traditionally been studied independently. To the best of our knowledge, no existing work considers the simultaneous execution of parking and gathering under parking-capacity constraints on infinite grids. We address this gap by introducing the $\mathcal {SPG}$ , which integrates these two objectives within a unified distributed framework and introduces new challenges arising from symmetry, asynchronous execution, collision avoidance, and parking-capacity constraints.

\subsection{Technical Challenges}
The $\mathcal {SPG}$ problem poses several challenges that do not arise in classical gathering or parking problems. First, the robots must simultaneously accomplish two interdependent objectives. Unlike the parking problem, not all robots can be assigned to the designated parking nodes because the total number of robots exceeds the combined capacities of all parking nodes. Unlike the classical gathering problem, however, a subset of the robots must always be parked at a parking node in accordance with the specified capacity. Therefore, the robots must simultaneously determine, in a distributed manner, which robots will occupy the parking nodes and which will constitute the gathering.

Second, the parking-node configuration may have reflectional or rotational symmetry, and the initial robot configuration may either preserve or break this symmetry. Since the robots are anonymous, oblivious, and execute the same deterministic algorithm, they may be unable to make different movement decisions in symmetric configurations. Therefore, the algorithm must preserve symmetry whenever possible and break it only when necessary to ensure deterministic progress.

Third, the robots operate in the asynchronous (\textsc{async}) model, where they are activated independently and may observe the configuration at different times. As a result, a robot may compute its movement using an outdated view of the system. Moreover, since the robots cannot distinguish between moving and stationary robots, coordinating their movements becomes more difficult than in synchronous models.

Finally, all robot movements must remain collision-free while preserving the desired structural properties of the configuration. It is crucial that the algorithm avoids creating undesired multiplicity nodes that may prevent robots from uniquely identifying the current execution phase. Due to asynchronous activations and movement restrictions imposed by the underlying grid, the challenge is not only to achieve the desired final configuration, but also to maintain the necessary invariants throughout the execution. These challenges show that the $\mathcal{SPG}$ problem cannot be solved by simply combining existing parking and gathering algorithms. They motivate the design of the distributed algorithm presented in the next section.

\subsection{Our Contributions}

We introduce the $\mathcal {SPG}$ problem, a new coordination problem for autonomous mobile robots deployed on an infinite grid. In $\mathcal{SPG}$, a set of designated parking nodes is associated with prescribed capacities, and the objective is to ensure that each parking node is occupied by exactly the number of robots specified by its capacity while all remaining surplus robots gather at a common grid node.

Our main contributions are summarized as follows:
\begin{itemize}
\item We formalize the $\mathcal {SPG}$ problem and its underlying computational model.

\item We characterize the initial configurations for which a deterministic solution to $\mathcal{SPG}$ is impossible.

\item For all remaining solvable configurations, we present a deterministic distributed algorithm that solves the $\mathcal {SPG}$ problem.

\item We show that the proposed algorithm is collision-free despite the movement restrictions imposed by the grid environment and the asynchronous execution of the robots.

\item We prove that the proposed algorithm always terminates in finite time, saturates every parking node according to its prescribed capacity, and gathers all surplus robots at a common node.
\item We establish an upper bound of $O(n(a+b)+n^2)$ and a worst-case lower bound of $\Omega(n(a+b))$ on the move complexity of the proposed algorithm and the $\mathcal {SPG}$  problem, respectively.
\end{itemize}

\subsection{Outline}

The remainder of the paper is organized as follows. Section~\ref{sec:model}
introduces the computational model, notation, and formal definition of the
Surplus Parking Gathering ($\mathcal {SPG}$ ) problem. Section~\ref{sec:impossibility}
characterizes the initial configurations for which the $\mathcal {SPG}$  problem is
unsolvable. Section~\ref{sec:algorithm} presents our deterministic distributed
algorithm for solving the $\mathcal {SPG}$  problem from all solvable configurations.
Section~\ref{sec:correctness} establishes the correctness of the proposed
algorithm and proves that it terminates in finite time. Section~\ref{sec:complexity} presents the move complexity analysis of the proposed algorithm by establishing an upper bound and a worst-case lower bound. Finally, Section~\ref{sec:conclusion} concludes the paper and
discusses possible directions for future research.

\section{Model, Definitions, and Notations}
\label{sec:model}
We consider a system of $n$ mobile robots operating on an infinite square grid. 
The robots are autonomous, anonymous, homogeneous, identical, silent, oblivious, 
dimensionless, and disoriented. Let $\mathscr P_\infty=(\mathbb Z,E')$ denote the 
infinite path graph on the set of integers, where 
$
E'=\{(i,i+1)\mid i\in\mathbb Z\}.
$
The underlying environment is the infinite square grid graph
$
G=(V,E)=\mathscr P_\infty \times \mathscr P_\infty,
$
obtained as the Cartesian product of two infinite path graphs. The vertex set $V$ 
represents the grid nodes, and robot movements are restricted to adjacent grid nodes. 
Let
$
\mathcal R=\{r_1,r_2,\ldots,r_n\}
$
be the set of robots. The grid also contains a finite set of designated parking nodes
$
\mathcal P=\{p_1,p_2,\ldots,p_m\}\subset V.
$
Each parking node has an associated capacity, specifying the number of robots that 
must eventually occupy that node.

\subsection{Basic Notations}

For two grid nodes $u,v\in V$, let $d_{\mathcal M}(u,v)$ denote the Manhattan
distance between $u$ and $v$ in the grid graph $G$. For a robot
$r_i\in\mathcal R$, the node occupied by $r_i$ at time $t$ is denoted by
$r_i(t)$. The multiset of robot positions at time $t$ is denoted by
$
\mathcal R(t)=\{r_1(t),r_2(t),\ldots,r_n(t)\}.
$
Initially, say at $t_0$, all robots occupy distinct grid nodes, that is,
$
r_i(t_0)\neq r_j(t_0) \text{, for all } i\neq j.
$
However, during the execution, multiple robots may occupy the same grid node.

For every node $v\in V$, let
$
\lambda_t(v)=|\{r_i\in\mathcal R:r_i(t)=v\}|
$
denote the number of robots located at $v$ at time $t$. Hence,
$\lambda_t(v)=0$ if no robot occupies $v$.

Each parking node $p\in\mathcal P$ has a prescribed capacity, denoted by
$\kappa(p)$. We extend the capacity function $\kappa$ to all grid nodes by
defining
\[
\kappa(v)=
\begin{cases}
\kappa(p), & \text{if } v=p\in\mathcal P,\\
0, & \text{otherwise}.
\end{cases}
\]
Thus, $\kappa(v)>0$ precisely when $v$ is a parking node.

The configuration of the system at time $t$ is represented by
$
\mathcal C(t)=(\mathcal R(t),\mathcal P,\kappa).
$
The function $\lambda_t$ is not listed as an independent component of
$\mathcal C(t)$ because it is uniquely determined by the multiset
$\mathcal R(t)$.

For every node $v\in V$, we define its status at time $t$ by the pair
$
\Theta_t(v)=(\lambda_t(v),\kappa(v)).
$
The first component of $\Theta_t(v)$ records the number of robots located at
$v$, while the second component records the parking capacity of $v$. If $v$ is
not a parking node, then $\kappa(v)=0$.

Thus, the status pair $\Theta_t(v)$ distinguishes the following cases:
\[
\begin{array}{c|c|c|c}
\text{Node type} & \lambda_t(v) & \kappa(v) & \Theta_t(v) \\ \hline
\text{empty non-parking node} & 0 & 0 & (0,0)\\
\text{single robot on non-parking node} & 1 & 0 & (1,0)\\
\text{robot multiplicity on non-parking node} & \geq 2 & 0 & (\lambda_t(v),0)\\
\text{empty parking node of capacity } c & 0 & c & (0,c)\\
\text{one robot on parking node of capacity } c & 1 & c & (1,c)\\
\text{robot multiplicity on parking node of capacity } c & \geq 2 & c & (\lambda_t(v),c)
\end{array}
\]

\subsection{Configurations and Symmetry}

An automorphism of the infinite grid graph $G=(V,E)$ is a bijection
$\phi:V\to V$ such that, for any two nodes $u,v\in V$,
\[
(u,v)\in E
\quad \Longleftrightarrow \quad
(\phi(u),\phi(v))\in E.
\]
Thus, $\phi$ preserves adjacency in the grid.

An automorphism $\phi$ of $G$ is called an automorphism of the configuration
$\mathcal C(t)$ if it preserves the status of every grid node, that is,
$
\Theta_t(v)=\Theta_t(\phi(v))
\quad \text{for every } v\in V.
$
Equivalently, $\phi$ preserves both the robot multiplicity and the parking-node
capacity at every node. The set of all automorphisms of $\mathcal C(t)$ is
denoted by
$
\operatorname{Aut}(\mathcal C(t)).
$

If $\operatorname{Aut}(\mathcal C(t))$ contains only the identity automorphism,
then $\mathcal C(t)$ is called asymmetric. Otherwise, $\mathcal C(t)$ is called
symmetric.

Since the set of robots and parking nodes is finite, no non-trivial
translational symmetry can occur. Hence, a symmetric configuration in the
infinite square grid may admit only reflectional or rotational symmetry. A
reflectional symmetry is determined by a line of symmetry, which may be
horizontal, vertical, or diagonal. Such a line may pass through either grid nodes
or grid edges. A rotational symmetry is determined by a center of rotation and an
angle of rotation. In the square grid, the possible non-trivial rotation angles
are $90^\circ$ and $180^\circ$.

\begin{definition}[Reflection Map]
\label{def:reflection-map}
Let $\mathcal C(t)$ admit a unique line of symmetry $\mathcal L$, and let
$X$ be a finite set of entities of $\mathcal C(t)$, where an entity is
either an occupied robot node or a parking node. The reflection map with respect
to $\mathcal L$ is the map
$
\mu_{\mathcal L}:X\to X
$
defined as follows. For every $x\in X$, $\mu_{\mathcal L}(x)$ denotes
the entity whose position is the mirror image of the position of $x$ with
respect to $\mathcal L$. If $x$ lies on $\mathcal L$, then $x$ is fixed by the
reflection, and hence
$
\mu_{\mathcal L}(x)=x.
$
For notational simplicity, whenever the line of symmetry $\mathcal L$ is fixed
or uniquely determined by the configuration under consideration, we write
$\mu$ instead of $\mu_{\mathcal L}$.
\end{definition}

\begin{definition}[Rotational Map]
\label{def:rotational-map}
Let $\mathcal C(t)$ admit a rotational symmetry with center $c$ and order
$q\geq 2$. Let $X$ be a finite set of entities of $\mathcal C(t)$,
where an entity is either an occupied robot node or a parking node. The
rotational map about $c$ is the map
$
\rho_c:X\to X
$
defined as follows. For every $x\in X$, $\rho_c(x)$ denotes the entity
whose position is obtained by rotating the position of $x$ about $c$ through the
angle
$
\theta=\frac{2\pi}{q}.
$
Therefore,
$
\rho_c^q(x)=x
\quad \text{for every } x\in X.
$
The set
$
\mathbf{Orbit}_{x}
=
\{x,\rho_c(x),\rho_c^2(x),\ldots,\rho_c^{q-1}(x)\}
$
is called the rotational orbit of $x$ under $\rho_c$.
For notational simplicity, whenever the center of rotation $c$ is fixed or
uniquely determined by the configuration under consideration, we write $\rho$
instead of $\rho_c$.
\end{definition}

Let $\phi\in\operatorname{Aut}(\mathcal C(t))$ be an automorphism of finite
order $q>1$. The cyclic group generated by $\phi$ is denoted by
$
\langle \phi\rangle=
\{\phi^0,\phi^1,\ldots,\phi^{q-1}\},
$
where $\phi^0$ is the identity automorphism.

For a subgroup $\mathcal H\leq \operatorname{Aut}(\mathcal C(t))$, the orbit of
a node $v\in V$ under $\mathcal H$ is defined as
$
\mathcal O_{\mathcal H}(v)=\{\psi(v):\psi\in\mathcal H\}.
$
The orbits induced by $\mathcal H$ form a partition of $V$.

\begin{definition}[Partitive Automorphism]
Let $X\subseteq V$. An automorphism
$\phi\in\operatorname{Aut}(\mathcal C(t))$ is called \emph{partitive on $X$}
if the cyclic group $\langle\phi\rangle$ has order $q>1$ and every node of
$X$ has an orbit of size exactly $q$, i.e.,
$
|\mathcal O_{\langle\phi\rangle}(v)|=q
\quad \text{for every } v\in X.
$
Equivalently, no node of $X$ is fixed by any non-identity element of
$\langle\phi\rangle$.
\end{definition}

\begin{definition}[Partitive Configuration]
A configuration $\mathcal C(t)$ is said to be \emph{partitive on $X\subseteq V$}
if there exists an automorphism
$\phi\in\operatorname{Aut}(\mathcal C(t))$ such that $\phi$ is partitive on
$X$. In this case, the nodes of $X$ are divided into disjoint orbits of equal
size $q>1$ under the action of $\langle\phi\rangle$.

If $F=V\setminus X$ is the set of nodes fixed by the symmetry, then we also say
that $\mathcal C(t)$ is partitive outside $F$. Thus, every node in $X$ has a
non-trivial symmetric counterpart, whereas the nodes in $F$ may remain fixed
under the symmetry.
\end{definition}

For example, suppose that $\mathcal C(t)$ admits a reflection symmetry with
respect to a line $\mathcal L$. If no grid node lies on $\mathcal L$, then the
reflection maps every grid node to a distinct symmetric node. Hence,
$\mathcal C(t)$ is partitive on $V$. On the other hand, if $\mathcal L$ passes
through grid nodes, then the nodes on $\mathcal L$ are fixed, while all nodes
outside $\mathcal L$ occur in symmetric pairs. In this case, $\mathcal C(t)$ is
partitive on $V\setminus \mathcal L$.

Similarly, if $\mathcal C(t)$ admits a rotational symmetry about a center $c$,
then the configuration may be partitive on $V\setminus\{c\}$ whenever $c$ is a
grid node fixed by the rotation. If the center of rotation is not a grid node,
then no grid node is fixed by the center, and the configuration may be partitive
on the whole node set $V$.

\subsection{Minimum Enclosing Rectangle, Strings, Key Corner, and Leading Corner}

For a configuration $\mathcal C(t)$, let
$
\mathcal{MER}_{\mathcal C}(t)
$
denote the minimum grid-aligned rectangle containing all occupied robot nodes and
all parking nodes. Similarly, let
$
\mathcal{MER}_{\mathcal P}
$
denote the minimum grid-aligned rectangle containing only the parking-node set
$\mathcal P$. Both rectangles are aligned with the grid axes. The length of a
side of such a rectangle is measured by the number of grid edges on that side.

Suppose that the side lengths of $\mathcal{MER}_{\mathcal C}(t)$ are $a$ and
$b$. Since the side lengths are measured by the number of grid edges, there are
$a+1$ grid nodes along one side and $b+1$ grid nodes along the other side.
Hence, the total number of grid nodes contained in
$\mathcal{MER}_{\mathcal C}(t)$ is
$
(a+1)(b+1).
$

\begin{figure}[htbp]
    \centering
    \includegraphics[width=0.5\textwidth]{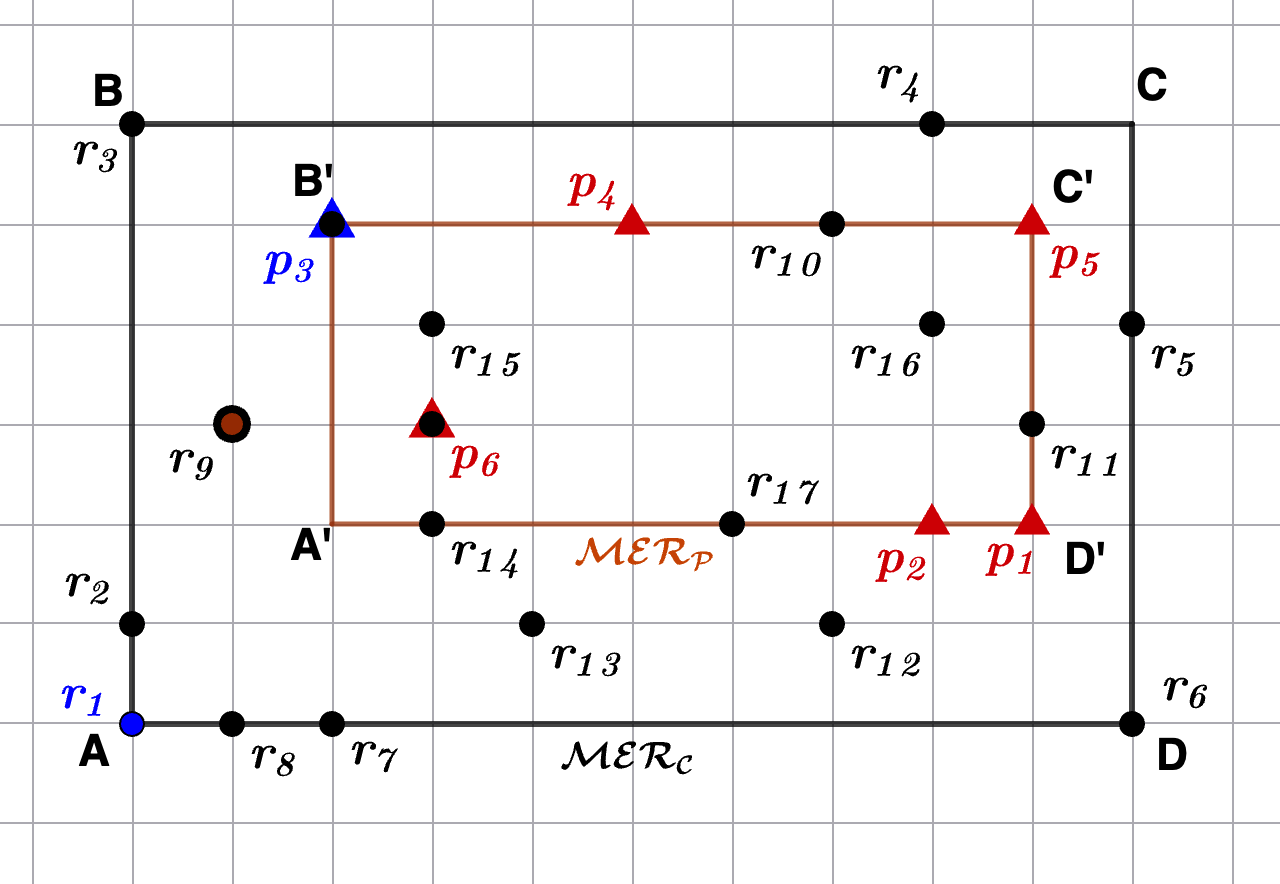}

    \caption[Leading-corner based ordering with respect to $\mathcal{MER}_{\mathcal P}$]
{An illustration of the key corner and leading-corner-based ordering of robots and parking nodes. The black rectangle denotes $\mathcal{MER}_{\mathcal C}$ with $\operatorname{len}(AB)<\operatorname{len}(AD)$, while the brown rectangle denotes $\mathcal{MER}_{\mathcal P}$ with $\operatorname{len}(A'B')<\operatorname{len}(A'D')$. The parking nodes are
$p_1,p_2,\ldots,p_6$ with capacities $2,2,4,1,2,$ and $3$, respectively.
The parking nodes $p_6$ and $p_3$ contain $2$ and $3$ robots, respectively,
and the node represented by $r_9$ is a multiplicity node containing $4$ robots. The key corner of $\mathcal{MER}_{\mathcal C}$ is $A$ with string direction $AD$, and the leading corner of $\mathcal{MER}_{\mathcal P}$ is $B'$ with string direction $B'A'$, with parking string $\rho_{B'}=40000030000010000000000000022002$.}

    \label{fig:leading}
\end{figure}

A scan string associated with a corner of $\mathcal{MER}_{\mathcal C}(t)$ is
defined as the sequence of symbols obtained by visiting all nodes of
$\mathcal{MER}_{\mathcal C}(t)$ exactly once in a prescribed scan order. During
such a scan, each encountered node $v$ contributes the status pair
$
\Theta_t(v)=(\lambda_t(v),\kappa(v)),
$
where $\lambda_t(v)$ denotes the number of robots located at $v$ at time $t$,
and $\kappa(v)$ denotes the capacity of $v$. Therefore, every scan string is a finite sequence of status pairs of
length
$
(a+1)(b+1).
$

The scan strings are compared lexicographically. The underlying order on the
status pairs is defined as follows: for two symbols $(x,y)$ and $(x',y')$,
$
(x,y)\prec (x',y')
$
if either $x<x'$, or $x=x'$ and $y<y'$. Hence, robot multiplicities are compared
first, and parking capacities are used to break ties.

Let $Q$ be a corner of $\mathcal{MER}_{\mathcal C}(t)$, and let $Q_1$ and
$Q_2$ be the two corners adjacent to $Q$. Starting from $Q$, the rectangle can
be scanned in two possible directions, namely in the directions parallel to
$QQ_1$ and $QQ_2$. The corresponding scan strings are denoted by $\mathcal{STR}^{QQ_1}(t)$
 and $\mathcal{STR}^{QQ_2}(t)$, respectively. Thus, each corner of $\mathcal{MER}_{\mathcal C}(t)$ gives rise to two scan strings, and altogether eight scan strings are obtained from the four corners.

First, consider the case where $\mathcal{MER}_{\mathcal C}(t)$ is a non-square
rectangle. For a corner $Q$, among the two sides incident to $Q$, the direction
parallel to the shorter side is chosen as the string direction associated with
$Q$. Suppose, without loss of generality, that the side from $Q$ to $Q_1$ is shorter
than the side from $Q$ to $Q_2$, i.e.,
$
\operatorname{len}(QQ_1)<\operatorname{len}(QQ_2).
$
Then the string associated with $Q$ is defined as
$
\mathcal{STR}^Q(t)=\mathcal{STR}^{QQ_1}(t).
$
The direction parallel to $QQ_1$ is called the string direction associated with
$Q$, while the direction parallel to $QQ_2$ is called the non-string direction
associated with $Q$.

If $\mathcal{MER}_{\mathcal C}(t)$ is a square, then the two strings
$\mathcal{STR}^{QQ_1}(t)$ and $\mathcal{STR}^{QQ_2}(t)$ associated with the same corner $Q$
are compared lexicographically, and the string associated with $Q$ is defined as
$
\mathcal{STR}^Q(t)=\max\{\mathcal{STR}^{QQ_1}(t),\mathcal{STR}^{QQ_2}(t)\}.
$

A corner $Q^*$ of $\mathcal{MER}_{\mathcal C}(t)$ is called a \emph{key corner}
of $\mathcal C(t)$ if
\[
\mathcal{STR}^{Q^*}(t)
=
\max\{\mathcal{STR}^Q(t): Q \text{ is a corner of }
\mathcal{MER}_{\mathcal C}(t)\}.
\]
Any corner that is not a key corner is called a \emph{non-key corner}. If the
maximum string is unique, then the key corner is uniquely determined. In
particular, if $\mathcal C(t)$ is asymmetric, then the lexicographically maximum
string is unique, and hence the key corner and the corresponding string direction are uniquely determined. In Figure~\ref{fig:leading}, the lexicographically maximum scan string associated with the key corner $A$ is  $\mathcal{STR}^{AB}(t)$=((1,0), (1,0), (0,0), (0,0), (0,0), (0,0), (1,0), (1,0), (0,0), (0,0), (4,0), (0,0), (0,0), (0,0), (1,0), (0,0), (0,0), (0,0), (0,0), (3,4), (0,0), (0,0), (0,0), (1,0), (2,3), (1,0), (0,0), (0,0), (0,0), (1,0), (0,0), (0,0), (0,0), (0,0), (0,0), (0,0), (0,0), (0,0), (0,0), (0,0), (0,1), (0,0), (0,0), (0,0), (1,0), (0,0), (0,0), (0,0), (0,0), (0,0), (0,0), (1,0), (0,0), (0,0), (0,0), (1,0), (0,0), (0,0), (0,0), (0,2), (0,0), (1,0), (0,0), (1,0), (0,0), (0,0), (0,2), (1,0), (0,0), (0,2), (0,0), (1,0), (0,0), (0,0), (0,0), (1,0), (0,0), (0,0)).

We define the leading corner of the parking-node set in an analogous way, but
using only parking capacities. For every node $v\in V$, define
$
\pi(v)=\kappa(v).
$
Thus, $\pi(v)=0$ if $v$ is not a parking node, and $\pi(v)$ is the capacity of
$v$ otherwise. For each corner $Q$ of $\mathcal{MER}_{\mathcal P}$, a string
$\rho_Q$ is constructed by scanning $\mathcal{MER}_{\mathcal P}$ and recording
$\pi(v)$ for every visited node $v$.

A corner $Q_{\mathcal P}$ of $\mathcal{MER}_{\mathcal P}$ is called a
\emph{leading corner} if
\[
\rho_{Q_{\mathcal P}}
=
\max\{\rho_Q: Q \text{ is a corner of } \mathcal{MER}_{\mathcal P}\}.
\]
If the maximum string is unique, then the parking-node set admits a unique leading corner (see Figure~\ref{fig:leading}). Otherwise, the parking-node set is symmetric with respect to the corresponding rectangle-string representation.

\subsection{Configuration View}

The configuration view is defined using the key corner and the string direction
introduced above. Suppose that the configuration $\mathcal C(t)$ has a unique
key corner $Q^*$ of $\mathcal{MER}_{\mathcal C}(t)$, and let the corresponding
string direction be fixed by the string $\mathcal{STR}^{Q^*}(t)$.

Starting from $Q^*$ and scanning $\mathcal{MER}_{\mathcal C}(t)$ in this string
direction, the nodes of $\mathcal{MER}_{\mathcal C}(t)$ are visited in a
canonical order. For a node $v\in\mathcal{MER}_{\mathcal C}(t)$, let
$
\operatorname{ind}_{Q^*}(v)
$
denote the position of $v$ in this scan order. The \emph{configuration view} of
$v$ is defined as
$
\operatorname{View}_{\mathcal C(t)}(v)
=
\bigl(\operatorname{ind}_{Q^*}(v),\Theta_t(v)\bigr).
$
Equivalently,
$
\operatorname{View}_{\mathcal C(t)}(v)
=
\bigl(\operatorname{ind}_{Q^*}(v),(\lambda_t(v),\kappa(v))\bigr).
$

Thus, the configuration view of a node records two pieces of information: its
position in the canonical scan order and its status pair, which encodes both the number of robots located at that node and the parking capacity of that node. Whenever the key corner and the corresponding string direction are unique, the
configuration view induces a deterministic ordering of the nodes of
$\mathcal{MER}_{\mathcal C}(t)$. Consequently, occupied robot nodes can be
ordered according to their positions in the scan order. If more than one robot
occupies the same node, then these robots have the same configuration view and
are treated as a multiplicity at that node. Similarly, if the parking-node set $\mathcal P$ has a unique leading corner
$Q_{\mathcal P}$, then the scan of $\mathcal{MER}_{\mathcal P}$ starting from
$Q_{\mathcal P}$ induces a deterministic ordering of the parking nodes. For two
parking nodes $p_i,p_j\in\mathcal P$, we write
$
p_i \prec_{\mathcal P} p_j
$, if $p_i$ appears before $p_j$ in this scan order.

\subsection{The Surplus Parking Gathering Problem}
We assume that $
n>\sum_{i=1}^{m}\kappa_i,
$ and define $
s=n-\sum_{i=1}^{m}\kappa_i
$ to be the number of surplus robots. Assume that $s \geq2$, since these $s$ robots are required for the gathering process. Let $
\mathcal{R}(v,t)=\{\,r\in\mathcal{R}\mid r(t)=v\,\}$. The objective of the $\mathcal {SPG}$ problem (\(\mathcal{SPG}\)) is to transform any initial configuration \(\mathcal C(t_0)\) into a configuration \(\mathcal C(t)\), for some finite time \(t>0\), such that the following conditions hold:

\begin{enumerate}
    \item Every parking node \(p_i\) is saturated at time \(t_0\), that is,
$|\mathcal R(p_i,t)| = \kappa_i,
    \forall i=1,2,\ldots,m,
    $
    where \(\mathcal R(p_i,t)\) denotes the set of robots occupying \(p_i\) and $\kappa_i$ denote the capacity of the parking node $p_i$.

    \item There exists a node \(\mathscr G \in V \setminus \mathcal P\), called the
    \emph{gathering node}, such that all surplus robots occupy \(\mathscr G\), i.e., at time $t$,
    $
    |\mathcal R(\mathscr G, t)| = s.
    $
    No capacity constraint is imposed on the gathering node \(\mathscr G\).

    \item Any robot activated after time \(t\) computes a null move and therefore remains at its current location.
\end{enumerate}

The gathering node \(\mathscr G\) is not specified a priori and must be determined autonomously by the robots during the execution of the algorithm.

A configuration satisfying Conditions~(1)--(3) is called a \emph{surplus parking gathering configuration}. The $\mathcal {SPG}$ problem is solved if, starting from any initial configuration, the robots reach a surplus parking gathering configuration in finite time.

% \begin{lemma} \label{chap4-lemma1}
% Let $\mathcal A$ be any algorithm that solves the parking problem in infinite grids. If there exists an execution of $\mathcal A$ such that the configuration $C(t)$ contains a robot multiplicity at a node that is not a parking node, then $\mathcal A$ cannot solve the parking problem.  
% \end{lemma}

\begin{figure}[htbp]
    \centering
    \includegraphics[width=0.6\textwidth]{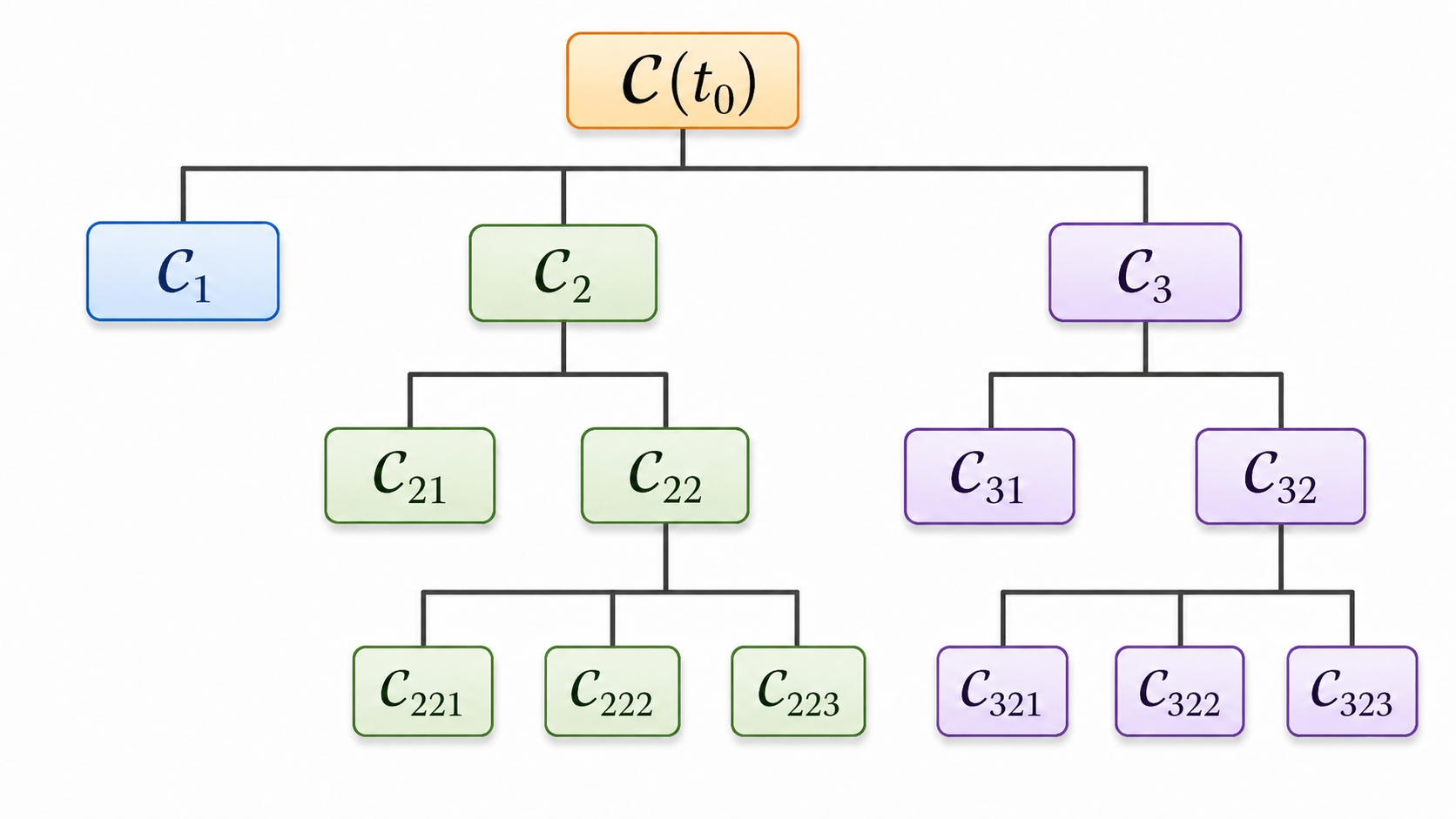}
    \caption{Partition of the configuration space based on the initial configuration $\mathcal C(t_0)$.}
    \label{Treee}
\end{figure}

\subsection{Partitioning of the Initial Configuration $\mathcal C(t_0)$}

The symmetry of the parking-node configuration plays a fundamental role in determining the behavior of the robots. Since anonymous and oblivious robots executing the same deterministic algorithm cannot distinguish symmetric situations, different symmetry classes require different algorithmic strategies. Accordingly, we partition the initial configurations into several classes based on the symmetry of the parking-node configuration and the overall robot configuration. (see Figure~\ref{Treee})

\paragraph{$\mathcal C_1$: Asymmetric parking-node configuration.}
The parking nodes are asymmetric. Consequently, the overall configuration
$\mathcal C(t_0)$ is also asymmetric. An example is shown in Figure \ref{fig:leading}.

\begin{figure}[htbp]
    \centering
    \includegraphics[width=0.4\textwidth]{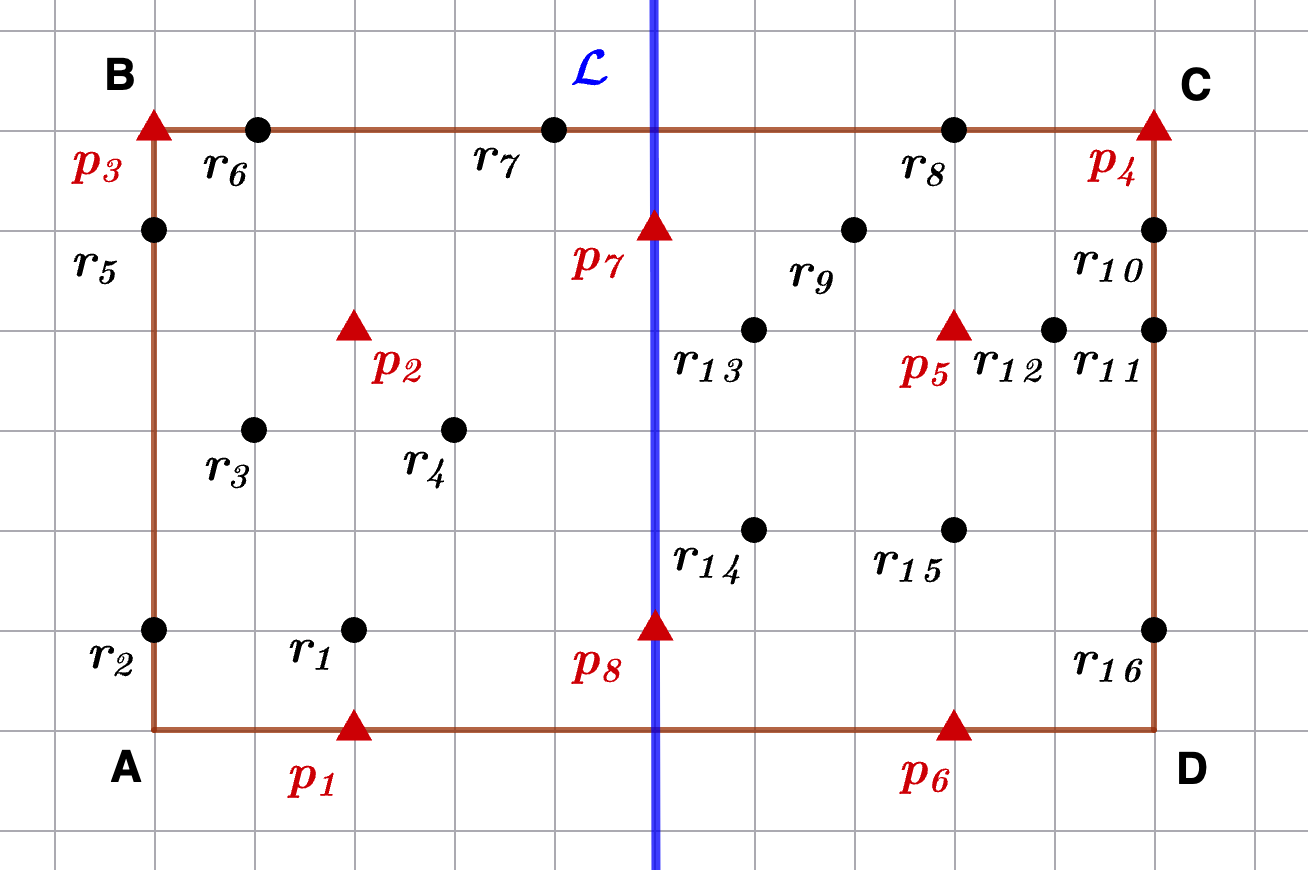}
    \caption{An example of a $\mathcal C_{21}$ configuration, where the parking-node set is symmetric with respect to line $\mathcal L$, and the robot configuration is asymmetric.}
    \label{fig:c21}
\end{figure}

\paragraph{$\mathcal C_2$: Parking nodes with a unique line of symmetry.}
The parking nodes are symmetric with respect to a unique line of symmetry $\mathcal L$. This class is divided according to whether the overall configuration $\mathcal C(t_0)$ is symmetric or asymmetric.

\begin{itemize}

     \item $\mathcal C_{21}$: If $\mathcal C(t_0)$ is asymmetric, then the configuration belongs to $\mathcal C_{21}$. (see Figure~ \ref{fig:c21})
     
    \item $\mathcal C_{22}$: If $\mathcal C(t_0)$ is symmetric with respect to $\mathcal L$, then the configuration belongs to $\mathcal C_{22}$. The class $\mathcal C_{22}$ is further partitioned as follows:

\begin{figure}[htbp]
    \centering

    \begin{subfigure}[t]{0.3\textwidth}
        \centering
        \includegraphics[width=\textwidth]{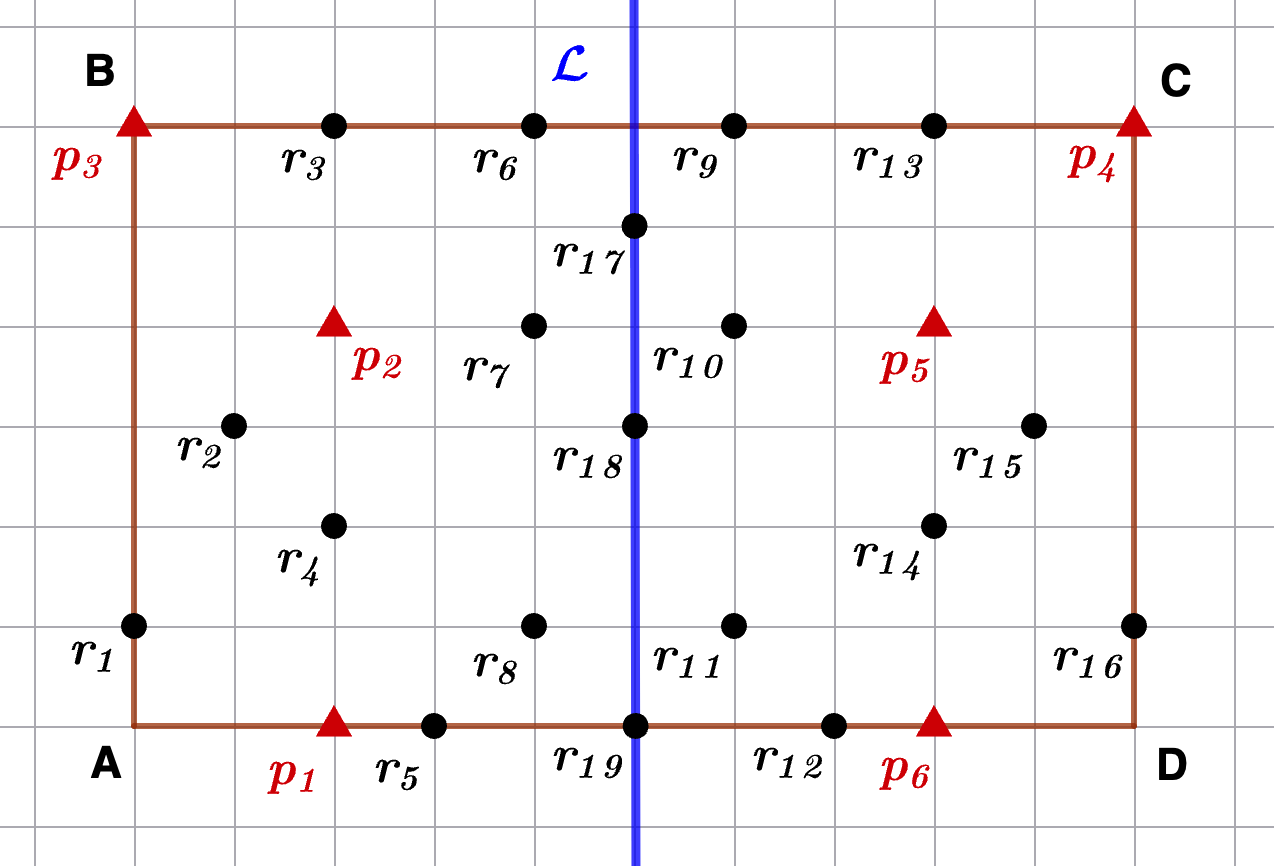}
        \caption{$\mathcal C_{221}$: robots on $\mathcal L$.}
        \label{fig:c221}
    \end{subfigure}
    \hspace{0.03\textwidth}
    \begin{subfigure}[t]{0.3\textwidth}
        \centering
        \includegraphics[width=\textwidth]{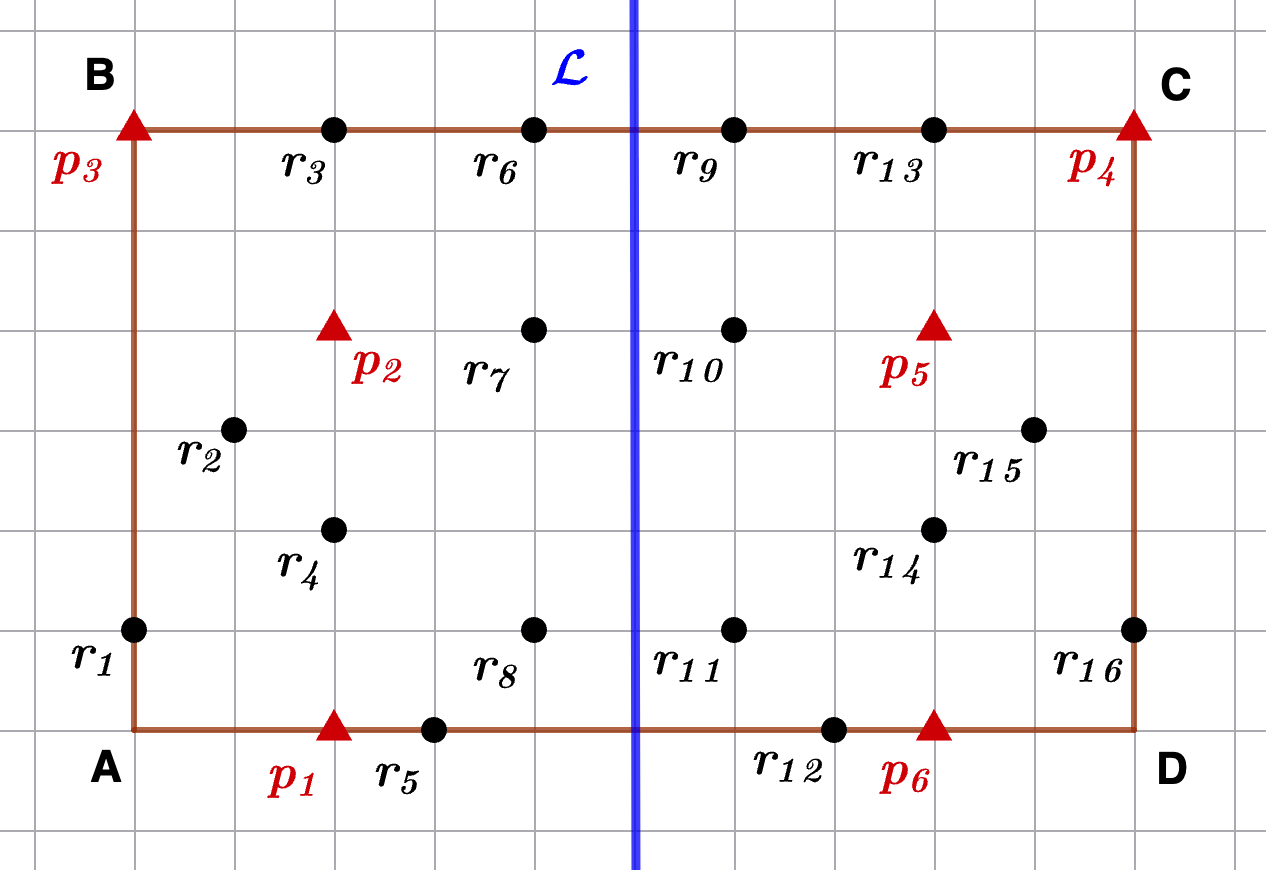}
        \caption{$\mathcal C_{222}$: empty $\mathcal L$.}
        \label{fig:c222}
    \end{subfigure}
    \hspace{0.03\textwidth}
    \begin{subfigure}[t]{0.3\textwidth}
        \centering
        \includegraphics[width=\textwidth]{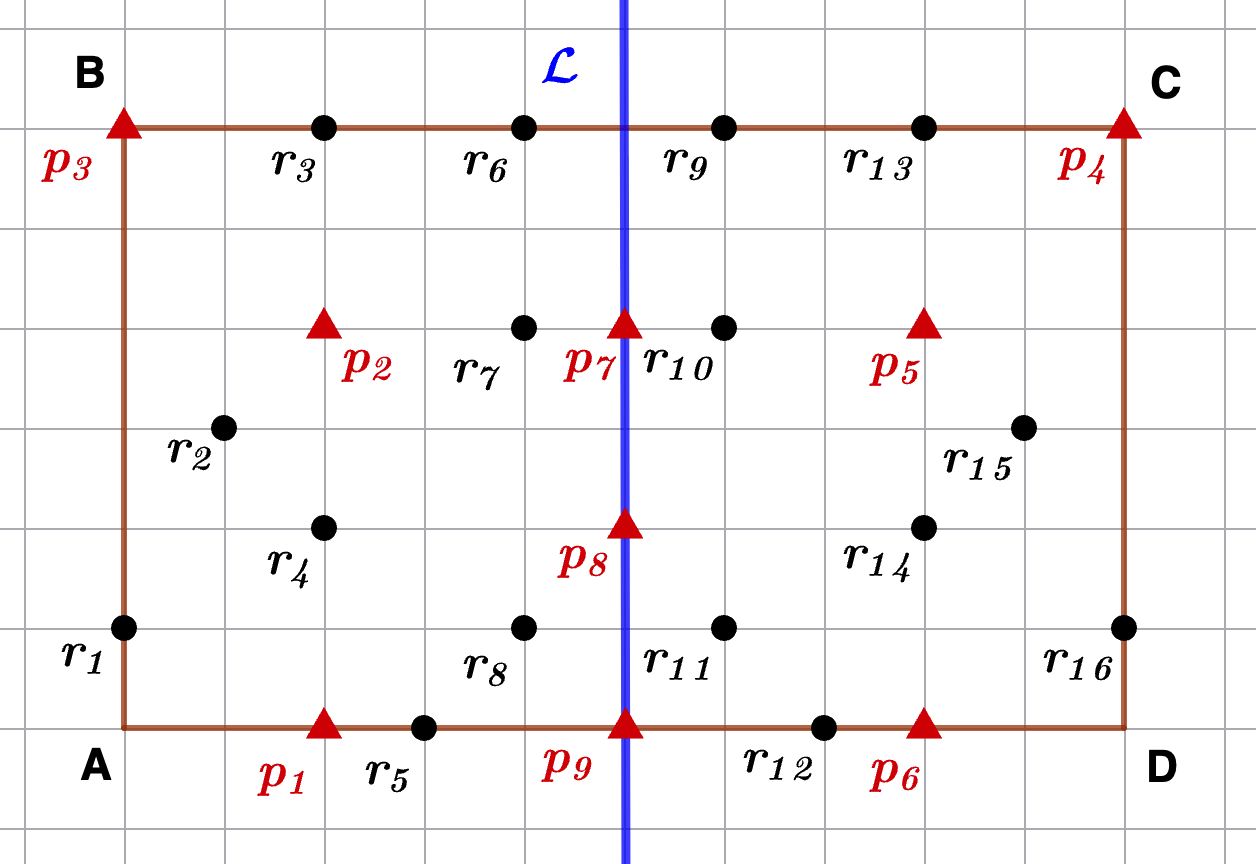}
        \caption{$\mathcal C_{223}$: parking nodes on $\mathcal L$.}
        \label{fig:c223}
    \end{subfigure}

    \caption{Representative configurations for the subcases of $\mathcal C_{22}$.}
    \label{fig:sym1}
\end{figure}

    \begin{itemize}
        \item $\mathcal C_{221}$: At least one robot occupies a position on $\mathcal L$. (see Figure~ \ref{fig:sym1}(a))
         \item $\mathcal C_{222}$: No robot position and no parking node lie on the line of symmetry
$\mathcal L$. (see Figure~ \ref{fig:sym1}(b))
         \item $\mathcal C_{223}$: No robot occupies a position on $\mathcal L$, whereas at least one
parking node lies on $\mathcal L$. (see Figure \ref{fig:sym1}(c))
    \end{itemize}
\end{itemize}

       \begin{figure}[htbp]
    \centering
    \includegraphics[width=0.35\textwidth]{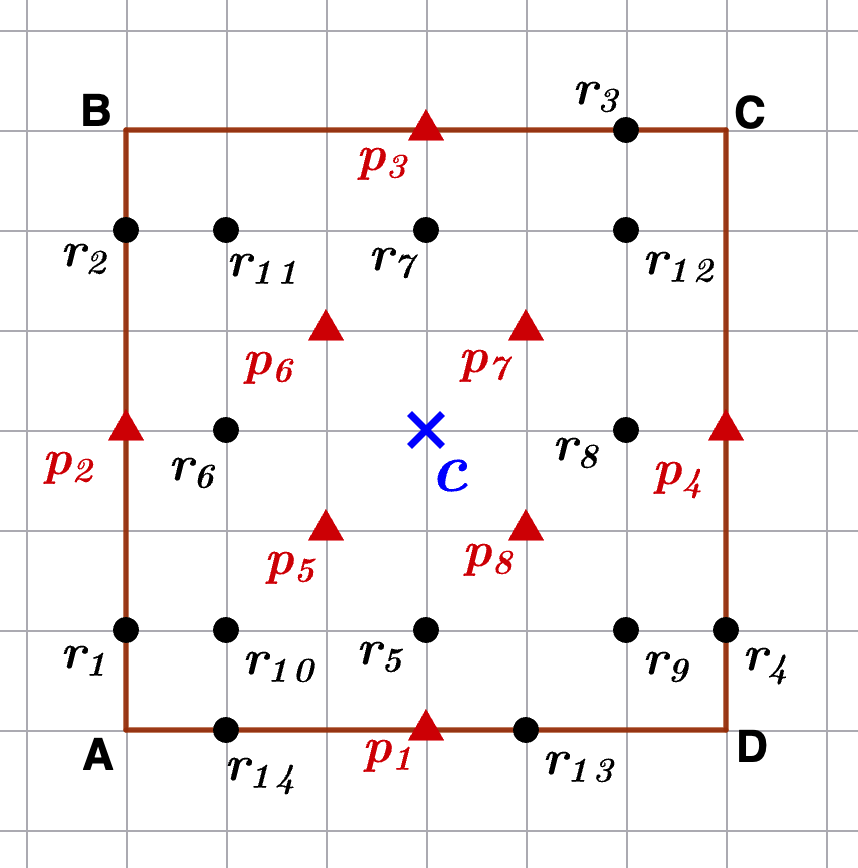}
    \caption{ An example of a $\mathcal C_{31}$ configuration, where the parking-node set is rotationally symmetric, and the robot configuration is asymmetric.}
    \label{fig:c31}
\end{figure}

\paragraph{$\mathcal C_3$: Rotationally symmetric parking-node configuration.}
The parking-node configuration admits rotational symmetry with center $c$. This class is further partitioned according to whether the overall configuration $\mathcal C(t_0)$ is symmetric or asymmetric.

\begin{itemize}
    \item $\mathcal C_{31}$: If $\mathcal C(t_0)$ is asymmetric, then the configuration belongs to $\mathcal C_{31}$. (see Figure~ \ref{fig:c31})

    \item $\mathcal C_{32}$: If $\mathcal C(t_0)$ is rotationally symmetric with respect to $c$, then the configuration belongs to $\mathcal C_{32}$. The class $\mathcal C_{32}$ is further partitioned as follows:

    \begin{itemize}
        \item $\mathcal C_{321}$: Exactly one robot occupies the center of symmetry $c$. (see Figure~ \ref{fig:sym2}(a))

        \item $\mathcal C_{322}$: Neither a robot nor a parking node occupies the center of symmetry $c$. (see Figure~ \ref{fig:sym2}(b))

        \item $\mathcal C_{323}$: A parking node occupies the center of symmetry $c$, but no robot occupies $c$. (see Figure~ \ref{fig:sym2}(c))
    \end{itemize}
\end{itemize}

   \begin{figure}[htbp]
    \centering

    \begin{subfigure}[b]{0.3\textwidth}
        \centering
        \includegraphics[width=\textwidth]{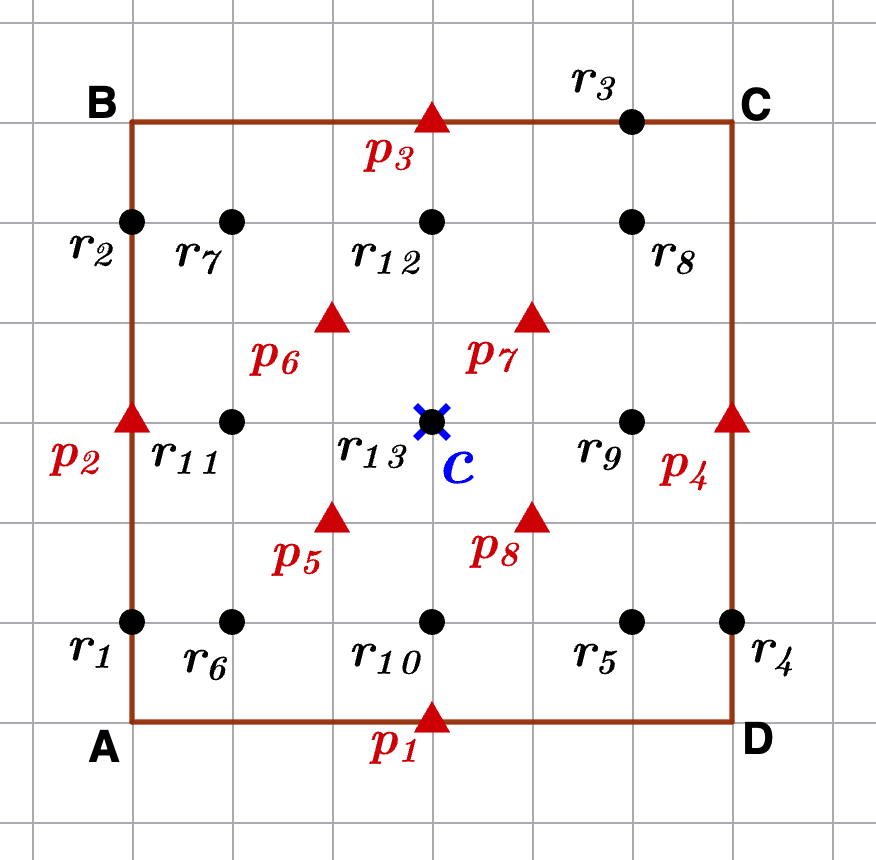}
        \caption{$\mathcal C_{321}$: robot at $c$.}
        \label{fig:c321}
    \end{subfigure}
    \hspace{0.03\textwidth}
    \begin{subfigure}[b]{0.3\textwidth}
        \centering
        \includegraphics[width=\textwidth]{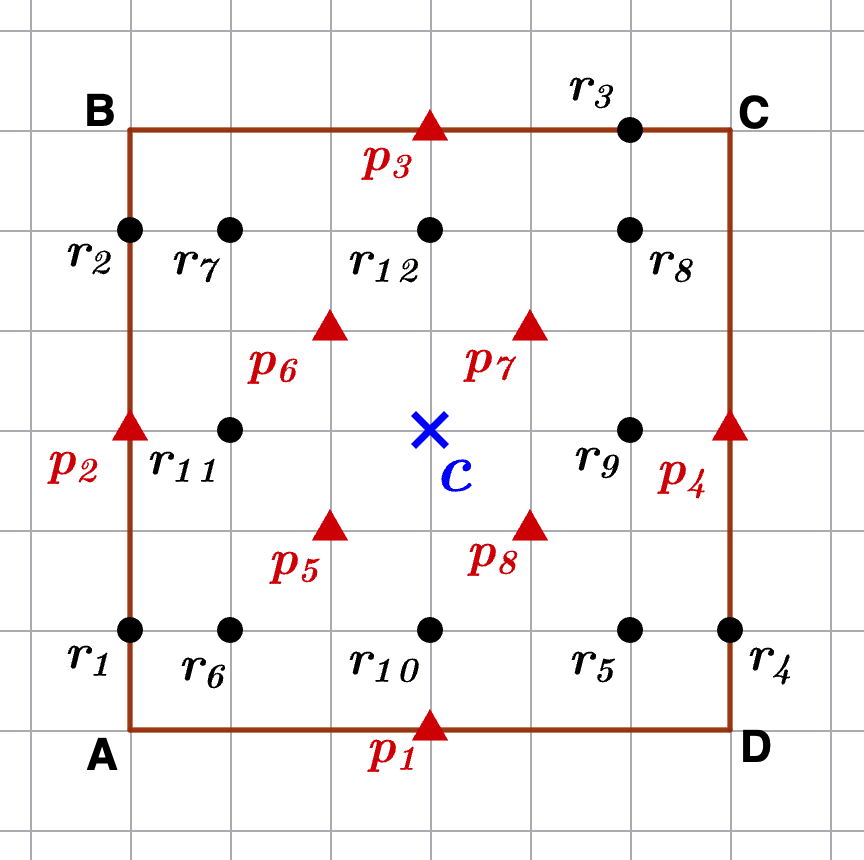}
        \caption{$\mathcal C_{322}$: empty $c$.}
        \label{fig:c322}
    \end{subfigure}
    \hspace{0.03\textwidth}
    \begin{subfigure}[b]{0.3\textwidth}
        \centering
        \includegraphics[width=\textwidth]{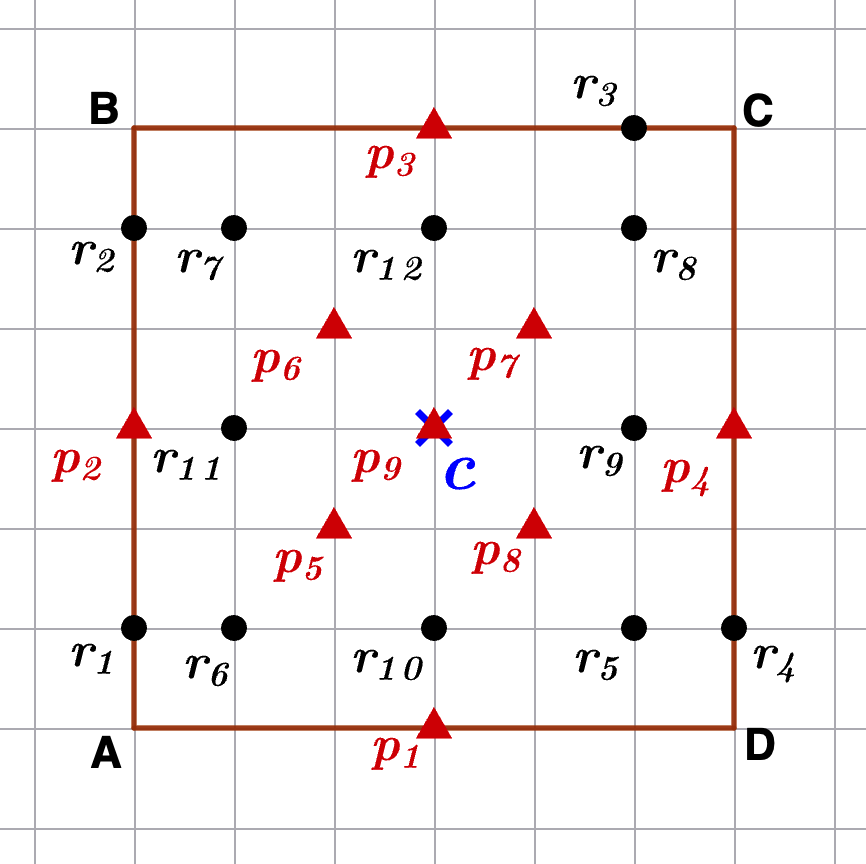}
        \caption{$\mathcal C_{323}$: parking node at $c$.}
        \label{fig:c323}
    \end{subfigure}

    \caption{Representative configurations for the subcases of $\mathcal C_{32}$.}
    \label{fig:sym2}
\end{figure}
   The capacities of the parking nodes are part of the input and are known to all robots. This information is required for the robots to compute the configuration view correctly. Without this information, robots may be unable to distinguish between different configuration classes. The notation frequently used in the description of the model, the algorithm, and the correctness analysis is summarized in Table~\ref{tab:notation}.

\begin{longtable}{p{0.22\linewidth}p{0.70\linewidth}}
\toprule
\textbf{Notation} & \textbf{Description} \\
\midrule
\endfirsthead

\toprule
\textbf{Notation} & \textbf{Description} \\
\midrule
\endhead

\midrule
\multicolumn{2}{r}{\textit{Continued on the next page}}\\
\endfoot

\bottomrule\\
\caption{Summary of notation used in the paper.}
\label{tab:notation}\\
\endlastfoot

$G=(V,E)$ & Infinite square grid graph, where $V$ is the set of grid nodes and $E$ is the set of grid edges. \\

$\mathscr P_\infty=(\mathbb Z,E')$ & Infinite path graph on the integer set $\mathbb Z$, used to define the square grid as $\mathscr P_\infty\times \mathscr P_\infty$. \\

$\mathcal R=\{r_1,r_2,\ldots,r_n\}$ & Set of $n$ mobile robots. \\

$r_i(t)$ & Node occupied by robot $r_i$ at time $t$. \\

$\mathcal R(t)$ & Multiset of robot positions at time $t$. \\

$\mathcal P=\{p_1,p_2,\ldots,p_m\}$ & Finite set of designated parking nodes. \\

$\kappa(p)$ & Prescribed capacity of parking node $p\in\mathcal P$. \\

$\lambda_t(v)$ & Number of robots located at node $v$ at time $t$. \\

$\Theta_t(v)$ & Status pair of node $v$ at time $t$, defined by $\Theta_t(v)=(\lambda_t(v),\kappa(v))$. \\

$\mathcal C(t)$ & Configuration of the system at time $t$, $\mathcal C(t)=(\mathcal R(t),\mathcal P,\kappa)$. \\

$d_{\mathcal M}(u,v)$ & Manhattan distance between two grid nodes $u,v\in V$. \\

$\mathcal{MER}_{\mathcal C}(t)$ & Minimum grid-aligned rectangle containing all occupied robot nodes and all parking nodes. \\

$\mathcal{MER}_{\mathcal P}$ & Minimum grid-aligned rectangle containing only the parking-node set $\mathcal P$. \\

$\operatorname{Aut}(\mathcal C(t))$ & Set of all automorphisms of the configuration $\mathcal C(t)$. \\

$\mathcal L$ & Unique line of reflectional symmetry, when it exists. \\

$c$ & Center of rotational symmetry, when it exists. \\

$\mu$ & Reflection map with respect to $\mathcal L$. \\

$\rho$ & Rotational map about the center $c$. \\

$q$ & Order of rotational symmetry. \\

$\operatorname{Orbit}_x$ & Rotational orbit of an entity $x$. \\

$\operatorname{View}_{\mathcal C(t)}(v)$ & Configuration view of node $v$ with respect to the key corner and scan direction. \\

% $\operatorname{ind}_{Q^*}(v)$ & Position of node $v$ in the canonical scan order starting from the key corner $Q^*$. \\

$\mathcal{STR}^Q(t)$ & String representation associated with corner $Q$. \\

$Q^*$ & Key corner of $\mathcal{MER}_{\mathcal C}(t)$. \\

$Q_{\mathcal P}$ & Leading corner of $\mathcal{MER}_{\mathcal P}$. \\

$\rho_Q$ & Parking string associated with corner $Q$ of $\mathcal{MER}_{\mathcal P}$. \\

$\pi(v)$ & Parking-capacity indicator used for parking strings, defined by $\pi(v)=\kappa(v)$. \\

$s$ & Number of surplus robots, $s=n-\sum_{i=1}^{m}\kappa(p_i)$. \\

$\mathcal R_s(t)$ & Set of surplus robots at time $t$. \\

$\mathcal R_u(t)$ & Set of unsaturated robots at time $t$. \\

$\mathcal {P^S}$ & Set of saturated parking nodes. \\

$\mathcal {P^U}$ & Set of unsaturated parking nodes. \\

$\mathscr G$ & Gathering node for surplus robots. \\

$\mathbf N_{\mathcal L}$ & Multiplicity or gathering node created on the formation line in reflectional cases. \\

$\mathbf N_c$ & Multiplicity or gathering node associated with the center of rotational symmetry. \\

$\Delta_i$ & Distance value used for selecting surplus robots. \\

% $\Delta_{\min}$ & Minimum value among the corresponding $\Delta_i$ values. \\

% $D_i$ & Distance value used for ordering robots or robot orbits. \\

% $D_{\min}$ & Minimum distance value among candidate robots or robot orbits. \\

% $d_{\max}$ & Maximum distance value used to select a farthest parking-node pair or orbit. \\

\end{longtable}

\FloatBarrier

\section{Impossibility Results}
\label{sec:impossibility}

In this section, we identify the configurations for which the surplus parking
gathering problem cannot be solved under the considered model. The
$\mathcal{SPG}$ problem consists of two requirements: saturating all parking
nodes according to their prescribed capacities and gathering all surplus robots
at a uniquely identifiable node. Hence, if either exact parking or surplus
gathering is impossible, then $\mathcal{SPG}$ is also impossible. Our impossibility results are based on known impossibility arguments for
grid-based mobile robots. In particular, we use the parking impossibility
results of Chakraborty and Mukhopadhyaya~\cite{chakraborty2025parking}, the
symmetry-based gathering impossibility of D'Angelo et
al.~\cite{d2016gathering}, and the partitive-configuration argument for
gathering over meeting nodes by Bhagat et al.~\cite{bhagat2022gathering}.
These results are adapted as necessary conditions for solving $\mathcal{SPG}$.
Accordingly, we separate the impossibility results into two parts: those arising
from the parking requirement and those arising from the surplus gathering
requirement.

\subsection{Impossibility of the Parking Requirement}

We first state some impossibility results that arise from the parking requirement of $\mathcal {SPG}$. Since $\mathcal {SPG}$ requires every parking node to be saturated exactly according to its prescribed capacity, any violation of the parking structure makes the entire problem unsolvable. The following results are based on the impossibility arguments of Chakraborty and Mukhopadhyaya~\cite{chakraborty2025parking}.

\begin{theorem}[Chakraborty and Mukhopadhyaya~\cite{chakraborty2025parking}]
\label{thm:non-designated-multiplicity}
Let $\mathcal{A}$ be a deterministic algorithm for $\mathcal{SPG}$ on an infinite grid. If, during some execution of $\mathcal{A}$, a multiplicity node is created at a node other than a parking node or the eventual gathering node, then $\mathcal{A}$ cannot guarantee a correct solution to $\mathcal{SPG}$ under \textsc{async} scheduler.
\end{theorem}

The proof follows from the indistinguishability argument of Chakraborty and Mukhopadhyaya~\cite{chakraborty2025parking}. Assume, for contradiction, that there exists a deterministic algorithm $\mathcal{A}$ that correctly solves $\mathcal{SPG}$. Suppose that, during some execution of $\mathcal{A}$, a multiplicity is created at a node that is neither a parking node nor the eventual gathering node. Since impossibility under the semi-synchronous (\textsc{ssync}) model also implies impossibility under the asynchronous (\textsc{async}) model, it is sufficient to establish the result under the \textsc{ssync} model. Consider an adversarial \textsc{ssync} scheduler. All robots occupying the multiplicity node have identical views and execute the same deterministic algorithm. Therefore, whenever the scheduler activates these robots simultaneously, they compute identical destinations and perform identical movements. Consequently, the multiplicity at the non-designated node is preserved throughout the execution. However, in any correct execution of $\mathcal{SPG}$, multiplicity can occur only at parking nodes and the eventual gathering node. Hence, the persistent multiplicity at the non-designated node prevents the algorithm from reaching a valid final configuration. This contradicts the assumption that $\mathcal{A}$ correctly solves $\mathcal{SPG}$. Therefore, no such deterministic algorithm exists.

\begin{theorem}[Chakraborty and Mukhopadhyaya~\cite{chakraborty2025parking}]
\label{thm:strong-md-necessary}
Without strong multiplicity detection, the surplus parking gathering problem is unsolvable under \textsc{async} scheduler.
\end{theorem}

The proof follows from the necessity of strong multiplicity detection established by Chakraborty and Mukhopadhyaya~\cite{chakraborty2025parking}. In $\mathcal{SPG}$, each parking node $p_i$ is associated with a prescribed capacity $\kappa_i$, and the robots must determine whether exactly $\kappa_i$ robots occupy $p_i$. Without strong multiplicity detection, the robots cannot determine the exact multiplicity at a parking node and, consequently, cannot verify whether the required capacity has been achieved.

Furthermore, the proposed algorithm relies on multiplicity nodes to identify different execution phases and to determine the gathering of the surplus robots. Without strong multiplicity detection, the robots cannot distinguish multiplicity nodes from ordinary occupied nodes or determine the exact number of robots gathered at a node. Consequently, they cannot verify whether the desired terminal configuration has been reached. Therefore, no deterministic algorithm can correctly solve $\mathcal{SPG}$ without strong multiplicity detection.

\begin{theorem}[Chakraborty and Mukhopadhyaya~\cite{chakraborty2025parking}]
\label{thm:c223-impossible}
Let $\mathcal C(t_0)\in\mathcal C_{223}$, and let $p_l$ be a parking node located on the line of symmetry $\mathcal L$. If the capacity of $p_l$ is an odd integer, then $\mathcal {SPG}$ is unsolvable under \textsc{async} scheduler.
\end{theorem}

The proof follows from Lemma~3 of Chakraborty and Mukhopadhyaya~\cite{chakraborty2025parking}. Assume, for contradiction, that there exists a deterministic algorithm $\mathcal A$ that correctly solves $\mathcal{SPG}$. Since $\mathcal C(t_0)\in\mathcal C_{223}$, the configuration is symmetric with respect to the line $\mathcal L$, no robot is initially located on $\mathcal L$, and at least one parking node lies on $\mathcal L$. Let $p_l$ be a parking node on $\mathcal L$ with odd capacity $\kappa(p_l)=2k+1$. Since impossibility under the semi-synchronous (\textsc{ssync}) model also implies impossibility under the asynchronous (\textsc{async}) model, it is sufficient to consider an adversarial \textsc{ssync} scheduler. The scheduler activates every robot simultaneously with its mirror image with respect to $\mathcal L$. Because each pair of symmetric robots has identical local views and executes the same deterministic algorithm, both robots compute identical destinations and perform identical movements. Consequently, robots can reach the parking node $p_l$ only in symmetric pairs. Hence, the number of robots occupying $p_l$ is always even. Since the capacity of $p_l$ is odd, namely $2k+1$, the parking node can never be saturated exactly. Therefore, the parking requirement of $\mathcal{SPG}$ cannot be satisfied, contradicting the assumption that $\mathcal A$ correctly solves $\mathcal{SPG}$. Hence, $\mathcal{SPG}$ is unsolvable for such configurations.

\begin{theorem}[Chakraborty and Mukhopadhyaya~\cite{chakraborty2025parking}]
\label{thm:c323-impossible}
Let $\mathcal C(t_0)\in\mathcal C_{323}$, and let $p_c$ be the parking node located at the center of rotational symmetry $c$. If the capacity of $p_c$ is not divisible by the order of the rotational symmetry, then $\mathcal {SPG}$ is unsolvable under \textsc{async} scheduler.
\end{theorem}

The proof follows from the rotational-symmetry impossibility result of Chakraborty and Mukhopadhyaya~\cite{chakraborty2025parking}. Assume, for contradiction, that there exists a deterministic algorithm $\mathcal A$ that correctly solves $\mathcal{SPG}$. Since $\mathcal C(t_0)\in\mathcal C_{323}$, the parking-node configuration is rotationally symmetric, and no robot is initially located at the center of symmetry $c$. Since impossibility under the semi-synchronous (\textsc{ssync}) model also implies impossibility under the asynchronous (\textsc{async}) model, it is sufficient to consider an adversarial \textsc{ssync} scheduler that preserves the rotational symmetry of the execution.

Because the robots belonging to the same rotational orbit have identical local views and execute the same deterministic algorithm, they compute identical destinations and perform identical movements. Consequently, robots can reach the parking node at the center only in complete rotational orbits. Hence, the number of robots that can occupy the parking node at $c$ must always be a multiple of the order of rotational symmetry.

If the prescribed capacity of $p_c$ is not divisible by the order of the rotational symmetry, then the parking node can never be saturated exactly. Therefore, the parking requirement of $\mathcal{SPG}$ cannot be satisfied, contradicting the assumption that $\mathcal A$ correctly solves $\mathcal{SPG}$. Hence, $\mathcal{SPG}$ is unsolvable for such configurations.

\subsection{Impossibility of the Surplus Gathering Requirement}

We now state the impossibility results that arise from the surplus gathering requirement of $\mathcal {SPG}$. Since $\mathcal {SPG}$ requires both the exact saturation of all parking nodes and the gathering of all surplus robots at a single node, the problem becomes unsolvable whenever the surplus gathering subproblem cannot be solved deterministically. The following result is based on the partitive-configuration argument for gathering over meeting nodes.

\begin{figure}[htbp]
    \centering

    \begin{subfigure}[t]{0.45\textwidth}
        \centering
        \includegraphics[width=\textwidth]{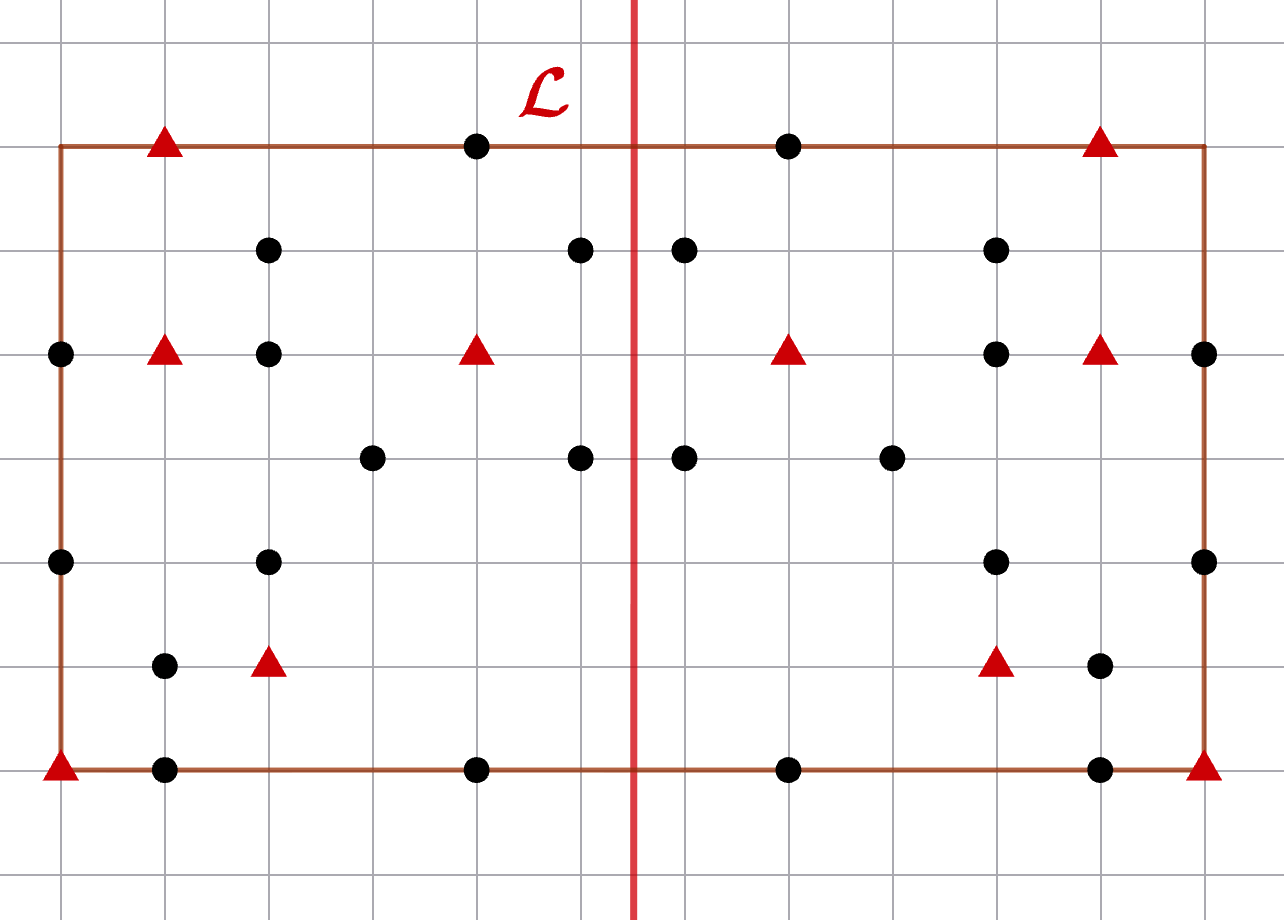}
        \caption{}
        
    \end{subfigure}
   \hspace{0.04\textwidth}%
    \begin{subfigure}[t]{0.32\textwidth}
        \centering
        \includegraphics[width=\textwidth]{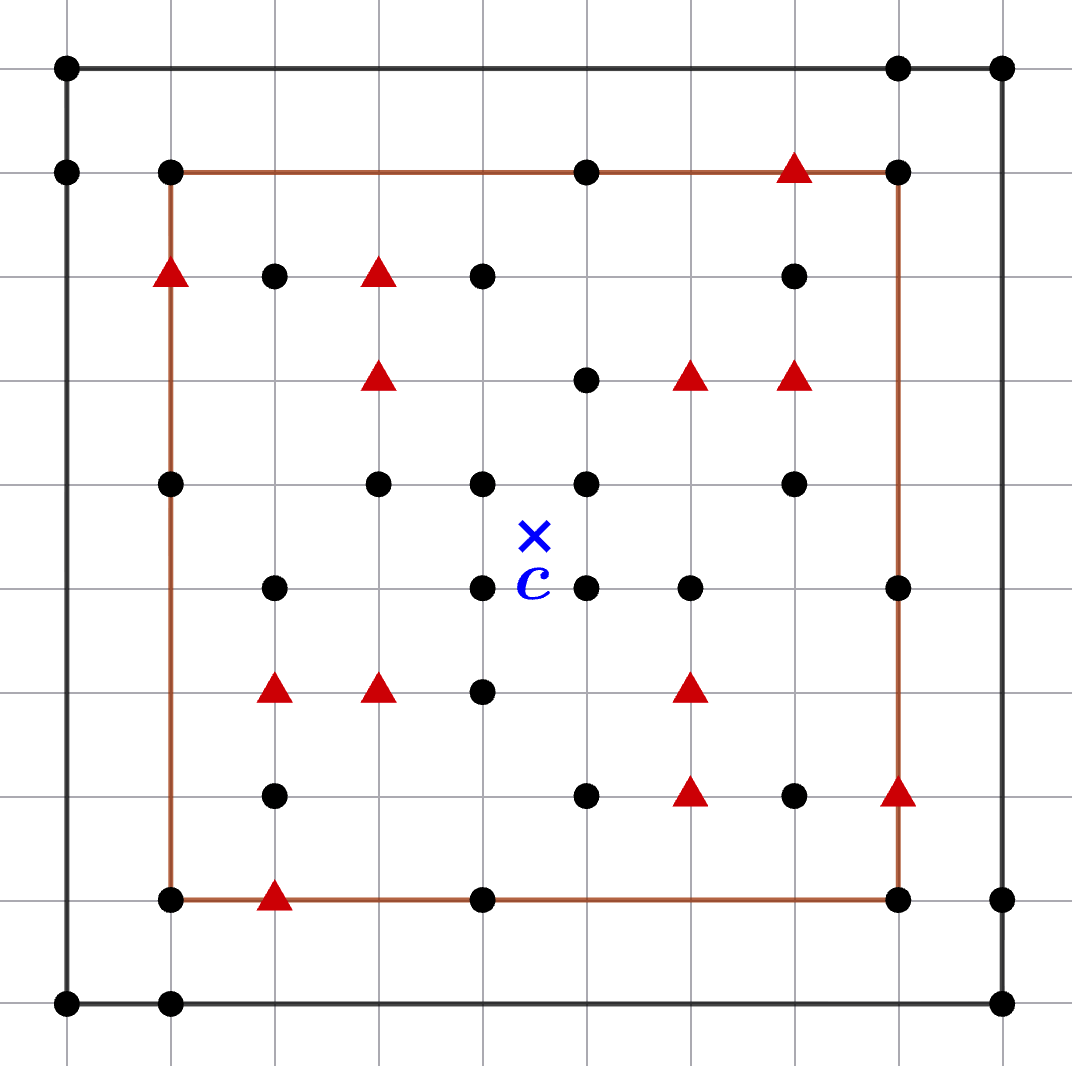}
        \caption{}
        
    \end{subfigure}

   \caption{Illustrative examples of partitive configurations: (a) with respect to the line of symmetry $\mathcal L$; (b) with respect to the rotational symmetry centered at $c$.}
    \label{partitive}
\end{figure}

\begin{theorem}[Bhagat et al.~\cite{bhagat2022gathering}]
\label{thm:partitive-gathering-impossible}
Let $\mathcal C(t_0)$ be a symmetric initial configuration, and let $V_0\subseteq V$ be a set of nodes fixed by an automorphism $\phi$ such that $V_0\cap \mathcal R(t_0)=\emptyset$. Suppose that $\phi$ is partitive on $V\setminus V_0$. If no admissible gathering node of $\mathcal {SPG}$ belongs to $V_0$, then no deterministic algorithm can guarantee the surplus gathering requirement of $\mathcal {SPG}$ under \textsc{async} scheduler.
\end{theorem}

The argument follows from the partitive-symmetry impossibility result of Bhagat et al.~\cite{bhagat2022gathering}. Under a symmetric execution, robots belonging to the same orbit have identical views and execute the same deterministic algorithm. An adversarial scheduler can preserve the orbit structure throughout the execution. Hence, the robots cannot deterministically select a unique gathering node in $V\setminus V_0$. Since no admissible gathering node lies in $V_0$, the surplus gathering requirement cannot be guaranteed. Consequently, $\mathcal {SPG}$ is unsolvable for such configurations.

\begin{corollary}
\label{cor:line-gathering-impossible}
Let $\mathcal C(t_0)$ admit a unique line of symmetry $\mathcal L$. If $\mathcal L$ contains neither a robot nor an admissible gathering node, and the configuration is partitive with respect to $\mathcal L$, then $\mathcal {SPG}$ is unsolvable under \textsc{async} scheduler.
\end{corollary}

Indeed, by taking $V_0$ as the set of nodes on $\mathcal L$, the configuration remains partitive on $V\setminus V_0$ under a symmetric execution. Therefore, gathering can be guaranteed only at a node on $\mathcal L$. Since $\mathcal L$ contains no admissible gathering node, the surplus gathering requirement of $\mathcal {SPG}$ cannot be satisfied (see Figure \ref{partitive} (a)).

\begin{corollary}
\label{cor:rotational-gathering-impossible}
Let $\mathcal C(t_0)$ admit rotational symmetry with center $c$, where $c$ is a grid node. If $c$ contains neither a robot nor an admissible gathering node, and the configuration is partitive with respect to $c$, then $\mathcal {SPG}$ is unsolvable under \textsc{async} scheduler.
\end{corollary}

The proof is analogous to Corollary~\ref{cor:line-gathering-impossible}. By taking $V_0=\{c\}$, every symmetry-preserving execution keeps the robots partitioned into rotational orbits. Thus, a unique gathering node outside $c$ cannot be deterministically selected. Since $c$ is not an admissible gathering node, the surplus gathering requirement cannot be achieved (see Figure \ref{partitive} (b)).

\section{Algorithm \textbf{\textsc{spg()}}} \label{sec:algorithm}
\subsection{High-Level Idea}
Depending on the type of the initial configuration, the algorithm \textbf{\textsc{spg()}} executes different sequences of phases. For configurations $\mathcal C_{21}$ and $\mathcal C_{31}$, it proceeds through four phases: \textit{Line Formation}, \textit{Multiplicity Creation}, \textit{Saturation}, and \textit{Gathering}. For the remaining configurations, the algorithm works in only two phases: \textit{Saturation} and \textit{Gathering}. During the \textit{Line Formation} phase, the robots arrange themselves on a line on distinct nodes in a manner that implicitly identifies two groups: parking robots, which will eventually occupy the parking nodes, and surplus robots, which will ultimately gather at a common node. The \textit{Multiplicity Creation} phase follows the line formation phase and creates a unique multiplicity node on the line. Due to the presence of this unique multiplicity, the robots can detect the transition between the line formation phase and the saturation phase. In the \textit{Saturation} phase, the robots are assigned to the parking nodes according to the ordering defined for the corresponding configuration. Each robot moves sequentially in a collision-free manner, ensuring that every parking node reaches its prescribed capacity while maintaining the invariants established in the earlier phases. After all parking nodes have been saturated, every robot that does not occupy a parking node is defined as a surplus robot. Finally, in the \textit{Gathering} phase, the remaining robots are treated as surplus robots and move to the designated gathering node. An exception occurs for configurations $\mathcal C_{222}$ and $\mathcal C_{322}$, where the \textit{Gathering} phase is executed before the \textit{Saturation} phase.

\paragraph{\textbf{Pending Move Analysis:}}
In this section, we analyze the effect of a \emph{pending move} on the configuration 
$\mathcal C(t)$ during the execution of the algorithm \textbf{\textsc{spg()}}. 
If $\mathcal C(t)$ is asymmetric, then the ordering of robots remains unaffected even in 
the presence of a pending move. Hence, pending moves require special attention only for 
symmetric configurations.

Pending moves play two different roles depending on the configuration class. In 
$\mathcal C_{221}$ and $\mathcal C_{321}$, the algorithm intentionally breaks the symmetry 
by moving a uniquely selected robot away from the line of symmetry or from the center of 
rotation. In this case, \textit{AllowtoMove()} is used as a safety test before initiating 
the symmetry-breaking move. In contrast, in $\mathcal C_{222}$ and $\mathcal C_{322}$, 
the algorithm does not break the symmetry. Instead, the selected robots move in symmetric 
pairs or rotational orbits, and pending moves are used only to detect incomplete symmetric 
movements.

\begin{definition}
For any fixed time $t\geq 0$, let
$
\beta_t=\{v\in V\mid v \text{ is occupied by a robot at time }t\}.
$
Define the occupancy function $\alpha_t:V\rightarrow\{0,1\}$ by
$
\alpha_t(v)=\mathbf{1}_{\beta_t}(v),
$
where $\mathbf{1}_{\beta_t}$ denotes the indicator function on the set $\beta_t$.
\end{definition}

\begin{figure}[htbp]
    \centering

    \begin{subfigure}[t]{0.50\textwidth}
        \centering
        \includegraphics[width=\textwidth]{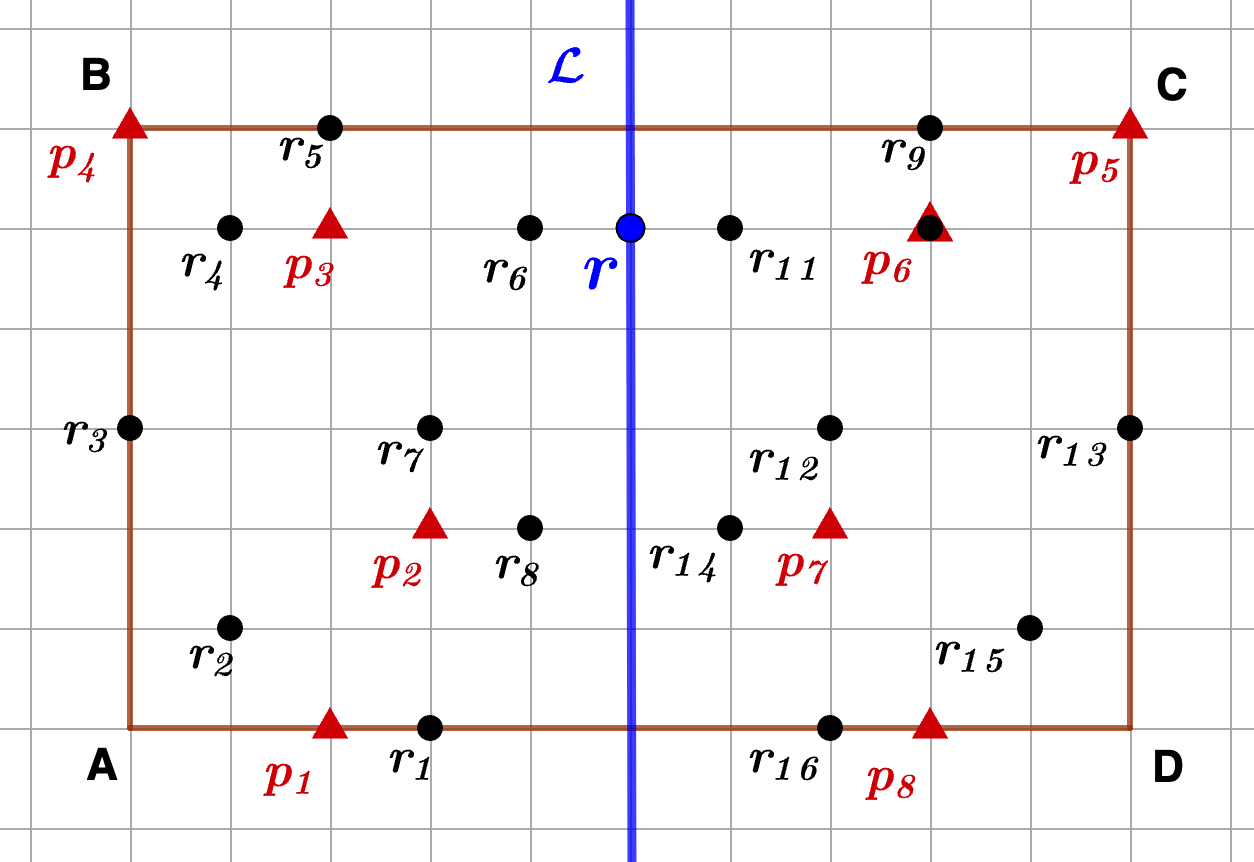}
        \caption{}
        \label{fig:near-rot}
    \end{subfigure}
   \hspace{0.05\textwidth}
    \begin{subfigure}[t]{0.36\textwidth}
        \centering
        \includegraphics[width=\textwidth]{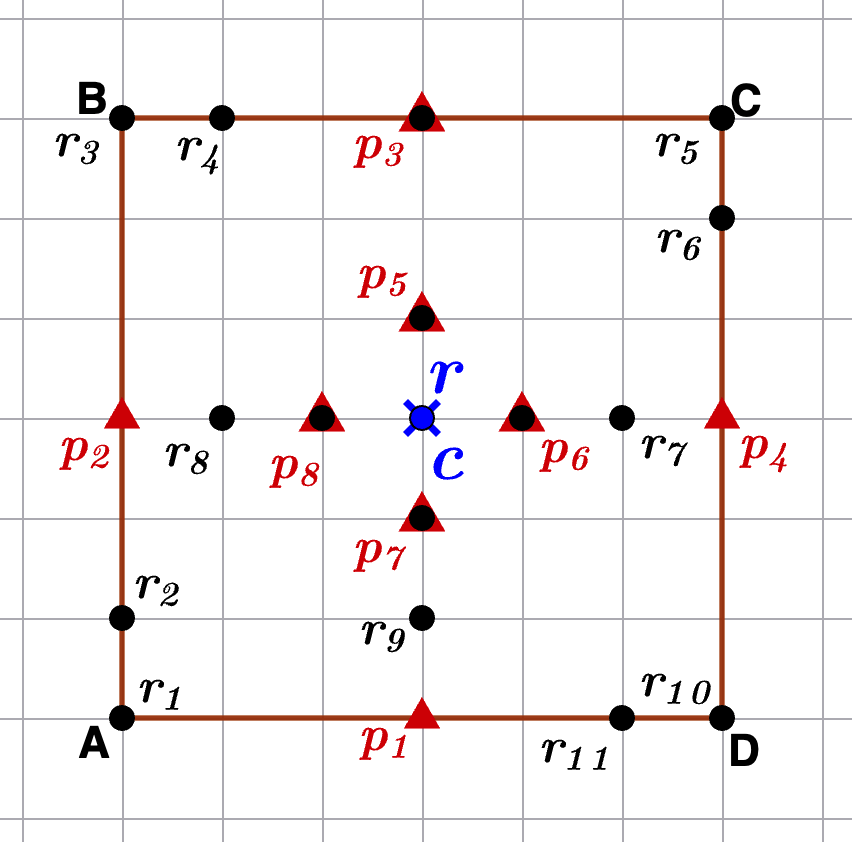}
        \caption{}
        \label{fig:near-refl}
    \end{subfigure}

\caption{
Examples of \textsc{nearly reflective} and \textsc{nearly rotational} configurations. 
(a) A \textsc{nearly reflective} configuration, where the strings generated by the robot $r$ 
in the directions away from $\mathcal L$ are $\{1010,1001\}$. 
(b) A \textsc{nearly rotational} configuration, where the strings generated by the robot $r$ 
are $\{110,110,110,101\}$. The parking nodes $p_3$, $p_5$, $p_6$, $p_7$, and $p_8$ are occupied by robots.
}
    \label{strings}
\end{figure}

Two binary strings are said to be \textsc{nearly equal} if one can be transformed into the 
other by replacing exactly one occurrence of the substring $01$ by $10$, or conversely. 
Such a local change represents the effect of a pending move of a robot from an occupied 
node to an adjacent free node along the corresponding half-line.

\subparagraph{\textbf{Pending moves in $\mathcal C_{221}$ and $\mathcal C_{321}$ (symmetry-breaking case):}}
We first consider the configurations in which the algorithm intentionally breaks the 
symmetry by moving a uniquely selected robot.

Let $\mathcal C(t)\in\mathcal C_{221}$ be symmetric with respect to a unique line of 
symmetry $\mathcal L$. Let $r$ be the unique robot on $\mathcal L$ having the maximum 
configuration view according to the ordering $\mathscr O_3$. Since $r$ lies on 
$\mathcal L$, the two relevant strings are generated from $r$ in the two directions away 
from $\mathcal L$. These strings are denoted by
$
\mathcal{STR}^{\text{left}}(r)
\quad\text{and}\quad
\mathcal{STR}^{\text{right}}(r).
$
The configuration is said to be \textsc{nearly reflective} if these two strings are 
\textsc{nearly equal} (see Figure~\ref{strings}(a)). In this case, a pending move has already affected one side of 
$\mathcal L$, and therefore the robot $r$ is not allowed to initiate a new 
symmetry-breaking move.

Next, let $\mathcal C(t)\in\mathcal C_{321}$, where the center of rotation $c$ is occupied 
by a robot $r$. In this case, the strings are generated with respect to the robot $r$ at 
the center $c$. From $r$, consider the four half-lines in the left, right, upward, and 
downward directions. The corresponding strings are denoted by
$
\mathcal{STR}^{\text{left}}(r),
\mathcal{STR}^{\text{right}}(r),
\mathcal{STR}^{\text{upward}}(r),
\mathcal{STR}^{\text{downward}}(r).
$
These strings are compared according to the rotational symmetry around $c$. If the 
rotational order is $2$, then the opposite directions are compared in above. If the rotational 
order is $4$, then all four directions are compared cyclically. The configuration is said 
to be \textsc{nearly rotational} if the strings are identical except for a local 
$01\leftrightarrow 10$ change caused by a pending move in one or more directions. (see Figure~\ref{strings}(b))

Before allowing the unique robot on $\mathcal L\cup\{c\}$ to move to an adjacent free 
node, each robot executes \textit{AllowtoMove()}. If the current configuration is detected 
as \textsc{nearly reflective} or \textsc{nearly rotational}, then the procedure returns 
false, and no new symmetry-breaking move is started. Otherwise, if a suitable adjacent 
free node exists, the unique robot is allowed to move away from $\mathcal L\cup\{c\}$, 
and the symmetry is broken.

\subparagraph{\textbf{Pending moves in $\mathcal C_{222}$ and $\mathcal C_{322}$ (symmetry-preserving case):}}

We now consider the configurations in which the algorithm preserves the symmetry during 
movement.

Let $\mathcal C(t)\in\mathcal C_{222}$. In this case, the line of symmetry $\mathcal L$ 
contains no robot. Hence, robots are selected in symmetric pairs with respect to 
$\mathcal L$. Let $(r,\mu(r))$ be such a selected pair, where $\mu$ denotes the reflection 
map with respect to $\mathcal L$. Let $\ell(r,\mu(r))$ be the line passing through 
$r$ and $\mu(r)$, and let
$
x=\ell(r,\mu(r))\cap \mathcal L
$
be its intersection node with the line of symmetry.

The string representation is now defined with respect to the node $x$. From
$x$, we consider the two opposite half-lines along $\ell(r,\mu(r))$, one toward
$r$ and the other toward $\mu(r)$. Along each half-line, we record the robot
occupancy value $\lambda_t(v)$ of every visited node $v$, starting from the node
adjacent to $x$ and continuing up to the farthest occupied node in that
direction. Let the two resulting strings be denoted by $\mathcal{STR}^{+}(x)$ and $\mathcal{STR}^{-}(x)$. If $\mathcal{STR}^{+}(x)$ and $\mathcal{STR}^{-}(x)$ are \textsc{nearly equal}, then the 
configuration represents a pending reflective movement. This means that one robot of the 
selected symmetric pair has already moved toward its destination, while the other robot 
has not yet completed the corresponding symmetric move. In this case, the robot that has 
already progressed remains stationary, and the pending robot continues its movement so 
that the reflectional symmetry is restored. No new symmetric pair is selected until this 
pending move is completed.

Next, let $\mathcal C(t)\in\mathcal C_{322}$. Here, the configuration is rotationally 
symmetric with center $c$, and the center is not occupied by a robot. Robots are selected 
in rotational orbits with respect to $c$. Let $\rho$ be the rotational map about $c$. Mote that the rotational order is $q\in\{2,4\}$. Suppose that the selected orbit is
$
\mathbf{Orbit}_r=\{r,\rho(r),\rho^2(r),\ldots,\rho^{q-1}(r)\}.
$
The string representation is defined with respect to the center $c$. From $c$, the 
occupancy values are recorded along the directions determined by the robots in the 
selected orbit and their rotated copies.

For each $i\in\{0,1,\ldots,q-1\}$, define
\[
b_i=
\begin{cases}
1, & \text{if $\rho^i(r)$ has already progressed toward its assigned destination,}\\
0, & \text{otherwise.}
\end{cases}
\]
The binary string
$
B(\mathbf{Orbit}_r)=b_0b_1\cdots b_{q-1}
$
is called the movement-status string of the selected rotational orbit.

If $q=2$, then the possible movement-status strings are
$
00, 01, 10, 11.
$
The strings $00$ and $11$ correspond to symmetric states before and after the completion 
of the orbit movement, respectively. The strings $01$ and $10$ correspond to pending 
rotational movements, since exactly one robot of the selected orbit has progressed while 
the other has not.

If $q=4$, then the movement-status string has length four. The strings $0000$ and $1111$ 
correspond to symmetric states before and after the completion of the orbit movement. 
Every other string, for example $1000, 0100, 1010, 1100, \text{ and, } 1110$ represents a pending rotational movement, since some but not all robots of the selected orbit have progressed toward their assigned destinations.

Thus, in $\mathcal C_{222}$ and $\mathcal C_{322}$, pending moves are not used to break 
symmetry. They only indicate that a symmetry-preserving movement has been partially 
completed. During such a pending state, the robots that have already progressed remain 
stationary, while the remaining robots of the same symmetric pair or rotational orbit 
continue toward their corresponding destinations. No new pair or orbit is allowed to move 
until the current pending movement is completed and the corresponding symmetry is restored. The pseudocode for maintaining symmetry during pending move is presented in Algorithm~\ref{alg:handle-pending-move}.

\begin{algorithm}[ht]
\tiny
\caption{\textsc{: HandlePendingMove}$(\mathcal C(t))$}
\label{alg:handle-pending-move}
\begin{algorithmic}[1]
\Require $\mathcal C(t)\in\{\mathcal C_{222},\mathcal C_{322}\}$
\Ensure Pending symmetry-preserving move is completed

\If{$\mathcal C(t)\in\mathcal C_{222}$}
    \State Let $(r,\mu(r))$ be the selected symmetric pair
    \State Let $x=\ell(r,\mu(r))\cap \mathcal L$
    \State Construct $\mathcal{STR}^{+}(x)$ and $\mathcal{STR}^{-}(x)$ along the two opposite half-lines from $x$
    \If{$\mathcal{STR}^{+}(x)$ and $\mathcal{STR}^{-}(x)$ are nearly equal}
        \State Move only the pending robot of $(r,\mu(r))$
        \State Keep its symmetric partner stationary
    \Else
        \State Move $r$ and $\mu(r)$ symmetrically toward their destinations
    \EndIf
\EndIf

\If{$\mathcal C(t)\in\mathcal C_{322}$}
    \State Let $\mathbf{Orbit}_r=\{r,\rho(r),\ldots,\rho^{q-1}(r)\}$ be the selected orbit
    \State Construct the movement-status string $B(\mathbf{Orbit}_r)=b_0b_1\cdots b_{q-1}$
    \If{$B(\mathbf{Orbit}_r)=00$ or $0000$}
        \State Move all robots in $\mathbf{Orbit}_r$ symmetrically toward their destinations
    \ElsIf{$B(\mathbf{Orbit}_r)=11$ or $1111$}
        \State The orbit movement is complete
    \Else
        \State Move only the pending robots of $\mathbf{Orbit}_r$
        \State Keep the already-progressed robots stationary
    \EndIf
\EndIf

\State No new pair or orbit is selected until the symmetry is restored

\end{algorithmic}
\end{algorithm}

% \paragraph{Nearly rotational configuration.}
% Consider a configuration $\mathcal C(t)$ having rotational symmetry with center $c$ and rotational order $q\in\{2,4\}$. Let $\rho$ denote the rotational map about $c$. Among all rotational orbits of robots, let $\mathbf{Orbit}_r$ be the orbit whose robots have minimum Manhattan distance from $c$. Only the robots in $\mathbf{Orbit}^*$ are allowed to perform the prescribed symmetry-breaking move.

% Let
% \[
% \mathbf{Orbit}^*=\{r,\rho(r),\rho^2(r),\ldots,\rho^{q-1}(r)\}.
% \]
% For each $i\in\{0,1,\ldots,q-1\}$, define
% \[
% b_i=
% \begin{cases}
% 1, & \text{if $\rho^i(r)$ has already performed its prescribed move,}\\
% 0, & \text{otherwise.}
% \end{cases}
% \]
% The binary string
% \[
% B(\mathbf{Orbit}^*)=b_0b_1\cdots b_{q-1}
% \]
% is called the movement-status string of the selected rotational orbit.

% The configuration $\mathcal C(t)$ is said to be \textsc{nearly rotational} if the movement-status string contains at least one $1$ and at least one $0$. Equivalently, a nonempty proper subset of the robots in $\mathbf{Orbit}^*$ has already performed the prescribed move. Thus, the movement of the selected rotational orbit is partially completed.

\subsection{Ordering for Algorithmn \textbf{\textsc{spg()}}}
\label{Ordering}
In this section, we introduce different orderings of parking nodes and robot positions that are used throughout the algorithm. The different orderings are as follows:

\begin{itemize} 
    \item $\mathscr O_1:$ Consider the case where the set of parking nodes is asymmetric. By the definition of asymmetry for the set $\mathcal P$, there always exists a unique lexicographically maximum string $\mathcal{STR}^i$. Consequently, a unique leading corner, say, $\mathcal {K} _ {p} $ of $\mathcal{MER}_{\mathcal P}$ of the set $\mathcal {P} $ of the parking nodes can be determined. Now consider the string representation $\mathcal{STR}^i$ associated with $\mathcal K_p$. The parking nodes are ordered according to their order of appearance in this string, from first to last. Let us assume that $p_1 \prec_{\mathcal P} p_2 \prec_{\mathcal P} \cdots \prec_{\mathcal P} p_m$ be the sequence of parking nodes obtained in this way. This ordering is denoted by $\mathscr O_1$.

% \begin{figure}[htbp]
%     \centering
%     \begin{subfigure}[t]{0.58\textwidth}
%         \centering
%         \includegraphics[width=\textwidth]{I_21.png}
%         \caption{}
%         \label{fig:refl21}
%     \end{subfigure}
%     \hfill
%     \begin{subfigure}[t]{0.4\textwidth}
%         \centering
%         \includegraphics[width=\textwidth]{I_31.png}
%         \caption{}
%         \label{fig:rot31}
%     \end{subfigure}

%     \caption{}
%     \label{ordering O2}
% \end{figure}

  \item $\mathscr{O}_2$: Consider an asymmetric configuration $\mathcal{C}(t)$ in which the parking-node configuration is either reflective, with a unique line of symmetry $\mathcal{L}$, or rotationally symmetric with center $c$. Although the parking-node configuration itself is symmetric, the overall configuration $\mathcal{C}(t)$ is asymmetric. Consequently, the asymmetry of $\mathcal{C}(t)$ induces a total ordering on all entities (i.e., both the robots and the parking nodes). Whenever multiple robots or parking nodes share the same maximum Manhattan distance from $\mathcal{L}$ (in the reflective case) or from $c$ (in the rotational case), the tie is broken by selecting the entity that possesses the highest order under the induced total ordering. We denote this resulting ordering by $\mathscr{O}_2$. Since the entities in an asymmetric configuration are totally orderable~\cite{bose2020arbitrary}, the ordering $\mathscr{O}_2$ is structurally well-defined.

\begin{figure}[htbp]
    \centering
    \includegraphics[width=0.75\textwidth]{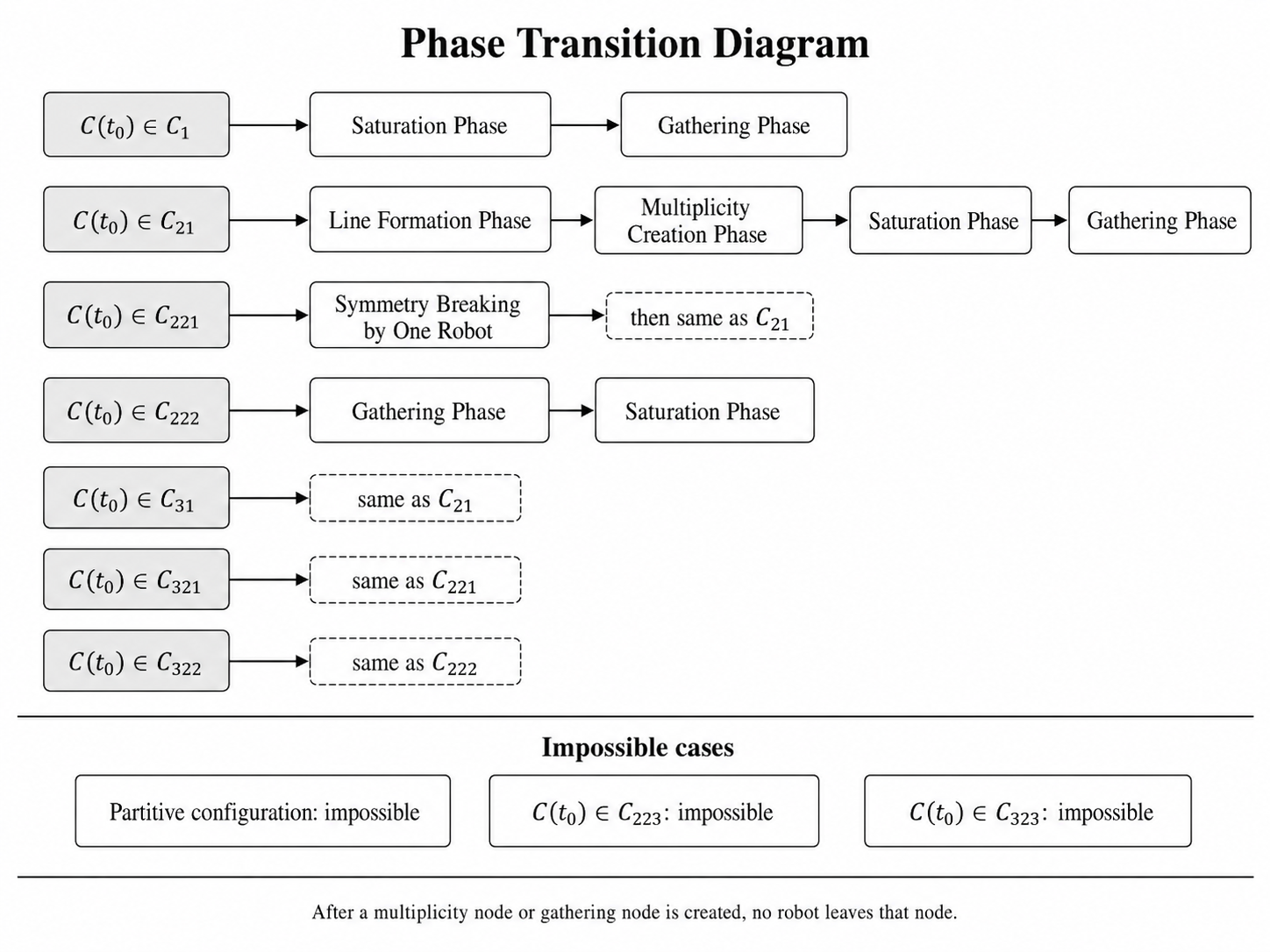}
    \caption{Phase transition diagram of the algorithm for different initial configuration classes.}
    \label{fig:phase-transition}
\end{figure}

  \item $\mathscr{O}_3$: Consider the case where a set of entities (robots or parking nodes) is symmetric with respect to a unique line of symmetry $\mathcal{L}$, and at least one entity lies on $\mathcal{L}$. In this scenario, there may exist either one or two key corners of the corresponding $\mathcal{MER}$ (Minimal Enclosing Rectangle). If there exists a unique key corner $\mathcal{K}$, then $\mathcal{L}$ must be a diagonal line of symmetry. The two strings associated with $\mathcal{K}$ are identical, and one of the two directions is chosen arbitrarily as the string direction. Let $\mathcal{STR}^i$ be the corresponding string representation. The entities lying on $\mathcal{L}$ are then ordered from first to last according to their appearance in $\mathcal{STR}^i$. This ordering is denoted by $\mathscr{O}_3$. 

%   \begin{figure}[htbp]
%     \centering

%     \begin{subfigure}[t]{0.55\textwidth}
%         \centering
%         \includegraphics[width=\textwidth]{I_221.png}
%        \caption{}
%         \label{fig:refl221}
%     \end{subfigure}
%     \hfill
%     \begin{subfigure}[t]{0.4\textwidth}
%         \centering
%         \includegraphics[width=\textwidth]{I_321.png}
%         \caption{}
%         \label{fig:rot321}
%     \end{subfigure}

%     \caption{}
%     \label{ordering O3}
% \end{figure}
    
    If there exist two key corners, the string representations associated with both key corners are evaluated, and $\mathscr{O}_3$ is defined in an identical manner. Since $\mathscr{O}_3$ depends solely on the positions of the entities and the key corners, it uniquely determines an ordering at every time instant. Consequently, when applied to robots, this ordering may change as they move during the execution; when applied to parking nodes, which are stationary, the ordering remains invariant.

% \begin{figure}[htbp]
%     \centering

%     \begin{subfigure}[t]{0.55\textwidth}
%         \centering
%         \includegraphics[width=\textwidth]{I_222.png}
%        \caption{}
%         \label{fig:refl222}
%     \end{subfigure}
%     \hfill
%     \begin{subfigure}[t]{0.4\textwidth}
%         \centering
%         \includegraphics[width=\textwidth]{I_322.png}
%         \caption{}
%         \label{fig:rot322}
%     \end{subfigure}

%     \caption{}
%     \label{ordering O4}
% \end{figure}

\item $\mathscr{O}_4$: Consider a configuration
$\mathcal{C}(t)\in\{\mathcal{C}_{222},\mathcal{C}_{322}\}$.

\begin{enumerate}

\item Suppose that $\mathcal{C}(t)\in\mathcal{C}_{222}$, and let
$\mathcal{L}$ be the line of symmetry. Define $ \mathcal S_{\mathcal L}
    =
    \mathcal L\cap\mathcal{MER}_{\mathcal P}$ to be the line segment of the symmetry line $\mathcal L$ contained in
$\mathcal{MER}_{\mathcal P}$. Let $ V_{\mathcal L}=\mathcal S_{\mathcal L}\cap V$. Thus, $V_{\mathcal L}$ denotes the set of grid nodes lying on the segment
$\mathcal S_{\mathcal L}$. Let $Q_{\mathcal P}$ be the leading corner of $\mathcal{MER}_{\mathcal P}$.
The string direction determined from $Q_{\mathcal P}$ induces a total order
on the nodes of $V_{\mathcal L}$. We denote this order by
$\prec_{\mathcal L}$. Hence, for any two nodes
$u,v\in V_{\mathcal L}$, either $u\prec_{\mathcal L} v$ or
$v\prec_{\mathcal L} u$ holds. This gives a unique ordering of the nodes of
$V_{\mathcal L}$ along the symmetry line $\mathcal L$.

Let $\mu$ denote the image with respect to $\mathcal L$. Since
$\mathcal{C}(t)$ is symmetric with respect to $\mathcal L$, the robot set
$\mathcal R$ can be partitioned into symmetric pairs $ \Pi=\{(r,\mu(r)): r\in\mathcal R\}$.
For each symmetric pair $\pi=(r,\mu(r))\in\Pi$, define $d_{\mathcal L}(\pi)
    =
    \min_{g\in V_{\mathcal L}} d_{\mathcal M}(r,g)$. Since every node $g\in V_{\mathcal L}$ lies on the symmetry line
$\mathcal L$, we have $d_{\mathcal M}(r,g)=d_{\mathcal M}(\mu(r),g)$. Therefore, $d_{\mathcal L}(\pi)$ is also equal to $ \min_{g\in V_{\mathcal L}} d_{\mathcal M}(\mu(r),g)$. Among the nodes of $V_{\mathcal L}$ attaining the minimum distance
$d_{\mathcal L}(\pi)$, let $g_{\mathcal L}(\pi)$ be the first node according
to the order $\prec_{\mathcal L}$. In other words,
$g_{\mathcal L}(\pi)$ is the first node in the leading-corner string order
among all nodes of $V_{\mathcal L}$ that are closest to the symmetric pair
$\pi$.

Now let $Q_{\mathcal C}$ be the key corner of the current configuration
$\mathcal C(t)$. The string direction determined from $Q_{\mathcal C}$
induces an order on the symmetric robot pairs in $\Pi$. We denote this
pair-order by $\prec_{\mathcal C}$.

The ordering $\mathscr O_4$ is now defined as follows. The symmetric pairs
in $\Pi$ are ordered in non-decreasing order of $d_{\mathcal L}(\pi)$. If two
symmetric pairs $\pi$ and $\pi'$ satisfy $ d_{\mathcal L}(\pi)=d_{\mathcal L}(\pi')$, then the tie is resolved by comparing their corresponding closest nodes
$g_{\mathcal L}(\pi)$ and $g_{\mathcal L}(\pi')$ according to the order
$\prec_{\mathcal L}$. That is, $\pi$ precedes $\pi'$ if $g_{\mathcal L}(\pi)\prec_{\mathcal L} g_{\mathcal L}(\pi')$. If a tie still remains, then it is resolved using the pair-order
$\prec_{\mathcal C}$ induced by the key-corner string of $\mathcal C(t)$. The resulting ordered sequence of symmetric robot pairs is defined as the
ordering $\mathscr O_4$ for configurations in $\mathcal C_{222}$. Since
$\prec_{\mathcal L}$ and $\prec_{\mathcal C}$ are total orders, the first
symmetric pair in $\mathscr O_4$ is uniquely determined.

\item Suppose that $\mathcal{C}(t)\in\mathcal{C}_{322}$, and let $c$ be
the center of rotation. Let $\rho$ denote the smallest non-trivial rotation
about $c$ that preserves the configuration. The robot set $\mathcal R$ is
partitioned into distinct closed rotational orbits $\Omega_1,\Omega_2,\ldots,\Omega_k$, where each orbit is of the form $\Omega(r)=\{r,\rho(r),\rho^2(r),\ldots,\rho^{q-1}(r)\}$, $q$ is the order of the rotational symmetry, and $\rho^q(r)=r$.

For each closed orbit $\Omega_j$, define its distance from the center of
rotation $c$ as $ D(\Omega_j)=\min_{r\in\Omega_j} d_{\mathcal M}(r,c)$. The closed rotational orbits are first ordered in non-decreasing order of their distances from $c$. That is, for two orbits $\Omega_i$ and $\Omega_j$, $ \Omega_i \preceq \Omega_j$, if $D(\Omega_i)\leq D(\Omega_j)$. If two distinct closed orbits have the same distance from $c$, then the tie is resolved using the string direction induced by the leading corner of the
parking-node configuration. More precisely, let $Q_{\mathcal P}$ be the
leading corner of $\mathcal{MER}_{\mathcal M}$. The string direction
determined from $Q_{\mathcal P}$ induces a total order on the robot
positions. For each orbit $\Omega_j$, let $\operatorname{rep}(\Omega_j)$ be
the first robot of $\Omega_j$ in this string order. Then, for two orbits
$\Omega_i$ and $\Omega_j$ satisfying $ D(\Omega_i)=D(\Omega_j)$, we order them according to their representatives:
\[
    \Omega_i \prec \Omega_j
    \quad \text{if} \quad
    \operatorname{rep}(\Omega_i)
    \prec
    \operatorname{rep}(\Omega_j)
\]
in the leading-corner string order. The resulting ordered sequence of closed rotational orbits is defined as the ordering $\mathscr O_4$ for configurations in $\mathcal C_{322}$. Since the leading-corner string order is a total order, the first closed orbit in
$\mathscr O_4$ is uniquely determined.

\end{enumerate}

\end{itemize}

\subsection{Description of the Algorithm \textbf{\textsc{spg()}}}

In this section, we describe the algorithm for solving the \textit{Surplus Parking Gathering Problem}. The algorithm is designed according to the structural class of the initial configuration. Depending on the symmetry of the parking nodes and the robot configuration, the execution is divided into either two or four different phases. The overall execution flow of the algorithm depends on the initial configuration class $\mathcal C(t_0)$; the corresponding phase transitions, together with the impossible cases, are summarized in Figure~\ref{fig:phase-transition}.

\paragraph{Line Formation Phase:} This phase is executed only when the initial configuration $\mathcal C(t_0)$ belongs to either $\mathcal C_{21}$ or $\mathcal C_{31}$. The movement of robots for each configuration is performed by the algorithm \textit{MoveToDestination()}. A detailed description of the phase corresponding to different configurations is provided below.

%%%\FloatBarrier

\begin{algorithm}[ht]
\tiny
\caption{\textsc{: LineFormation}($\mathcal{C}(t)$)}
\label{alg:line-formation}
\begin{algorithmic}[1]
\Require Current configuration $\mathcal{C}(t)$, robot set $\mathcal{R}(t)$, parking node set $\mathcal{P}$
\Ensure All robots occupy distinct nodes on the formation line $\mathscr{L}$

\If{$\mathcal{C}(t_0) \notin \{\mathcal{C}_{21}, \mathcal{C}_{31}\}$}
    \State \Return
\EndIf

\If{$\mathcal{C}(t_0) \in \mathcal{C}_{21}$}
    \State Let $\mathcal{L}$ be the unique line of symmetry of $\mathcal{P}$
    \State For each robot $r_i \in \mathcal{R}(t)$, compute $\eta(r_i) = d_{\mathcal{M}}(r_i, \mathcal{L})$
    \State $d_{\mathcal{L}} \gets \max \{ \eta(r_i) : r_i \in \mathcal{R}(t) \}$
    \State Select a unique robot $r_g$ using $\mathscr{O}_2$ among all robots satisfying $\eta(r_i) = d_{\mathcal{L}}$

\ElsIf{$\mathcal{C}(t_0) \in \mathcal{C}_{31}$}
    \State Let $c$ be the center of rotational symmetry of $\mathcal{P}$
    \State For each robot $r_i \in \mathcal{R}(t)$, compute $\eta(r_i) = d_{\mathcal{M}}(r_i, c)$
    \State $d_c \gets \max \{ \eta(r_i) : r_i \in \mathcal{R}(t) \}$
    \State Select a unique robot $r_g$ using $\mathscr{O}_2$ among all robots satisfying $\eta(r_i) = d_c$
\EndIf

\State $(w, \mathscr{L}) \gets \textsc{ComputeFormationLine}(\mathcal{C}(t), r_g, \mathcal{MER}_{\mathcal{P}})$

\If{$r_g \notin \mathscr{L}$}
    \State \textsc{PlaceRobotOnFormationLine}($r_g, w, \mathscr{L}$)
\EndIf

\While{there exists a robot $r_i \in \mathcal{R}(t)$ such that $r_i \notin \mathscr{L}$}
    \State Let $\mathcal{R}_u(t) = \{ r_i \in \mathcal{R}(t) : r_i \notin \mathscr{L} \}$
    \State Select a robot $r_u \in \mathcal{R}_u(t)$ having maximum value of $\eta(r_u)$
    \State Break ties using the ordering $\mathscr{O}_2$
    \State Select a node $w_i \in \mathscr{L}$ as the destination node for $r_u$
    \State \textsc{PlaceRobotOnFormationLine}($r_u, w_i, \mathscr{L}$)
\EndWhile

\State \Return

\end{algorithmic}
\end{algorithm}

%%\FloatBarrier

\begin{figure}[htbp]
    \centering

    \begin{subfigure}[t]{0.59\textwidth}
        \centering
        \includegraphics[width=\textwidth]{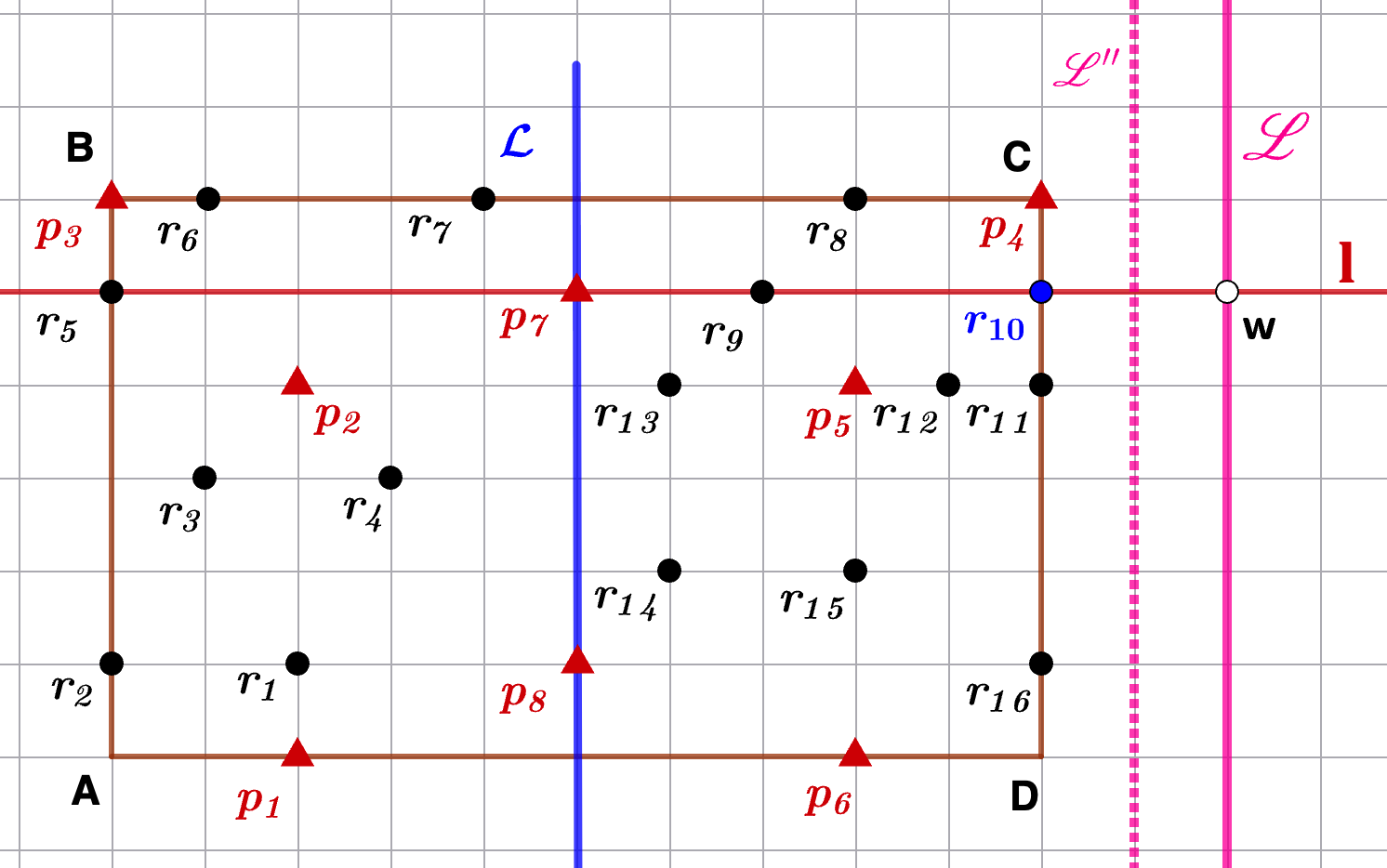}
       \caption{}
        \label{fig:Line_{I_21}}
    \end{subfigure}
    \hfill
    \begin{subfigure}[t]{0.35\textwidth}
        \centering
        \includegraphics[width=\textwidth]{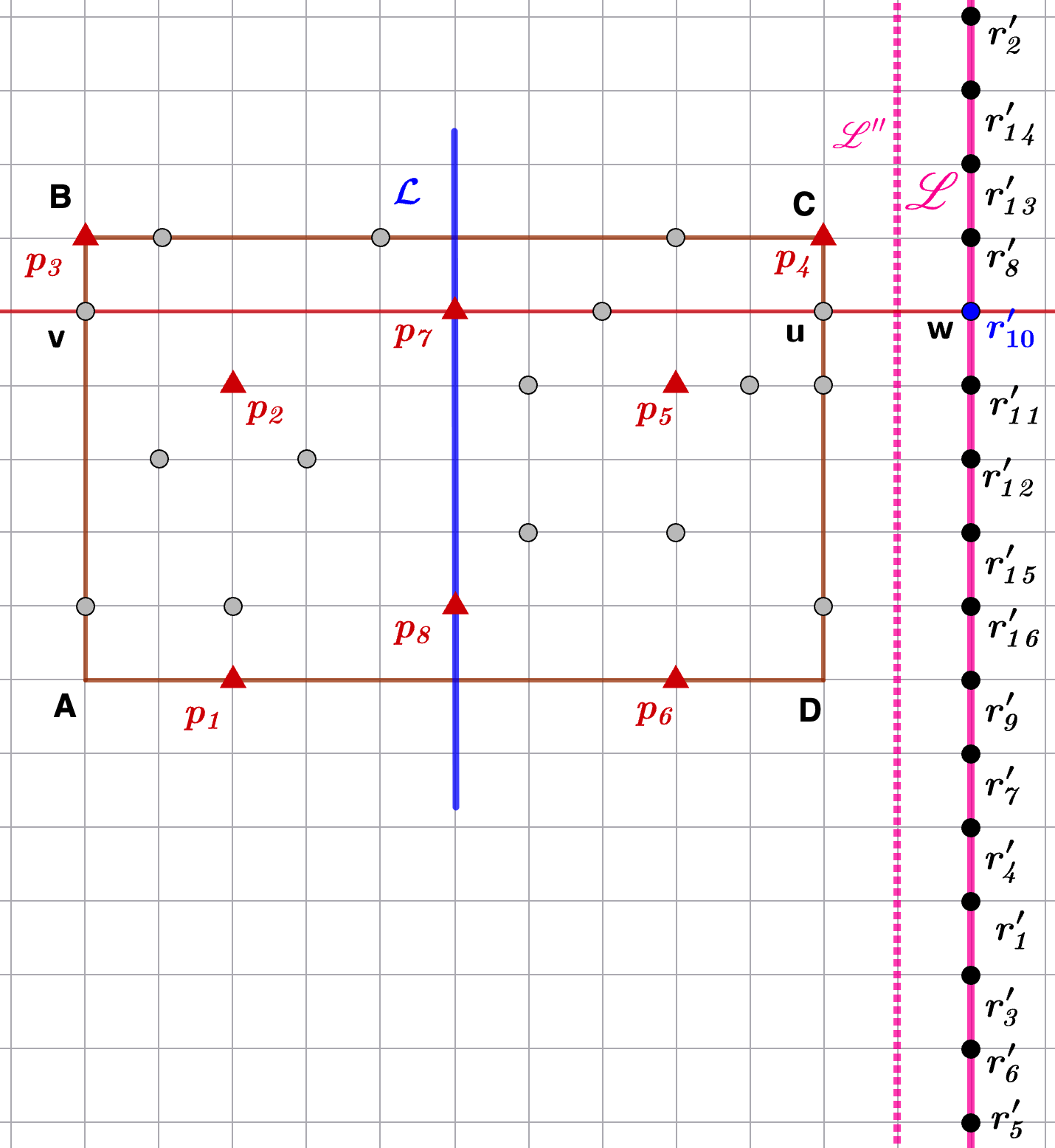}
        \caption{}
        \label{fig:TransitionOfLine_{I_21}}
    \end{subfigure}
\caption{ Illustration of the line-formation process for the configuration $\mathcal C_{21}$. (a) Selection of the formation line $\mathscr L$ and the auxiliary line $\mathscr L'$; (b) final placement of robots on $\mathscr L$ while preserving the prescribed ordering.}
    \label{TransitionOfLine}
\end{figure}

\begin{itemize}

\item \textbf{$\mathcal C_{21}$ Configuration:}
In this configuration, the parking node set $\mathcal P$ admits a unique line of symmetry $\mathcal L$, whereas the robot configuration $\mathcal R(t_0)$ is asymmetric. Define
$d_{\mathcal L}=\max\{d_{\mathcal M}(r_i,\mathcal L):r_i\in\mathcal R\}$.
A unique robot $r_g$ is selected using the ordering $\mathscr O_2$ among all robots attaining the distance $d_{\mathcal L}$. Let $AB$, $BC$, $CD$, and $DA$ be the four sides of $\mathcal{MER}_{\mathcal P}$ such that $AB$ and $CD$ are parallel to $\mathcal L$. Let $l$ be the line passing through $r_g$ and intersecting the sides $BC$ and $DA$ at nodes $u$ and $v$, respectively, where $d_{\mathcal M}(r_g,u)<d_{\mathcal M}(r_g,v)$.

If $r_g$ lies inside or on $\mathcal{MER}_{\mathcal P}$, then it moves along the line $l$ away from $v$ toward a node $w\in l$ outside $\mathcal{MER}_{\mathcal P}$ and stops at $w$ such that $d_{\mathcal M}(u,w)=2$ (see Figure~\ref{TransitionOfLine}(a)). Let $\mathscr L$ be the line passing through $w$ and parallel to $\mathcal L$. Clearly, in this case, $d_{\mathcal M}(\mathcal L,r_g)<d_{\mathcal M}(\mathcal L,w)$. 

If $r_g$ is outside $\mathcal{MER}_{\mathcal P}$ and satisfies $d_{\mathcal M}(\mathcal L,r_g)<d_{\mathcal M}(\mathcal L,w)$, then $r_g$ reaches $w$ (see Figure~\ref{TransitionOfLine}(b)). If the robot $r_g$ is outside $\mathcal{MER}_{\mathcal P}$ and $d_{\mathcal M}(\mathcal L,r_g)>d_{\mathcal M}(\mathcal L,w)$, then $r_g$ is away from $w$. In this case, $r_g$ moves towards $w$. If $\mathcal R \cap \{w\}=\emptyset$, that is, node $w$ does not contain another robot, then robot $r_g$ moves to $w$. Otherwise the robot $r_g$ moves along $l$ towards $w$ and stops at a node $w'\in l$ such that $d_{\mathcal M}(w,w')=1$. Consider an auxiliary line $\mathscr L'$ passing through $w'$ and parallel to $\mathscr L$. In this case, robot $r_g$ moves along a line $\mathscr L'$ and finds a free node, say $w''\in \mathscr L$. Once a free node $w''\in \mathscr L$ is identified, the robot $r_g$ moves to $w''$ and stops there. Similarly, for any robot $r_u\neq r_g$ satisfying $d_{\mathcal M}(\mathcal L,r_u)<d_{\mathcal M}(\mathcal L,w)$, an auxiliary line $\mathscr L''$ can be defined analogously to $\mathscr L'$ to reach distinct node on $\mathscr L$. Also, this line $\mathscr L''$ satisfies $d_{\mathcal M}(\mathcal L,\mathscr L'')<d_{\mathcal M}(\mathcal L,\mathscr L')$.

% %\FloatBarrier

\begin{algorithm}[ht]
\tiny
\caption{\textsc{: ComputeFormationLine}($\mathcal{C}(t), r_g, \mathcal{MER}_{\mathcal{P}}$)}
\label{alg:compute-formation-line}
\begin{algorithmic}[1]
\Require Initial configuration $\mathcal{C}(t_0) \in \{\mathcal{C}_{21}, \mathcal{C}_{31}\}$, guard robot $r_g$, rectangle $\mathcal{MER}_{\mathcal{P}}$
\Ensure A node $w$ and the formation line $\mathscr{L}$

\If{$\mathcal{C}(t_0) \in \mathcal{C}_{21}$}
    \State Let $\mathcal{L}$ be the unique line of symmetry of $\mathcal{P}$
    \State Let $AB, BC, CD, DA$ be the four sides of $\mathcal{MER}_{\mathcal{P}}$ such that $AB$ and $CD$ are parallel to $\mathcal{L}$
    \State Let $l$ be the line passing through $r_g$ and intersecting $BC$ and $DA$ at nodes $u$ and $v$, respectively
    \State Choose $u$ and $v$ such that $d_{\mathcal{M}}(r_g, u) < d_{\mathcal{M}}(r_g, v)$
    \State Choose a node $w \in l$ outside $\mathcal{MER}_{\mathcal{P}}$ in the direction away from $v$ such that $d_{\mathcal{M}}(u, w) = 2$
    \State Let $\mathscr{L}$ be the line passing through $w$ and parallel to $\mathcal{L}$

\ElsIf{$\mathcal{C}(t_0) \in \mathcal{C}_{31}$}
    \State Let $c$ be the center of rotational symmetry of $\mathcal{P}$
    \State Let $AB, BC, CD, DA$ be the four sides of $\mathcal{MER}_{\mathcal{P}}$
    \State Let $l$ be the line passing through $r_g$ and intersecting any two parallel sides of $\mathcal{MER}_{\mathcal{P}}$, $BC$ and $DA$, at nodes $u$ and $v$
    \State Choose $u$ and $v$ such that $d_{\mathcal{M}}(r_g, u) < d_{\mathcal{M}}(r_g, v)$
    \State Choose a node $w \in l$ outside $\mathcal{MER}_{\mathcal{P}}$ in the direction away from $c$ and $u$ such that $d_{\mathcal{M}}(u, w) = 2$
    \State Let $\mathscr{L}$ be the line passing through $w$ and perpendicular to $l$

\Else
    \State \Return
\EndIf

\State \Return $(w, \mathscr{L})$

\end{algorithmic}
\end{algorithm}

% %\FloatBarrier

\item \textbf{$\mathcal C_{31}$ Configuration:}
In this configuration, the parking node set $\mathcal P$ admits a center of rotational symmetry $c$, whereas the robot configuration $\mathcal R(t_0)$ is asymmetric. Define
$d_c=\max\{d_{\mathcal M}(r_i,c):r_i\in\mathcal R\}$. A unique robot $r_g$ is selected using the ordering $\mathscr O_2$, among all the robots attaining the distance $d_c$. Let $AB$, $BC$, $CD$, and $DA$ be the four sides of $\mathcal{MER}_{\mathcal P}$. Let $l$ be the line passing through $r_g$ and intersecting any two parallel sides of $\mathcal{MER}_{\mathcal P}$, say $BC$ and $DA$ at nodes $u$ and $v$, respectively, where $d_{\mathcal M}(r_g,u)<d_{\mathcal M}(r_g,v)$. 

If $r_g$ is located inside or on the boundary of $\mathcal{MER}_{\mathcal P}$, it moves along the line $l$, increasing its distance from both $c$ and $u$, until it reaches a node $w \in l$ that lies outside $\mathcal{MER}_{\mathcal P}$. The robot stops at $w$ such that $d_{\mathcal M}(u,w)=2$. Let $\mathscr L$ be the line passing through $w$ and perpendicular to the line $l$. Clearly, in this case, $d_{\mathcal M}(c,r_g)<d_{\mathcal M}(c,w)$. Once the line $\mathscr L$ is identified, the remaining computation for executing the \textit{Line Formation Phase} in configuration $\mathcal C_{31}$ follows the same procedure as in configuration $\mathcal C_{21}$, except that every comparison based on the distance from the line of symmetry $\mathcal L$ is replaced by the corresponding comparison based on the distance from the center of rotational symmetry $c$.

\end{itemize}

The auxiliary lines $\mathscr{L}'$ and $\mathscr{L}''$ are used to avoid the creation of multiple multiplicity nodes. If multiple multiplicity nodes are created in the configuration, then the \textit{Line Formation Phase} and the \textit{Saturation Phase} (discussed in detail on Page~\pageref{satu}) cannot be uniquely identified. 

Once the line $\mathscr{L}$ is formed (referred to as the \emph{formation line}), all robots uniquely identify it and move sequentially to distinct nodes on $\mathscr{L}$ according to the decreasing order of their Manhattan distances from $\mathcal{L}$ or $c$. Robots already located on $\mathscr{L}$ remain stationary during the \textit{Line Formation Phase}. This phase terminates when all $n$ robots occupy distinct nodes on $\mathscr{L}$. The detailed pseudocode for this phase is given in Algorithms~\ref{alg:line-formation}, \ref{alg:compute-formation-line}, and \ref{alg:place-robot-line}.

%\FloatBarrier

\begin{algorithm}[ht]
\tiny
\caption{\textsc{: PlaceRobotOnFormationLine}$(r,w,\mathscr L)$}
\label{alg:place-robot-line}
\begin{algorithmic}[1]
\Require A selected robot $r$, a node $w\in\mathscr L$, and the formation line $\mathscr L$
\Ensure Robot $r$ reaches a distinct node on $\mathscr L$

\If{$r\in\mathscr L$}
\State \Return
\EndIf

\If{$\mathcal R(t)\cap\{w\}=\emptyset$}
\While{$r\neq w$}
\State Robot $r$ moves by one hop along a shortest Manhattan path toward $w$
\State All other robots remain stationary
\EndWhile
\Else
\State Let $w'$ be a node on the approach line of $r$ such that $d_{\mathcal M}(w,w')=1$
\State Let $\mathscr L'$ be the auxiliary line passing through $w'$ and parallel to $\mathscr L$

\While{$r\neq w'$}
    \State Robot $r$ moves by one hop along a shortest Manhattan path toward $w'$
    \State All other robots remain stationary
\EndWhile

\State Select a free node $w''\in\mathscr L$ such that $\mathcal R(t)\cap\{w''\}=\emptyset$

\While{$r\neq w''$}
    \State Robot $r$ moves by one hop along a shortest Manhattan path toward $w''$
    \State All other robots remain stationary
\EndWhile

\EndIf

\State \Return

\end{algorithmic}
\end{algorithm}

%\FloatBarrier

\begin{figure}[htbp]
    \centering

    \begin{subfigure}[t]{0.35\textwidth}
        \centering
        \includegraphics[width=\textwidth]{TransitionOfLine__I_21.png}
        \caption{}
    \end{subfigure}%
    \hspace{0.05\textwidth}%
    \begin{subfigure}[t]{0.35\textwidth}
        \centering
        \includegraphics[width=\textwidth]{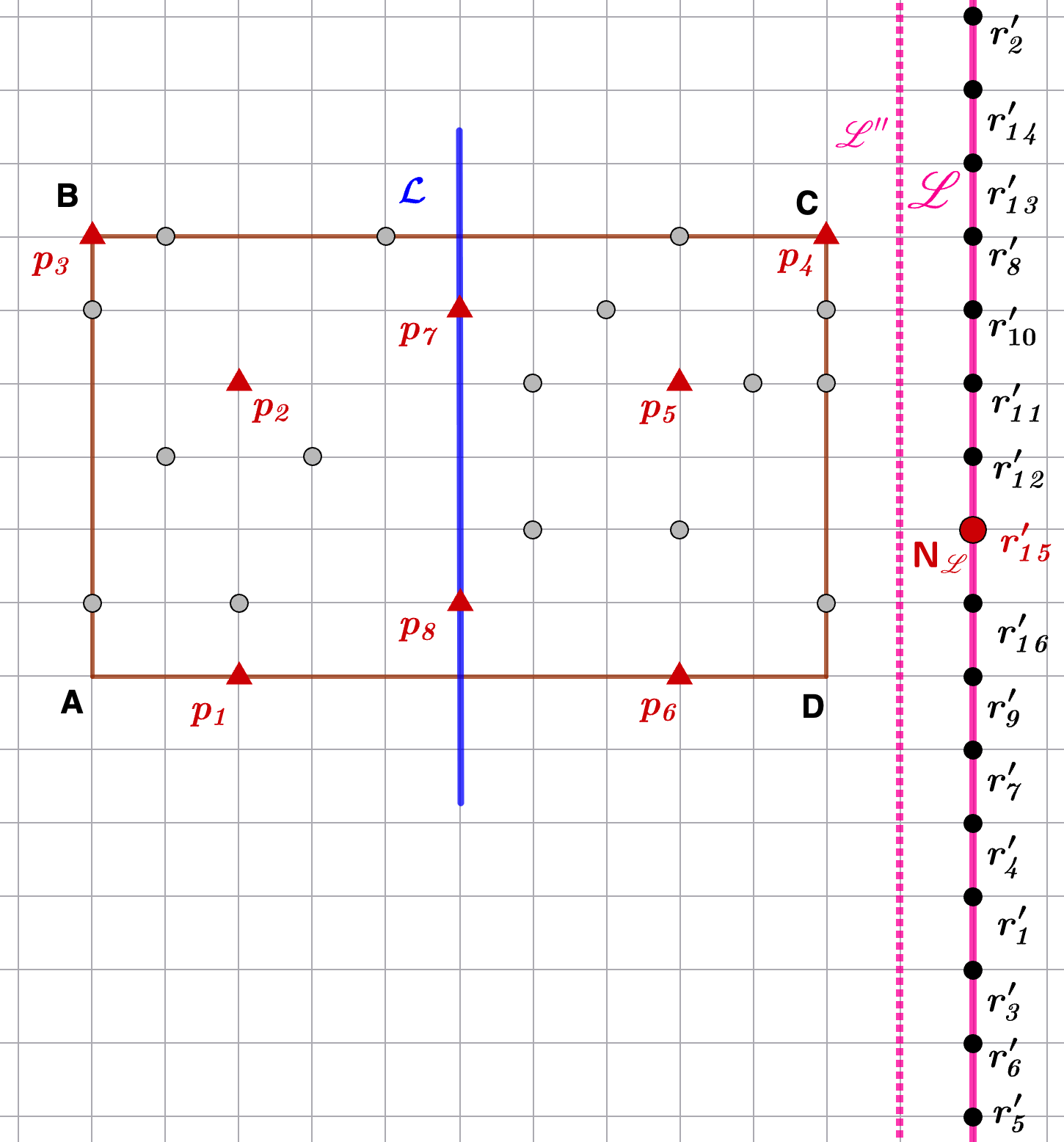}
        \caption{}
    \end{subfigure}

\caption{
Illustration of the multiplicity creation phase after line formation for the configuration $\mathcal C_{21}$: (a) all robots are positioned distinctly on the formation line $\mathscr L$ with \textit{terminal robots} $r_2$ and $r_5$; (b) the node $\mathbf N_{\mathscr L}$ is selected as the unique multiplicity node with two leading corners $p_3$ and $p_4$.
}
    \label{MultiCreatePhase}
\end{figure}

% \begin{figure}[htbp]
%     \centering

%     \begin{subfigure}[t]{0.35\textwidth}
%         \centering
%        \includegraphics[width=\textwidth]{TransitionOfLine_{I_21}.png}
%        \caption{}
%         \label{fig:Multi}
%     \end{subfigure}
%     \hfill
%     \begin{subfigure}[t]{0.35\textwidth}
%         \centering
%       \includegraphics[width=\textwidth]{MultiCreatePhase.png}
%         \caption{}
%         \label{}
%     \end{subfigure}

% \caption{
% Illustration of the multiplicity creation phase after line formation for the configuration $\mathcal C_{21}$: (a) all robots are positioned distinctly on the formation line $\mathscr L$ with \textit{terminal robots} $r_2$ and $r_5$; (b) the node $\mathbf N_{\mathscr L}$ is selected as the unique multiplicity node with two leading corners $p_3$ and $p_4$.
% }
%     \label{MultiCreatePhase}
% \end{figure}

\paragraph{Multiplicity Creation Phase:}
If $\mathcal C(t_0)\in \lbrace \mathcal C_{21},\mathcal C_{31} \rbrace$, then this phase is executed after the completion of the \textit{Line Formation Phase}. Otherwise, the algorithm proceeds directly to the next phase for all the other configurations. The purpose of this phase is to create a multiplicity node so that the robots can distinguish between the \textit{Line Formation Phase} and the \textit{Saturation Phase}. For configuration $\mathcal C_{21}$, let $[r_1,r_n]$ be the line segment on $\mathscr L$ induced by the robots, where $r_1$ and $r_n$ are the two terminal robots on $\mathscr L$. If $[r_1,r_n]$ contains an odd number of nodes, then its unique middle node is chosen as the multiplicity node. Otherwise, $[r_1,r_n]$ contains two middle nodes; the one closer to the leading corners of $\mathcal{MER}_{\mathcal P}$ is chosen as the multiplicity node (see Figure~\ref{MultiCreatePhase}(a)). Note that, as the configuration is asymmetric, there always exists a unique key corner. Let the multiplicity node be denoted by $\mathbf N_{\mathscr L}$ for the configuration $\mathcal C_{21}$ (see Figure~\ref{MultiCreatePhase}(b)). For the configuration $\mathcal C_{31}$, let $\mathbf N_{c}$ denote the intersection node of $\mathscr L$ and the line passing through $c$ that is perpendicular to $\mathscr L$. This node is selected as the multiplicity node. Observe that if any robot is located at the multiplicity node in either configuration $\mathcal C_{21}$ or $\mathcal C_{31}$, then it remains stationary. The movement of robots for each configuration is performed by the algorithm \textit{MoveToDestination()}. The detailed pseudocode for this phase is given in Algorithm~\ref{alg:multiplicity-creation}.

%\FloatBarrier

\begin{algorithm}[ht]
\tiny
\caption{\textsc{: MultiplicityCreation}$(\mathcal C(t))$}
\label{alg:multiplicity-creation}
\begin{algorithmic}[1]
\Require Current configuration $\mathcal C(t_0)$,
         robot set $\mathcal R(t)$,
         formation line $\mathscr L$,
         rectangle $\mathcal{MER}_{\mathcal P}$,
         ordering $\mathscr O_2$
\Ensure A unique multiplicity node is created on $\mathscr L$

\If{$\mathcal C(t_0)\notin\{\mathcal C_{21},\mathcal C_{31}\}$}
    \State \Return
\EndIf

\If{$\mathcal C(t_0)\in\mathcal C_{21}$}
    \State Let $r_1$ and $r_n$ be the two terminal robots on $\mathscr L$
    \State Let $[r_1,r_n]$ be the line segment on $\mathscr L$ induced by the robots
    \State Let $S$ denote the set of nodes on the segment $[r_1,r_n]$

    \If{$|S|$ is odd}
        \State Let $\mathbf N_{\mathscr L}$ be the unique middle node of $[r_1,r_n]$
    \Else
        \State Let $x$ and $y$ be the two middle nodes of $[r_1,r_n]$
        \State Let $K$ be the unique key corner of $\mathcal{MER}_{\mathcal P}$
        \If{$d_{\mathcal M}(x,K)<d_{\mathcal M}(y,K)$}
            \State $\mathbf N_{\mathscr L}\gets x$
        \Else
            \State $\mathbf N_{\mathscr L}\gets y$
        \EndIf
    \EndIf
    \State $\mathbf N\gets \mathbf N_{\mathscr L}$

\ElsIf{$\mathcal C(t_0)\in\mathcal C_{31}$}
    \State Let $c$ be the center of rotational symmetry of $\mathcal P$
    \State Let $l_c$ be the line passing through $c$ and perpendicular to $\mathscr L$
    \State Let $\mathbf N_c$ be the intersection node of $l_c$ and $\mathscr L$
    \State $\mathbf N\gets \mathbf N_c$
\EndIf

\While{$|\mathcal R(t)\cap\{\mathbf N\}|<2$}
    \State Let $\mathcal R_u(t)=\mathcal R(t)\setminus(\mathcal R(t)\cap\{\mathbf N\})$
    \State Select a robot $r\in\mathcal R_u(t)$ having minimum value of $d_{\mathcal M}(r,\mathbf N)$
    \State Break ties using the ordering $\mathscr O_2$
    \State Robot $r$ moves toward $\mathbf N$ using \textsc{MoveToDestination}$(r,\mathbf N)$
    \State Every robot already located at $\mathbf N$ remains stationary
\EndWhile

\State \Return

\end{algorithmic}
\end{algorithm}

%\FloatBarrier

\paragraph{Saturation Phase:}
\label{satu}
If $\mathcal C(t_0)\in\{\mathcal C_{21},\mathcal C_{31}\}$, then this phase is executed after the completion of the \textit{Multiplicity Phase}. If $\mathcal C(t_0)\in\{\mathcal C_{222},\mathcal C_{322}\}$, then this phase is executed after the completion of the \textit{Gathering Phase}. Otherwise, this phase is executed as the first phase of the algorithm \textbf{\textsc{spg()}}. Once a multiplicity node or a gathering node is created, no robot located at that node is allowed to leave it. Since the robots are equipped with global strong multiplicity detection, each robot can determine whether it is located at the multiplicity node or the gathering node. Hence, such robots remain stationary in all subsequent phases. Since the robots have global visibility, they can distinguish among the different configuration classes. Note that the movement of robots for each configuration is performed by the algorithm \textit{MoveToDestination()}. The detailed description of this phase is as follows.

%------------------------------------------------
% Common helper routine for saturation
%------------------------------------------------

%\FloatBarrier

\begin{algorithm}[ht]
\tiny
\caption{\textsc{: SaturateParkingNode}$(p,\mathscr O)$}
\label{alg:saturate-one-node}
\begin{algorithmic}[1]
\Require An unsaturated parking node $p$ and a tie-breaking ordering $\mathscr O$
\Ensure The parking node $p$ becomes saturated

\While{$|\mathcal R(t)\cap\{p\}|<\kappa$}
\State Let $\mathcal R_u(t)$ be the set of unsaturated robots
\State Order the robots in $\mathcal R_u(t)$ increasingly according to $d_{\mathcal M}(r,p)$
\State Break ties using the ordering $\mathscr O$
\State Let $r$ be the first robot in the resulting order
\State Robot $r$ moves toward $p$ using \textsc{MoveToDestination}$(r,p)$
\State All other robots remain stationary during this move
\EndWhile

\State \Return

\end{algorithmic}
\end{algorithm}

%\FloatBarrier

\begin{itemize}

\item \textbf{$\mathcal C_{1}$ Configuration:} 
Since the set of parking nodes $\mathcal P$ is asymmetric, all the parking nodes in $\mathcal P$ can be uniquely ordered. The robots in $\mathcal R$ first establish a common ordering $\mathscr O_1$ of the unsaturated parking nodes. They saturate these parking nodes sequentially according to the ordering $\mathscr O_1$, that is, $p_1$ is saturated first, followed by $p_2$, and so on. The procedure for saturating a single unsaturated parking node is formally described in Algorithm~\ref{alg:saturate-one-node}. Let $\mathcal R_u(t)=\{r_1(t),r_2(t),\ldots,r_m(t)\}\subset \mathcal R(t)$ denote the set of unsaturated robot positions, and let $p$ be the selected unsaturated parking node with capacity $\kappa$. The robots in $\mathcal R_u(t)$ are ordered such that $d_{\mathcal M}(r_1(t),p)\le d_{\mathcal M}(r_2(t),p)\le \cdots \le d_{\mathcal M}(r_m(t),p)$, where $d_{\mathcal M}(r_i,p)$ denotes the Manhattan distance between the robot $r_i$ and the parking node $p$. Note that a robot selected to saturate an unsaturated parking node must not occupy any saturated parking node. The set of selected robots is defined as $\mathcal R_p^*=\{r_1,r_2,\ldots,r_{\kappa}\}$. Thus, the first $\kappa$ robots in this ordering move toward $p$ in a sequential order. If multiple robots have the same distance from $p$, the ties are broken using a fixed ordering~\cite{bose2020arbitrary}, depending on the configuration of the robot positions. Since the capacities of the parking nodes are known, the \textit{strong multiplicity detection} capability enables every robot to determine when the occupancy of a parking node reaches its prescribed capacity. The unsaturated parking nodes are saturated sequentially in order to avoid collisions among robots and prevent the formation of multiple multiplicity nodes. Note that if a robot is already a saturated robot, then it will never become unsaturated.

%------------------------------------------------
% C_1 Configuration
%------------------------------------------------

 %\FloatBarrier

\begin{algorithm}[ht]
\tiny
\caption{\textsc{: SaturationPhaseC1}$(\mathcal C(t))$}
\label{alg:saturation-c1}
\begin{algorithmic}[1]
\Require Current configuration $\mathcal C(t_0)$, parking node set $\mathcal P$
\Ensure All parking nodes are saturated

\If{$\mathcal C(t_0)\notin\mathcal C_1$}
\State \Return
\EndIf

\State Construct the common ordering
$\mathscr O_1=(p_1,p_2,\ldots,p_m)$
of the unsaturated parking nodes in $\mathcal P$

\For{$i=1$ to $m$}
    \State \textsc{SaturateParkingNode}$(p_i,\mathscr O_1)$
\EndFor

\State \Return

\end{algorithmic}
\end{algorithm}

%\FloatBarrier

\item \textbf{$\mathcal{C}_{21}$ Configuration:} 
The parking-node configuration $\mathcal{P}$ admits a unique line of symmetry $\mathcal{L}$, whereas the robot configuration $\mathcal{R}(t_0)$ is asymmetric. Let $\mathcal{P^S}$ and $\mathcal{P^U}$ denote the sets of saturated and unsaturated parking nodes, respectively, where the set of parking nodes is $\mathcal{P} = \mathcal{P^S} \cup \mathcal{P^U}$. Further, let $\mathcal{P}_{\mathcal{L}} \subseteq \mathcal{P^U}$ denote the set of unsaturated parking nodes lying on $\mathcal{L}$ and let $\mathcal{P^U} \setminus \mathcal{P}_{\mathcal{L}}$ denote the set of unsaturated parking nodes lying outside $\mathcal{L}$. (see Figure~\ref{saturationphase})

% \begin{figure}[htbp]
%     \centering

%     \begin{subfigure}[t]{0.35\textwidth}
%         \centering
%        \includegraphics[width=\textwidth]{SaturationPha.png}
%        \caption{}
%     \end{subfigure}
%     \hfill
%     \begin{subfigure}[t]{0.35\textwidth}
%         \centering
%       \includegraphics[width=\textwidth]{saturation2.png}
%         \caption{}
%     \end{subfigure}

% \caption{
% Illustration of the saturation process for configuration $\mathcal C_{21}$: 
% (a) creation of the unique multiplicity node at $\mathbf N_{\mathscr L}$; 
% (b) since $\mathcal P_{\mathcal L}=\{p_7,p_8\}\subseteq \mathcal {P^U}$, 
% the parking node $p_7$ is saturated with capacity $\kappa_7=4$.
% }
%     \label{saturationphase}
% \end{figure}

\begin{figure}[htbp]
    \centering

    \begin{subfigure}[t]{0.35\textwidth}
        \centering
        \includegraphics[width=\textwidth]{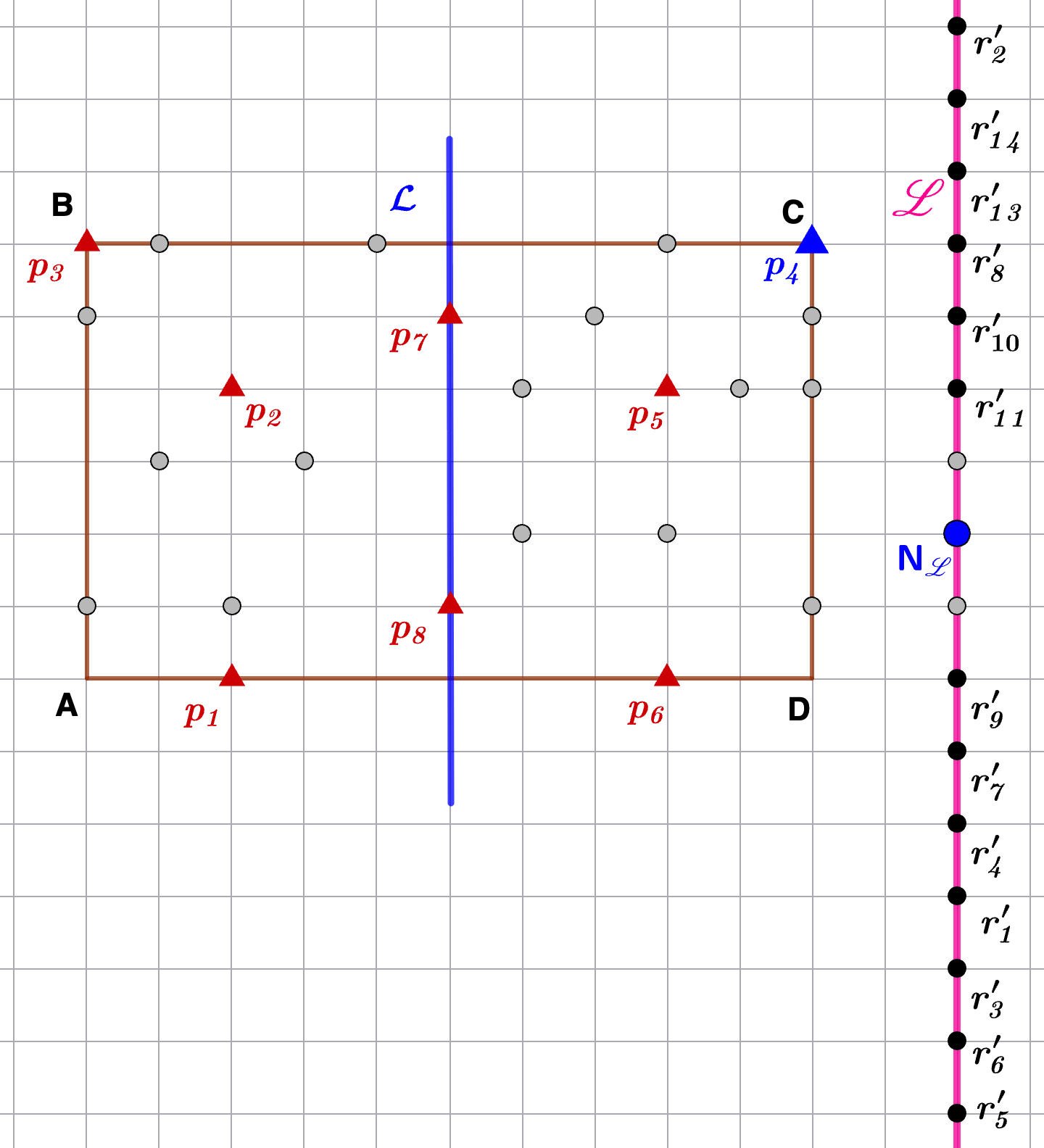}
        \caption{}
    \end{subfigure}%
    \hspace{0.05\textwidth}%
    \begin{subfigure}[t]{0.35\textwidth}
        \centering
        \includegraphics[width=\textwidth]{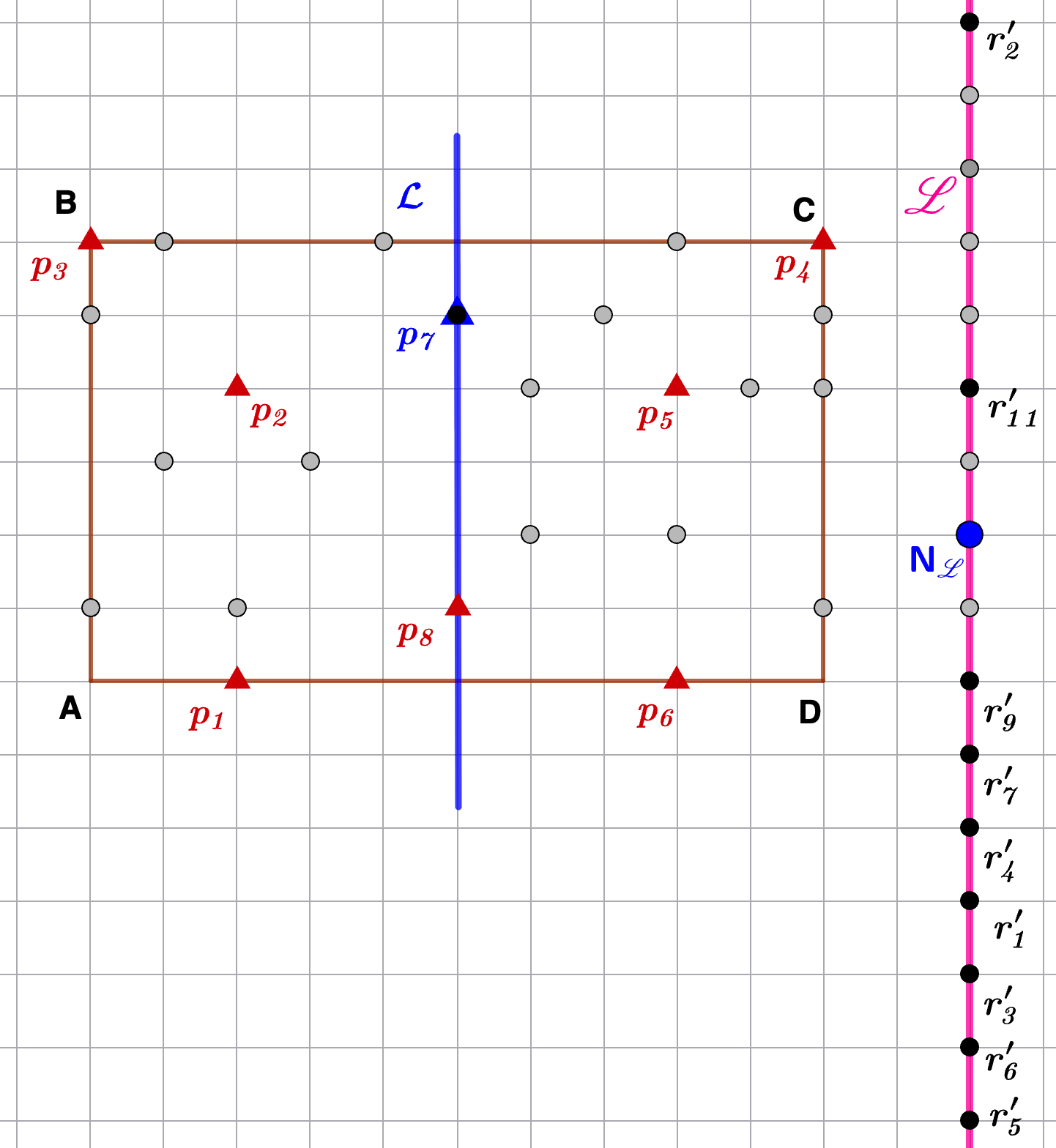}
        \caption{}
    \end{subfigure}

\caption{
Illustration of the saturation process for configuration $\mathcal C_{21}$: 
(a) creation of the unique multiplicity node at $\mathbf N_{\mathscr L}$; 
(b) since $\mathcal P_{\mathcal L}=\{p_7,p_8\}\subseteq \mathcal {P^U}$, 
the parking node $p_7$ is saturated with capacity $\kappa_7=4$.
}
\label{saturationphase}
\end{figure}

The saturation phase is executed in two stages. In the first stage, the parking nodes belonging to $\mathcal{P}_{\mathcal{L}}$ are saturated. Once all these parking nodes have reached their prescribed capacities, the algorithm proceeds to saturate the parking nodes in $\mathcal{P^U} \setminus \mathcal{P}_{\mathcal{L}}$.

For the first stage, let $p_{\mathcal{L}}$ be the unsaturated parking node on $\mathcal{L}$ selected according to the ordering $\mathscr{O}_3$, and let its capacity be $\kappa$. Let $\mathcal{R}_u(t) = \{r_1(t), r_2(t), \ldots, r_{\kappa}(t)\}$ denote the set of unsaturated robots on $\mathscr L$. The robots are ordered in nondecreasing order of their Manhattan distances from $p_{\mathcal{L}}$, i.e.,
$d_{\mathcal{M}}(r_1, p_{\mathcal{L}}) \le d_{\mathcal{M}}(r_2, p_{\mathcal{L}}) \le \cdots \le d_{\mathcal{M}}(r_{\kappa}, p_{\mathcal{L}}).$
Ties are resolved using the ordering $\mathscr{O}_2$. The first $\kappa$ robots in this ordering are selected to saturate $p_{\mathcal{L}}$. As the robots have strong multiplicity detection capability, every robot can determine when $p_{\mathcal{L}}$ has reached its prescribed capacity.

After all parking nodes on $\mathcal{L}$ have been saturated, consider the remaining unsaturated parking nodes
$\mathcal{P^U} \setminus \mathcal{P}_{\mathcal{L}} = \{p_1, p_2, \ldots, p_d\},$
ordered in nondecreasing order of their Manhattan distances from the multiplicity node $\mathbf{N}_{\mathscr{L}}$, with ties broken according to $\mathscr{O}_2$. The parking nodes are then saturated sequentially following this order. 

For each selected parking node $p_i$ with capacity $\kappa$, the unsaturated robots are ordered in non-decreasing order of their Manhattan distances from $p_i$, with ties broken according to $\mathscr{O}_2$. The first $\kappa$ robots in the resulting order are assigned to $p_i$. Again, strong multiplicity detection capability enables the robots to identify when the parking node has become saturated, after which the algorithm proceeds to the next parking node in the sequence. The detailed pseudocode for this phase is given in Algorithm~\ref{alg:saturation-c1}.

The parking nodes are saturated one at a time to ensure collision-free movements and to prevent the creation of unintended multiplicity nodes during the execution.

% %\FloatBarrier

\begin{algorithm}[ht]
\tiny
\caption{\textsc{: SaturationPhaseC21C31}$(\mathcal C(t))$}
\label{alg:saturation-c21-c31}
\begin{algorithmic}[1]
\Require Current configuration $\mathcal C(t_0)\in\{\mathcal C_{21},\mathcal C_{31}\}$
\Ensure All parking nodes are saturated

\If{$\mathcal C(t_0)\notin\{\mathcal C_{21},\mathcal C_{31}\}$}
\State \Return
\EndIf

\State Let $\mathcal {P^U}$ be the set of unsaturated parking nodes

\If{$\mathcal C(t_0)\in\mathcal C_{21}$}
\State Let $\mathcal L$ be the unique line of symmetry of $\mathcal P$
\State Let $\mathbf N=\mathbf N_{\mathscr L}$
\State Let $\mathcal P_{\mathcal L}=\lbrace p\in\mathcal {P^U}:p\in\mathcal L \rbrace$

\While{$\mathcal P_{\mathcal L}\neq\emptyset$}
    \State Select a parking node $p_{\mathcal L}\in\mathcal P_{\mathcal L}$ using the ordering $\mathscr O_3$
    \State \textsc{SaturateParkingNode}$(p_{\mathcal L},\mathscr O_2)$
    \State Update $\mathcal {P^U}$ and $\mathcal P_{\mathcal L}$
\EndWhile

\State Let $\mathcal P_{\mathrm{out}}=\mathcal {P^U}\setminus\mathcal P_{\mathcal L}$

\ElsIf{$\mathcal C(t_0)\in\mathcal C_{31}$}
\State Let $c$ be the center of rotational symmetry of $\mathcal P$
\State Let $\mathbf N=\mathbf N_c\in \mathscr L$

\If{there exists an unsaturated parking node $\mathbf p_c$ located at $c$}
    \State \textsc{SaturateParkingNode}$(\mathbf p_c,\mathscr O_2)$
    \State Update $\mathcal {P^U}$
    \State Let $\mathcal P_{\mathrm{out}}=\mathcal {P^U}\setminus\{\mathbf p_c\}$
\Else
    \State Let $\mathcal P_{\mathrm{out}}=\mathcal {P^U}$
\EndIf

\EndIf

\While{$\mathcal P_{\mathrm{out}}\neq\emptyset$}
    \State Select a farthest parking node $\mathbf p_d\in\mathcal P_{\mathrm{out}}$ from $\mathbf N$
    \State Break ties using the ordering $\mathscr O_2$
    \State \textsc{SaturateParkingNode}$(\mathbf p_d,\mathscr O_2)$
    \State Update $\mathcal {P^U}$
    \State $\mathcal P_{\mathrm{out}}\gets\mathcal {P^U}\cap\mathcal P_{\mathrm{out}}$
\EndWhile

\State \Return

\end{algorithmic}
\end{algorithm}

% %\FloatBarrier

\item \textbf{$\mathcal C_{31}$ Configuration:}
The parking-node configuration $\mathcal P$ admits rotational symmetry with center of rotation $c$, while the robot configuration $\mathcal R(t_0)$ is asymmetric. Let $\mathcal P=\mathcal {P^S}\cup\mathcal {P^U}$, where $\mathcal {P^S}$ and $\mathcal {P^U}$ denote the sets of saturated and unsaturated parking nodes, respectively. Let $\mathbf p_c\in\mathcal {P^U}$ be the unique unsaturated parking node located at the center of rotation $c$, and let $\mathcal {P^U}\setminus\{\mathbf p_c\}$ denote the set of all remaining unsaturated parking nodes. The \textit{Saturation Phase} is carried out in two stages. In the first stage, the parking node $\mathbf p_c$ is saturated. Let $\mathcal R_u(t)=\{r_1(t),r_2(t),\ldots,r_m(t)\}\subseteq\mathcal R(t)$ denote the set of unsaturated robots at time $t$. The robots are ordered according to their Manhattan distance from $\mathbf p_c$, namely,
$
d_{\mathcal M}(r_1,\mathbf p_c)\le d_{\mathcal M}(r_2,\mathbf p_c)\le\cdots\le d_{\mathcal M}(r_m,\mathbf p_c).
$
Suppose that $\mathbf p_c$ has capacity $\kappa$. Then the first $\kappa$ robots in this ordering,
$
\mathcal R_p^*=\{r_1,r_2,\ldots,r_\kappa\},
$
are selected to saturate $\mathbf p_c$ and move toward it. If multiple robots are equidistant from $\mathbf p_c$, ties are resolved using the fixed ordering $\mathscr O_2$ defined in~\cite{bose2020arbitrary}. By means of strong multiplicity detection, the robots can determine when $\mathbf p_c$ becomes saturated.

Once $\mathbf p_c$ is saturated, the algorithm then proceeds to saturate the remaining parking nodes in $\mathcal {P^U}\setminus\{\mathbf p_c\}$. Let
$
\mathcal {P^U}\setminus\{\mathbf p_c\}=\{\mathbf p_1,\mathbf p_2,\ldots,\mathbf p_d\},
$
where the parking nodes are ordered in non-decreasing order of their Manhattan distance from $\mathbf N_c$, that is,
$
d_{\mathcal M}(\mathbf p_1,\mathbf N_c)\le d_{\mathcal M}(\mathbf p_2,\mathbf N_c)\le\cdots\le d_{\mathcal M}(\mathbf p_d,\mathbf N_c).
$
The farthest parking node, $\mathbf p_d$, is selected first for saturation; ties are broken according to the ordering $\mathscr O_2$. The unsaturated robots are ordered according to their Manhattan distance from $\mathbf p_d$, and the first $\kappa$ robots on $\mathscr L$ in this ordering, where $\kappa$ is the capacity of $\mathbf p_d$, are selected to saturate it. Ties are again resolved using the ordering $\mathscr O_2$. After $\mathbf p_d$ becomes saturated, the remaining unsaturated parking nodes are saturated sequentially by repeatedly applying the same procedure until every parking node reaches its prescribed capacity. This sequential saturation strategy prevents collisions among robots and avoids the creation of multiple multiplicity nodes during execution. The detailed execution of the saturation phase for configurations $\mathcal C_{21}$ and $\mathcal C_{31}$ is given in 
Algorithm~\ref{alg:saturation-c21-c31}.

\begin{algorithm}[ht]
\tiny
\caption{\textsc{: SaturationPhaseC221C321}$(\mathcal C(t))$}
\label{alg:saturation-c221-c321}
\begin{algorithmic}[1]
\Require Initial configuration $\mathcal C(t_0)\in\{\mathcal C_{221},\mathcal C_{321}\}$
\Ensure The symmetry is broken and all parking nodes are saturated

\If{$\mathcal C(t_0)\notin\{\mathcal C_{221},\mathcal C_{321}\}$}
    \State \Return
\EndIf

\If{$\mathcal C(t_0)\in\mathcal C_{221}$}

    \State Let $\mathcal L$ be the unique line of symmetry of $\mathcal P$
    \State Let $r_{e_1}$ and $r_{e_2}$ be the two terminal robots on $\mathcal L$

    \If{both $r_{e_1}$ and $r_{e_2}$ lie strictly inside $\mathcal{MER}_{\mathcal P}$}
        \State Let $u$ and $v$ be the intersection nodes of $\mathcal L$ with the two sides of $\mathcal{MER}_{\mathcal P}$ perpendicular to $\mathcal L$
        \State Move $r_{e_1}$ along $\mathcal L$ to the node one hop beyond $u$
        \State Move $r_{e_2}$ along $\mathcal L$ to the node one hop beyond $v$
        \State Select one terminal robot, say $r_{e_1}$, using the ordering $\mathscr O_3$
    \Else
        \State Select a terminal robot lying on or outside $\mathcal{MER}_{\mathcal P}$, say $r_{e_1}$, using the ordering $\mathscr O_3$
        \State Move $r_{e_1}$ one hop along $\mathcal L$ away from the nearest side of $\mathcal{MER}_{\mathcal P}$
    \EndIf

    \State Select an adjacent free node $g$ such that $g\notin\mathcal L$

    \If{$r_{e_1}$ is in a pending move state}
        \State Execute \textsc{AllowtoMove}$(r_{e_1})$
    \Else
        \State Move $r_{e_1}$ to $g$
    \EndIf

    \If{$\mathcal C(t)$ is asymmetric}
        \State Execute \textsc{SaturationPhaseC21C31}$(\mathcal C(t))$
    \EndIf

\ElsIf{$\mathcal C(t_0)\in\mathcal C_{321}$}

    \State Let $c$ be the center of rotational symmetry of $\mathcal P$
    \State Let $r_c$ be the robot located at $c$
    \State Designate $r_c$ as the symmetry-breaking robot
    \State Execute \textsc{AllowtoMove}$(r_c)$

    \If{$\mathcal C(t)$ is asymmetric}
        \State Execute \textsc{SaturationPhaseC21C31}$(\mathcal C(t))$
    \EndIf

\EndIf

\State \Return

\end{algorithmic}
\end{algorithm}

\item \textbf{$\mathcal C_{221}$ Configuration:}
Consider the set of robots lying on the line of symmetry $\mathcal L$. Let $r_{e_1}$ and $r_{e_2}$ denote the two \textit{terminal} robots, i.e., the robots located at the endpoints of the line segment formed by all robots on $\mathcal L$. One of these terminal robots is eventually selected to perform the symmetry-breaking. If both $r_{e_1}$ and $r_{e_2}$ lie strictly inside $\mathcal{MER}_{\mathcal P}$, then each robot moves along $\mathcal L$ toward the nearest side of $\mathcal{MER}_{\mathcal P}$ whose supporting line is perpendicular to $\mathcal L$. Let $u$ and $v$ denote the intersection nodes of $\mathcal L$ with these two sides. The robots continue moving until they reach the nodes on $\mathcal L$ located one hop beyond $u$ and $v$, respectively. If a terminal robot already lies on a side of $\mathcal{MER}_{\mathcal P}$, then it moves along $\mathcal L$ to the node located one hop beyond its current position. After these movements are completed, one of the terminal robots is uniquely selected using the ordering $\mathscr O_3$. Assume that without loss of generality that the selected robot is $r_{e_1}$. The selected robot $r_{e_1}$ then moves to an adjacent free node $g$ such that $g\notin\mathcal L$, thereby breaking the symmetry of the configuration.

If at least one terminal robot lies on the boundary of $\mathcal{MER}_{\mathcal P}$ or outside $\mathcal{MER}_{\mathcal P}$, then such a terminal robot is selected for symmetry breaking. Ties are broken accordingly. The selected robot moves along $\mathcal L$, away from the corresponding side of $\mathcal{MER}_{\mathcal P}$, by one hop. Next, it moves to an adjacent free node $g$ such that $g\notin\mathcal L$, thereby breaking the symmetry of the configuration. If the selected robot is in a \textit{pending move} state, the symmetry of the configuration is, however, broken during the execution of the procedure \textit{AllowtoMove()}~\cite{chakraborty2024gathering}. Once the symmetry is broken and the configuration becomes asymmetric, the algorithm proceeds to the \textit{Saturation Phase} as described for the configuration $\mathcal C_{21}$.

\item \textbf{$\mathcal C_{321}$ Configuration:}  
In this configuration, one robot, say $r_c$, is located at the center of rotation $c$. This robot $r_c$ is eventually selected to perform the symmetry-breaking operation. The procedure of the algorithm \textit{AllowtoMove()}~\cite{chakraborty2024gathering} is used to transform the initial symmetric configuration $\mathcal C(t_0)$ into an asymmetric configuration. Also, if the selected robot is in a \textit{pending move} state, the symmetry of the configuration is, however, broken during the execution of the procedure \textit{AllowtoMove()}~\cite{chakraborty2024gathering}. Once the initial configuration $\mathcal C(t_0)$ becomes asymmetric, the \textit{Saturation Phase} is executed in the same manner as for the configuration $\mathcal C_{31}$, where no parking node is located at $c$. Algorithm~\ref{alg:saturation-c221-c321} describes the saturation phase for configurations $\mathcal C_{221}$ and $\mathcal C_{321}$, where the symmetry is first broken by a uniquely selected robot and the resulting asymmetric configuration is then handled by 
\textsc{SaturationPhaseC21C31}.

% %\FloatBarrier

% %\FloatBarrier

\item \textbf{$\mathcal C_{222}$ Configuration:}
In this configuration, both $\mathcal P$ and $\mathcal R$ admit the same unique line of symmetry $\mathcal L$, and initially no parking node or robot position lies on $\mathcal L$. In this configuration, the \textit{Gathering Phase} is executed first in order to preserve the symmetry of robot movements so that the configuration $\mathcal C(t)$ remains symmetric throughout the execution. If the configuration $\mathcal C(t)$ becomes asymmetric at any instant of time, collisions may occur, which may lead to an unsolvable configuration due to the creation of multiple multiplicity nodes. After completing the \textit{Gathering Phase}, the \textit{Saturation Phase} starts. At the end of the \textit{Gathering Phase}, all the surplus robots are gathered at the node $\mathscr G$ located on $\mathcal L$. 

%------------------------------------------------
% C_222 Configuration
%------------------------------------------------

% %\FloatBarrier

\begin{algorithm}[ht]
\tiny
\caption{\textsc{: SaturationPhaseC222}$(\mathcal C(t))$}
\label{alg:saturation-c222}
\begin{algorithmic}[1]
\Require Symmetric configuration $\mathcal C(t_0)\in\mathcal C_{222}$, line of symmetry $\mathcal L$, gathering node $\mathscr G$
\Ensure All parking nodes are saturated while preserving reflection symmetry

\If{$\mathcal C(t_0)\notin\mathcal C_{222}$}
    \State \Return
\EndIf

\While{there exists an unsaturated parking node}

    \State Let $\mathcal {P^U}$ be the set of unsaturated parking nodes
    \State Partition $\mathcal {P^U}$ into symmetric pairs $(p,\mu(p))$ with respect to $\mathcal L$
    \State For each representative $p_i$, compute $d_i=d_{\mathcal M}(p_i,\mathscr G)$
    \State Let $d_{\max}=\max_{p_i\in\mathcal P_1} d_i$
    \State Select a symmetric parking-node pair $(p_i,\mu(p_i))$ attaining $d_{\max}$ using the ordering $\mathscr O_4$
    \State Let $\kappa=\kappa(p_i)=\kappa(\mu(p_i))$

    \While{$|\mathcal R(t)\cap\{p_i\}|<\kappa$ \textbf{or} $|\mathcal R(t)\cap\{\mu(p_i)\}|<\kappa$}

        \State Let $\mathcal R_u(t)$ be the set of unsaturated robots
        \State Partition $\mathcal R_u(t)$ into symmetric pairs $(r,\mu(r))$ with respect to $\mathcal L$
        \State For each representative $r_j$, compute $D_j=d_{\mathcal M}(r_j,p_i)$
        \State Let $D_{\min}=\min D_j$
        \State Select a symmetric robot pair $(r_j,\mu(r_j))$ attaining $D_{\min}$ using the ordering $\mathscr O_4$

        \If{one robot of the selected pair is in a pending move state}
             \State {Execute Algorithm~\ref{alg:handle-pending-move} for the pending robot until reflection symmetry is restored}
        \Else
            \State Execute \textsc{MoveToDestination}$(r_j,p_i)$
            \State Execute \textsc{MoveToDestination}$(\mu(r_j),\mu(p_i))$
        \EndIf

    \EndWhile

\EndWhile

\State \Return

\end{algorithmic}
\end{algorithm}

% %\FloatBarrier

Let $\mathcal {P^U}$ denote the set of unsaturated parking nodes in the infinite grid $(V,E)$. Initially, $\mathcal {P^U}=\mathcal P$; hence, $|\mathcal {P^U}|=m$. Since the configuration $\mathcal C(t)\in\mathcal C_{222}$ is symmetric with respect to the line of symmetry $\mathcal L$, and no parking node lies on $\mathcal L$, the set $\mathcal {P^U}$ can be partitioned into symmetric pairs,
$
\mathcal {P^U}=\mathcal P_1\cup\mu(\mathcal P_1),
$
where
$
\mathcal P_1=\{p_1,p_2,\ldots,p_{|\mathcal {P^U}|/2}\},
$
and $\mu(p_i)$ denotes the reflection of $p_i$ with respect to $\mathcal L$.

For each representative parking node $p_i\in\mathcal P_1$, let
$
d_i=d_{\mathcal M}(p_i,\mathscr G)
$
denote its Manhattan distance from the multiplicity node $\mathscr G$, where $\mathscr G\in\mathcal L$. Since reflection preserves the Manhattan distance,
$
d_{\mathcal M}(p_i,\mathscr G)
=
d_{\mathcal M}(\mu(p_i),\mathscr G),
$
each symmetric parking-node pair $(p_i,\mu(p_i))$ is uniquely associated with the common distance $d_i$. Let
$
d_{\max}=\max_{p_i\in\mathcal P_1} d_i.
$
Among all symmetric parking-node pairs attaining the distance $d_{\max}$, the ordering $\mathscr O_4$ uniquely selects one pair, say $(p_i,\mu(p_i))$. Let
$
\kappa=\kappa(p_i)=\kappa(\mu(p_i))
$
denote their common capacity.

Let $\mathcal R_u(t)\subseteq\mathcal R(t)$ denote the set of unsaturated robots. Since $\mathcal C(t)$ remains symmetric with respect to $\mathcal L$, the robots in $\mathcal R_u(t)$ also occur in symmetric pairs. Accordingly,
$
\mathcal R_u(t)=\mathcal R_1(t)\cup\mu(\mathcal R_1(t)),
$
where
$
\mathcal R_1(t)=\{r_1,r_2,\ldots,r_{|\mathcal R_u(t)|/2}\}.
$
For each representative robot $r_j\in\mathcal R_1(t)$, compute
$
D_j=d_{\mathcal M}(r_j,p_i).
$
By reflection symmetry,
$
d_{\mathcal M}(r_j,p_i)
=
d_{\mathcal M}(\mu(r_j),\mu(p_i)).
$
Let
$
D_{\min}=\min_{r_j\in\mathcal R_1(t)}D_j.
$
Among all symmetric robot pairs attaining this minimum distance, the ordering $\mathscr O_4$ uniquely selects one pair, say $(r_j,\mu(r_j))$.

If one robot of the selected pair is in a \emph{pending move} state, the procedure in Algorithm~\ref{alg:handle-pending-move} is executed for the pending robot while its symmetric partner and all other robots remain stationary. Once the pending movement is completed, the reflection symmetry of the configuration is restored. Otherwise, the robots $r_j$ and $\mu(r_j)$ simultaneously move toward the parking nodes $p_i$ and $\mu(p_i)$, respectively, using the procedure \textsc{MoveToDestination}. This process is repeated until both parking nodes reach their prescribed capacity $\kappa$.

After the selected parking-node pair becomes saturated, it is removed from $\mathcal {P^U}$. The algorithm continues by repeating these steps. It uses the $\mathscr{O}_4$ ordering to find the next pair of symmetric parking nodes that are not yet saturated, specifically considering the pair with the largest Manhattan distance from the multiplicity node $\mathscr{G}$. The process continues until every parking node is saturated. Throughout the execution, the reflection symmetry of the configuration is preserved, and strong multiplicity detection enables the robots to determine when a parking node has reached its prescribed capacity. Algorithm~\ref{alg:saturation-c222} describes the saturation phase for configuration $\mathcal C_{222}$ while preserving reflection symmetry.

\item \textbf{$\mathcal C_{322}$ Configuration:}
In this configuration, both the parking-node configuration $\mathcal P$ and the robot configuration $\mathcal R$ admit rotational symmetry with center of rotation $c$, and initially neither a parking node nor a robot occupies $c$. The \textit{Gathering Phase} is executed first to preserve the rotational symmetry of robot movements throughout the execution. If the configuration $\mathcal C(t)$ becomes asymmetric at any time, collisions may occur, leading to the creation of multiple multiplicity nodes and exhibit the problem unsolvable. Let $\mathcal {P^U}$ denote the set of unsaturated parking nodes in the infinite grid $(V,E)$. Initially,
$
\mathcal {P^U}=\mathcal P,
$
and hence $|\mathcal {P^U}|=m$. Since the configuration $\mathcal C(t)\in\mathcal C_{322}$ is rotationally symmetric with respect to the center of rotation $c$, the set $\mathcal {P^U}$ is partitioned into disjoint rotational orbits under the rotational map $\rho$. Let
$
\mathcal P_1=
\{p_1,p_2,\ldots,p_{|\mathcal {P^U}|/q}\}
$
be a set containing exactly one representative parking node from each rotational orbit. Then,
$
\mathcal {P^U}
=
\bigcup_{p_i\in\mathcal P_1}
\mathbf{Orbit}_{p_i},
$
where
$
\mathbf{Orbit}_{p_i}
=
\{p_i,\rho(p_i),\rho^2(p_i),\ldots,\rho^{q-1}(p_i)\}.
$

For each representative parking node $p_i\in\mathcal P_1$, let
$
d_i=d_{\mathcal M}(p_i,c)
$
denote its Manhattan distance from the center of rotation $c$. Since the rotational map $\rho$ preserves the Manhattan distance,
$
d_{\mathcal M}(p_i,c)
=
d_{\mathcal M}(\rho^j(p_i),c),$$
0\le j\le q-1.
$
Hence, every parking node in the orbit $\mathbf{Orbit}_{p_i}$ is associated with the common distance $d_i$. Let
$
d_{\max}
=
\max_{p_i\in\mathcal P_1} d_i.
$
Among all parking-node orbits attaining the maximum distance $d_{\max}$, the ordering $\mathscr O_4$ uniquely selects one orbit, denoted by $\mathbf{Orbit}_{p_i}$. Let
$
\kappa=\kappa(p_i)
$
denote the common capacity of every parking node in the selected orbit.

%------------------------------------------------
% C_322 Configuration
%------------------------------------------------

% %\FloatBarrier

\begin{algorithm}[ht]
\tiny
\caption{\textsc{: SaturationPhaseC322}$(\mathcal C(t_0))$}
\label{alg:saturation-c322}
\begin{algorithmic}[1]
\Require Initial configuration $\mathcal C(t_0)\in\mathcal C_{322}$, center of rotation $c$, rotational map $\rho$
\Ensure All parking nodes are saturated while preserving rotational symmetry

\If{$\mathcal C(t_0)\notin\mathcal C_{322}$}
    \State \Return
\EndIf

\While{there exists an unsaturated parking node}

    \State Let $\mathcal {P^U}$ be the set of unsaturated parking nodes
    \State Partition $\mathcal {P^U}$ into rotational orbits under $\rho$
    \State Let $\mathcal P_1$ be a set of representatives of the parking-node orbits
    \State For each $p_i\in\mathcal P_1$, compute
    $d_i=d_{\mathcal M}(p_i,c)$
    \State Let
    $d_{\max}=\displaystyle\max_{p_i\in\mathcal P_1}d_i$
    \State Select the parking-node orbit
    $\mathbf{Orbit}_{p_i}$
    attaining $d_{\max}$ using the ordering $\mathscr O_4$
    \State Let $\kappa\gets\kappa(p_i)$

    \While{some parking node in $\mathbf{Orbit}_{p_i}$ has occupancy less than $\kappa$}

        \State Let $\mathcal R_u(t)$ be the set of unsaturated robots
        \State Partition $\mathcal R_u(t)$ into rotational orbits under $\rho$
        \State Let $\mathcal R_1(t)$ be a set of representatives of the robot orbits
        \State For each $r_j\in\mathcal R_1(t)$, compute
        $D_j=d_{\mathcal M}(r_j,p_i)$
        \State Let
        $D_{\min}=\displaystyle\min_{r_j\in\mathcal R_1(t)}D_j$
        \State Select the robot orbit
        $\mathbf{Orbit}_{r_j}$
        attaining $D_{\min}$ using the ordering $\mathscr O_4$

        \If{some robot in $\mathbf{Orbit}_{r_j}$ is in a pending move state}
            \State Execute Algorithm~\ref{alg:handle-pending-move} for the pending robot until the rotational symmetry is restored
        \Else
            \For{$a=0$ to $q-1$}
                \State Execute \textsc{MoveToDestination}$(\rho^a(r_j),\rho^a(p_i))$
            \EndFor
        \EndIf

    \EndWhile

\EndWhile

\State \Return

\end{algorithmic}
\end{algorithm}

% %\FloatBarrier

Since $\rho^j(p_i)$ is obtained by rotating $p_i$ about the center of rotation $c$, and the rotational map $\rho$ preserves the Manhattan distance from $c$, we have
$
d_{\mathcal M}(p_i,c)
=
d_{\mathcal M}(\rho^j(p_i),c), \qquad
0\le j\le q-1.
$
Hence, every parking node in the orbit $\mathbf{Orbit}_{p_i}$ is associated with the common distance $d_i$. Let
$
d_{\max}
=
\max_{p_i\in\mathcal P_1} d_i.
$
Among all parking-node orbits attaining the maximum distance $d_{\max}$, the ordering $\mathscr O_4$ uniquely selects one orbit, denoted by $\mathbf{Orbit}_{p_i}$. Let
$
\kappa
$
be the common capacity of every parking node in the selected orbit.

Let $\mathcal R_u(t)\subseteq\mathcal R(t)$ denote the set of unsaturated robots at time $t$. Since the configuration $\mathcal C(t)\in\mathcal C_{322}$ is rotationally symmetric with respect to the center of rotation $c$, the set $\mathcal R_u(t)$ can be partitioned into disjoint rotational orbits under the rotational map $\rho$. Let
$
\mathcal R_1(t)=
\{r_1,r_2,\ldots,r_{|\mathcal R_u(t)|/q}\}
$
be a set containing exactly one representative robot from each rotational orbit. Then,
$
\mathcal R_u(t)
=
\bigcup_{r_i\in\mathcal R_1(t)}
\mathbf{Orbit}_{r_i},
$
where
$
\mathbf{Orbit}_{r_i}
=
\{r_i,\rho(r_i),\rho^2(r_i),\ldots,\rho^{q-1}(r_i)\}.
$ For each robot $r_i\in\mathcal R_u'$, let $p_i$ be the unsaturated parking node selected for $r_i$ according to the ordering $\mathscr O_4$. Define $D_i=d_{\mathcal M}(r_i,p_i)$. Similarly, for every rotational image $\rho^j(r_i)$ of $r_i$, the corresponding parking node is $\rho^j(p_i)$, and its Manhattan distance is
$
d_{\mathcal M}(\rho^j(r_i),\rho^j(p_i)),
\quad 0\le j\le q-1.
$

For each representative robot $r_i\in\mathcal R_1(t)$, let
$
D_i=d_{\mathcal M}(r_i,p_i)
$
denote the Manhattan distance between $r_i$ and the representative parking node $p_i$ of the selected parking-node orbit. Since $\rho^j(r_i)$ and $\rho^j(p_i)$ are obtained by rotating $r_i$ and $p_i$, respectively, about the center of rotation $c$, and the rotational map $\rho$ preserves the Manhattan distance, we have
$
d_{\mathcal M}(r_i,p_i)
=
d_{\mathcal M}(\rho^j(r_i),\rho^j(p_i))
\, 0\le j\le q-1.
$
Hence, every robot orbit is associated with the common distance $D_i$. Let
$
D_{\min}
=
\min_{r_i\in\mathcal R_1(t)}D_i.
$
Among all robot orbits attaining the minimum distance $D_{\min}$, the ordering $\mathscr O_4$ uniquely selects one orbit, denoted by $\mathbf{Orbit}_{r_i}$.

If one robot in the selected orbit is in a \emph{pending move} state, the procedure Algorithm~\ref{alg:handle-pending-move} is executed for the pending robot until the rotational symmetry of the configuration is restored. Otherwise, for every $0\le a\le q-1$, the robot $\rho^a(r_i)$ moves toward the parking node $\rho^a(p_i)$ using the procedure \textsc{MoveToDestination}. This process is repeated until every parking node in the selected orbit reaches its prescribed capacity $\kappa$.

After the selected parking-node orbit becomes saturated, it is removed from $\mathcal {P^U}$. The algorithm then repeats the above procedure by selecting another parking-node orbit attaining the maximum Manhattan distance from the center of rotation $c$, with ties broken according to the ordering $\mathscr O_4$. The process continues until every parking node is saturated. Throughout the execution, the rotational symmetry of the configuration is preserved, and strong multiplicity detection enables the robots to determine when a parking node reaches its prescribed capacity. The pseudo-code corresponding to this phase is given in Algorithm \ref{alg:saturation-c322}.
\end{itemize}

\paragraph{Gathering Phase:}
If $\mathcal C(t_0)\in \{\mathcal C_{222}, \mathcal C_{322}\}$, then this phase is executed as the initial phase of the algorithm \textbf{\textsc{spg()}}. For all other solvable configurations, the \textit{Gathering Phase} is executed only after the completion of the \textit{Saturation Phase}. Note that the problem $\mathcal {SPG}$ is unsolvable when the initial configuration satisfies $\mathcal C(t_0)\in \{\mathcal C_{223}, \mathcal C_{323}\}$. The \textit{Gathering Phase} is described in detail below.

\begin{itemize}

\item \textbf{$\mathcal C_{1}$ Configuration:}
In this configuration, the \textit{Gathering Phase} is executed only after the completion of \textit{Saturation Phase}. Since the set of parking nodes $\mathcal P$ is asymmetric, all parking nodes in $\mathcal P$ can be uniquely ordered. Consequently, a unique leading corner, say $\mathcal K_p$ of $\mathcal{MER}_{\mathcal P}$ can be determined. Let $\mathcal{STR}^i$ denote the string representation associated with the corner $\mathcal K_p$. The gathering node, denoted by $\mathscr G$, is chosen as the unique node outside $\mathcal{MER}_{\mathcal P}$ such that:
\begin{enumerate}
\item $\mathscr G$ is adjacent to the corner $\mathcal K_p$, and
\item there exists a ray (grid-line) originating from $\mathscr G$ and lying outside $\mathcal{MER}_{\mathcal P}$ such that the ray is parallel to the string direction $\mathcal{STR}^i$.
\end{enumerate}

% %\FloatBarrier

\begin{algorithm}[ht]
\tiny
\caption{\textsc{GatheringPhaseC1}$(\mathcal C(t_0))$}
\label{alg:gathering-c1}
\begin{algorithmic}[1]

\Require Initial configuration $\mathcal C(t_0)\in\mathcal C_1$, parking-node set $\mathcal P$, surplus robot set $\mathcal R_s(t)$
\Ensure All surplus robots gather at the unique gathering node $\mathscr G$

\If{$\mathcal C(t_0)\notin\mathcal C_1$}
    \State \Return
\EndIf

\If{the \textit{Saturation Phase} is not completed}
    \State \Return
\EndIf

\State Determine the unique leading corner $\mathcal K_p$ of $\mathcal{MER}_{\mathcal P}$
\State Let $\mathcal{STR}^i$ be the string representation associated with $\mathcal K_p$
\State Select the unique gathering node $\mathscr G$ outside $\mathcal{MER}_{\mathcal P}$ adjacent to $\mathcal K_p$
\State Let $\mathcal L_g$ be the ray originating at $\mathscr G$, lying outside $\mathcal{MER}_{\mathcal P}$, and parallel to the direction of $\mathcal{STR}^i$

\While{there exists a surplus robot $r\in\mathcal R_s(t)$ such that $r\neq\mathscr G$}

    \State For each surplus robot $r_i\in\mathcal R_s(t)$, compute
    $\Delta_i=d_{\mathcal M}(r_i,\mathscr G)$
    
    \State Let
    $\Delta_{\min}=\displaystyle\min_{r_i\in\mathcal R_s(t)}\Delta_i$
    
    \State Select a surplus robot
    $r\in\mathcal R_s(t)$ satisfying
    $d_{\mathcal M}(r,\mathscr G)=\Delta_{\min}$,
    {breaking ties using ordering on $\mathcal C(t)$}
    
    \State Execute \textsc{MoveToDestination}$(r,\mathscr G)$
    
    \State Update the surplus robot set $\mathcal R_s(t)$

\EndWhile

\State \Return

\end{algorithmic}
\end{algorithm}

% %\FloatBarrier

Consider the surplus robot set
$
\mathcal R_s(t)=\{r_1(t),r_2(t),\ldots,r_s(t)\}
\subseteq
\mathcal R(t).
$
For each surplus robot $r_i\in\mathcal R_s(t)$, let
$\Delta_i=d_{\mathcal M}(r_i,\mathscr G)
$. Define
$
\Delta_{\min}
=
\min_{r_i\in\mathcal R_s(t)}
\Delta_i.
$
Among all surplus robots attaining the minimum distance $\Delta_{\min}$, one robot $r_i$ is selected (ties are broken using ordering on $\mathcal C(t)$). The selected robot moves toward $\mathscr G$ using the procedure \textsc{MoveToDestination}. After the movement is completed, the surplus robot set is updated and the same procedure is repeated. The \textit{Gathering Phase} terminates when every surplus robot reaches the gathering node $\mathscr G$. The pseudo-code corresponding to this phase is given in Algorithm \ref{alg:gathering-c1}.

\item \textbf{$\mathcal C_{21}$ Configuration:}
For this configuration, the \textit{Gathering Phase} is executed only after the completion of the \textit{Saturation Phase} (see Figure~\ref{Gathering}(a))
. At this stage, all robots lie on the line $\mathscr L$ obtained during the \textit{Line Formation Phase}, and the multiplicity node $\mathbf{N}_{\mathscr L}$ has been established during the \textit{Multiplicity Creation Phase}. The objective of this phase is to move all surplus robots to $\mathbf{N}_{\mathscr L}$.

% %\FloatBarrier

\begin{algorithm}[ht]
\tiny
\caption{\textsc{: GatheringPhaseC21  C221}$(\mathcal C(t))$}
\label{alg:gathering-c21-c221}
\begin{algorithmic}[1]
\Require Current configuration $\mathcal C(t)\in\{\mathcal C_{21},\mathcal C_{221}\}$, formation line $\mathscr L$, multiplicity node $\mathbf N_{\mathscr L}$
\Ensure All surplus robots gather at the multiplicity node $\mathbf N_{\mathscr L}$

\If{$\mathcal C(t_0)\notin\{\mathcal C_{21},\mathcal C_{221}\}$}
\State \Return
\EndIf

\If{the \textit{Saturation Phase} is not completed}
\State \Return
\EndIf

\If{$\mathcal C(t_0)\in\mathcal C_{221}$}
\State The terminal robot $r_{e_1}$ breaks the symmetry during the \textit{Saturation Phase}
\State The configuration $\mathcal C(t)$ is transformed into a $\mathcal C_{21}$ configuration
\EndIf

\State Let $\mathcal R_s(t)=\{r_1(t),r_2(t),\ldots,r_s(t)\}$ be the set of surplus robots
\State All robots in $\mathcal R_s(t)$ lie on the formation line $\mathscr L$

\While{there exists a surplus robot $r_i\in\mathcal R_s(t)$ such that $r_i\neq \mathbf N_{\mathscr L}$}
\State For each surplus robot $r_i\in\mathcal R_s(t)$, compute $\Delta_i=d_{\mathcal M}(r_i,\mathbf N_{\mathscr L})$
\State Let $\Delta_{\min}=\min\{\Delta_i:r_i\in\mathcal R_s(t)\}$

\If{$\Delta_{\min}=0$}
    \State Every surplus robot already located at $\mathbf N_{\mathscr L}$ remains stationary
\EndIf

\State Select a surplus robot $r\in\mathcal R_s(t)$ satisfying $d_{\mathcal M}(r,\mathbf N_{\mathscr L})=\Delta_{\min}$
\State If more than one surplus robot attains $\Delta_{\min}$, break ties using the fixed ordering
\State Robot $r$ moves along $\mathscr L$ toward $\mathbf N_{\mathscr L}$ using \textsc{MoveToDestination}$(r,\mathbf N_{\mathscr L})$
\State All other surplus robots remain stationary during this move
\State Update the surplus robot set $\mathcal R_s(t)$

\EndWhile

\State \Return

\end{algorithmic}
\end{algorithm}

% %\FloatBarrier

Consider the surplus robot set $
\mathcal R_s(t)=\{r_1(t),r_2(t),\ldots,r_s(t)\}\subset \mathcal R(t)$ on $\mathscr L$.
Since every surplus robot lies on $\mathscr L$, for each robot
$r_i\in\mathcal R_s(t)$, define
$
\Delta_i=d_{\mathcal M}(r_i,\mathbf{N}_{\mathscr L}),
$
where $d_{\mathcal M}(r_i,\mathbf{N}_{\mathscr L})$ denotes the Manhattan distance between the robot $r_i$ and the gathering node $\mathbf{N}_{\mathscr L}$.

\begin{figure}[htbp]
    \centering

    \begin{subfigure}[t]{0.35\textwidth}
        \centering
       \includegraphics[width=\textwidth]{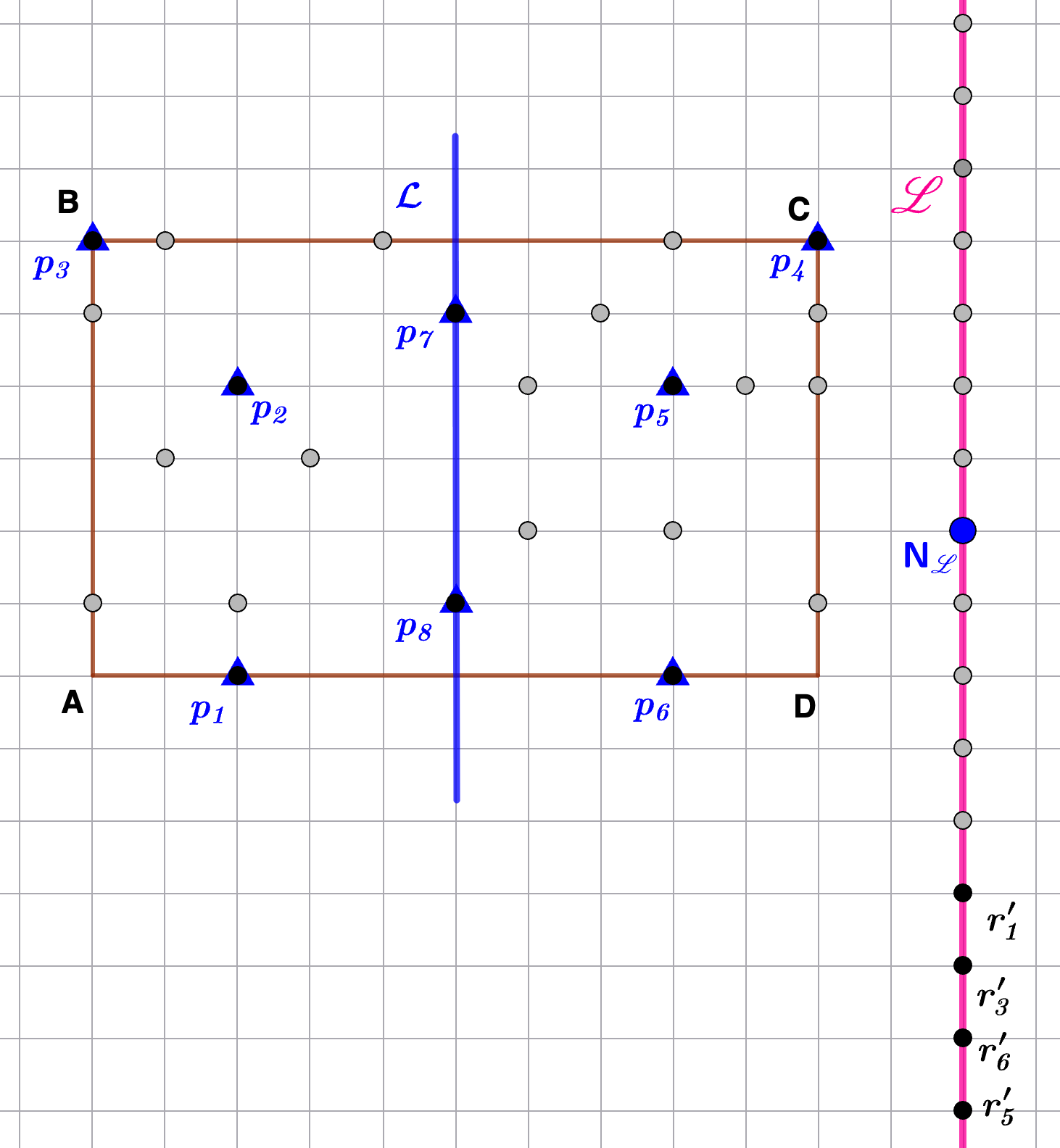}
       \caption{}
    \end{subfigure}
  \hspace{0.1\textwidth}%
    \begin{subfigure}[t]{0.35\textwidth}
        \centering
      \includegraphics[width=\textwidth]{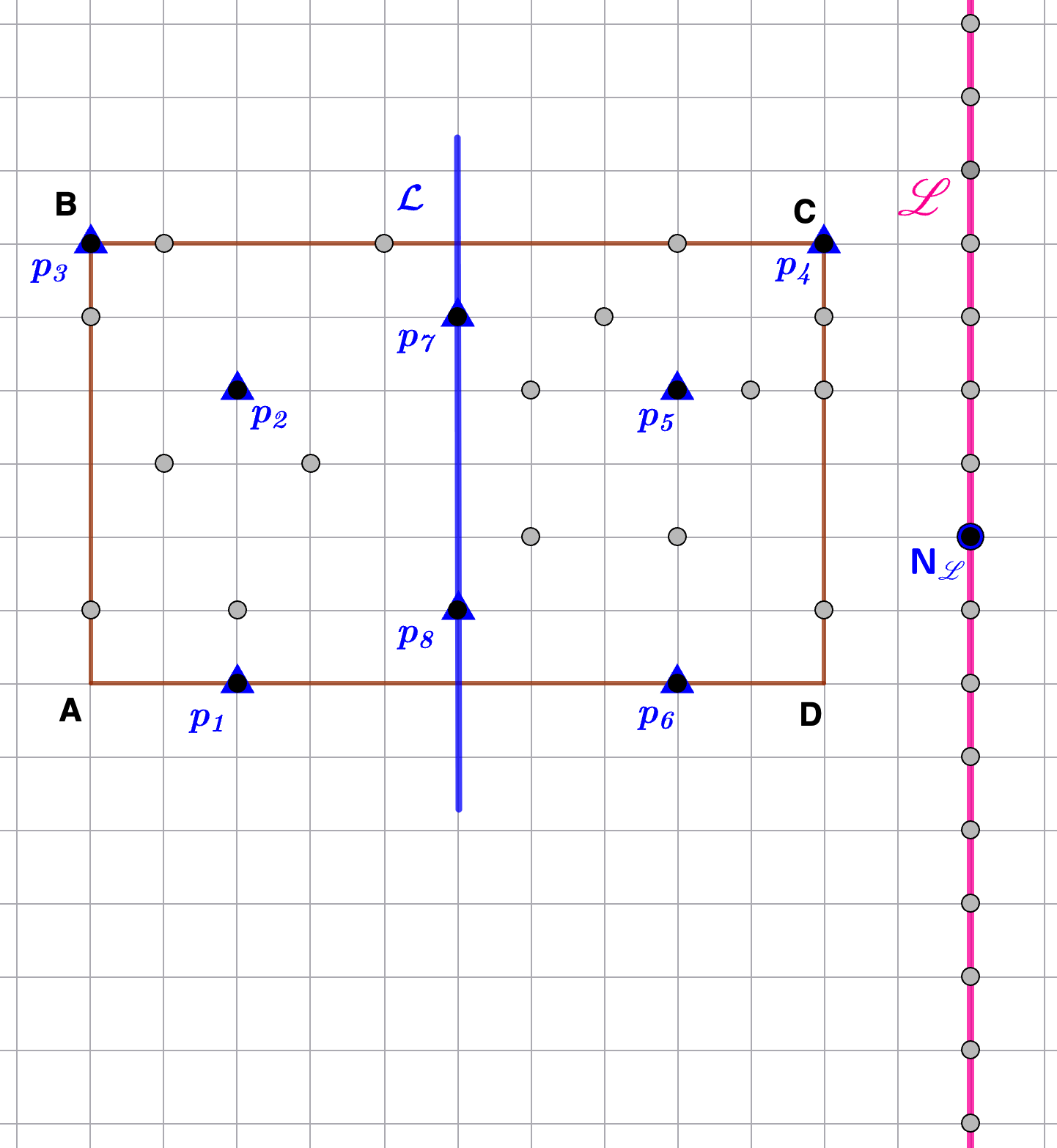}
        \caption{}
    \end{subfigure}

\caption{
Illustration of the final stage of the algorithm for the configuration $\mathcal C_{21}$: (a) all parking nodes are saturated; (b) the surplus robots gather at the node $\mathbf N_{\mathscr L}$.
}
    \label{Gathering}
\end{figure}

Assume
$
\Delta_{min}=\min\{\Delta_i: r_i\in\mathcal R_s(t)\}.
$
The surplus robot whose distance from $\mathbf{N}_{\mathscr L}$ is $\Delta_{\min}$ is selected to move first. If $\Delta_{\min}=0$, then a robot $r_i$ is already located at $\mathbf{N}_{\mathscr L}$ and hence remains stationary. If multiple robots attain the same minimum distance $\Delta_{min}$, ties are broken using a fixed ordering~\cite{bose2020arbitrary}, depending on the configuration of the robot positions. Without loss of generality, let the robots be ordered such that $
\Delta_1\le \Delta_2\le \cdots \le \Delta_s.
$
According to this ordering (ties are broken using a fixed ordering~\cite{bose2020arbitrary}, depending on the configuration of the robot positions), all surplus robots sequentially move along $\mathscr L$ towards $\mathbf{N}_{\mathscr L}$ and gather at the node $\mathbf{N}_{\mathscr L}$. This phase terminates when all surplus robots reach $\mathbf{N}_{\mathscr L}$. (see Figure~\ref{Gathering}(b)). The pseudo-code corresponding to this phase is given in Algorithm \ref{alg:gathering-c21-c221}.

\item \textbf{$\mathcal C_{221}$ Configuration:}
During the \textit{Saturation Phase}, the terminal robot $r_{e_1}$ breaks the symmetry of the configuration, transforming it to a $\mathcal C_{21}$ configuration. The subsequent \textit{Gathering Phase} is then executed exactly as described for $\mathcal C_{21}$. The pseudo-code corresponding to this phase is given in Algorithm \ref{alg:gathering-c21-c221}.

\item \textbf{$\mathcal C_{222}$ Configuration:}
In this configuration, both $\mathcal P$ and $\mathcal R$ admit the same unique line of symmetry $\mathcal L$, and neither a parking node nor a robot position lies on $\mathcal L$. To preserve the symmetry of the configuration, the \textit{Gathering Phase} is executed before the \textit{Saturation Phase}. This ensures that the symmetry of $\mathcal C(t)$ is maintained throughout the execution. Breaking the symmetry may lead to collisions among robots, resulting in the creation of multiple multiplicity nodes and consequently an unsolvable configuration.

Since every robot knows the total capacity of the parking nodes in $\mathcal P$, it can determine the number of surplus robots. As the total number of robots exceeds the sum of the capacities of all parking nodes by at least two, there exist surplus robots that can participate in the \textit{Gathering Phase}. Hence, the \textit{Gathering Phase} is well defined for the initial configuration $\mathcal C(t_0)$.

Consider the surplus robot set $\mathcal R_s\subseteq\mathcal R$, where
$|\mathcal R_s|=d$ and $d$ is even. Define
$
\mathcal S_{\mathcal L}=\mathcal L\cap\mathcal{MER}_{\mathcal R}
$
to be the line segment of the symmetry line $\mathcal L$ contained in
$\mathcal{MER}_{\mathcal R}$. Equivalently, $\mathcal S_{\mathcal L}$ is the sub-grid containing grid nodes on $\mathcal L$ that lie within
$\mathcal{MER}_{\mathcal R}$.

Now define
$
V_{\mathcal L}=\mathcal S_{\mathcal L}\cap V.
$
Thus, $V_{\mathcal L}$ denotes the set of all grid nodes lying on the line segment $\mathcal S_{\mathcal L}$. Since, the configuration $\mathcal C(t)\in\mathcal C_{222}$ is symmetric with respect to the line of symmetry $\mathcal L$, the robots in $\mathcal R_s$ appear in symmetric pairs with respect to $\mathcal L$. Therefore, $\mathcal R_s$ can be partitioned into two subsets $\mathcal R_s'$ and $\mu(\mathcal R_s')$, where
$
\mathcal R_s'=\{r_1,r_2,\ldots,r_{\frac d2}\}
$
and
$
\mu(\mathcal R_s')=\{\mu(r_1),\mu(r_2),\ldots,\mu(r_{\frac d2})\}.
$
Hence,
$
\mathcal R_s=\mathcal R_s'\cup \mu(\mathcal R_s').
$
Here, $\mu(r_i)$ denotes the image of $r_i$ under the reflection map $\mu$ (See Definition~\ref{def:reflection-map}) with respect to $\mathcal L$. Consider
$
\Delta_i=d_{\mathcal M}(r_i,V_{\mathcal L}),\ \text{where,}\ r_i\in \mathcal R_s'.
$
Similarly, the symmetric image $\mu(r_i)$ of $r_i$ is associated with the symmetric parking node $\mu(p_i)$, and its corresponding Manhattan distance is
$
d_{\mathcal M}(\mu(r_i),V_{\mathcal L}).
$
Since $\mu(r_i)$ is the reflection of $r_i$ with respect to the line of symmetry $\mathcal L$, and since reflection preserves the grid Manhattan distance, we have
$
d_{\mathcal M}(r_i,V_{\mathcal L})=d_{\mathcal M}(\mu(r_i),V_{\mathcal L}).
$

Assume
$
\Delta_{min}=\min\{\Delta_i:, r_i\in \mathcal R_s'\}.
$
For the fixed Manhattan distance value $\Delta_{\min}$, the number of robots in $\mathcal R_s'$ at distance $\Delta_{\min}$ may be either one or more than one.

If exactly one robot, say $r_i$, is at distance $\Delta_{\min}$, then it determines a unique symmetric pair of robots
$
(r_i,\mu(r_i)),
$
which in turn uniquely determines a gathering node, say
$
\mathscr G\in V_{\mathcal L}.
$ Otherwise, multiple symmetric pairs of robots correspond to the distance value $\Delta_{\min}$. In this case, the ordering $\mathscr O_4$ is used to uniquely identify one such symmetric pair, say $
(r_i,\mu(r_i)),
$
which consequently determines a unique gathering node
$
\mathscr G\in V_{\mathcal L}.
$

This pair $(r_i,\mu(r_i))$ of robots initiates the \textit{Gathering Phase}. Due to the asynchronous nature of the robot activations, one of the robots may be in a \textit{pending move} state. Our aim is to preserve the symmetry of the configuration in order to avoid collisions among robots. Using the Algorithm~\ref{alg:handle-pending-move}, a robot in the \textit{pending move} state can always be identified. Whenever a robot in the \textit{pending move} state is identified, its paired robot remains stationary until the pending robot reaches a position where the symmetry of the configuration is restored, and collision is avoided. Thus, the configuration remains symmetric. 

The robots in $\mathcal R_s'$ are ordered according to the ordering $\mathscr O'$ such that $
d_{\mathcal M}(r_1,\mathscr G)\le d_{\mathcal M}(r_2,\mathscr G)\le \cdots \le d_{\mathcal M}(r_{\frac{d}{2}},\mathscr G).
$ By symmetry, this induces the same ordering on their symmetric counterparts:
, i.e., $
d_{\mathcal M}(\mu(r_1),\mathscr G)\le d_{\mathcal M}(\mu(r_2),\mathscr G)\le \cdots \le d_{\mathcal M}(\mu(r_{\frac{d}{2}}),\mathscr G).
$

 % %\FloatBarrier

\begin{algorithm}[ht]
\tiny
\caption{\textsc{: GatheringPhaseC222}$(\mathcal C(t))$}
\label{alg:gathering-c222}
\begin{algorithmic}[1]
\Require Symmetric configuration $\mathcal C(t)\in\mathcal C_{222}$, line of symmetry $\mathcal L$, surplus robot set $\mathcal R_s(t)$
\Ensure All surplus robots gather at a unique node $\mathscr G\in\mathcal L$

\If{$\mathcal C(t_0)\notin\mathcal C_{222}$}
\State \Return
\EndIf

\State Let $\mathcal S_{\mathcal L}=\mathcal L\cap\mathcal{MER}_{\mathcal R}$
\State Let $V_{\mathcal L}=\mathcal S_{\mathcal L}\cap V$
\State Partition $\mathcal R_s(t)$ into symmetric pairs with respect to $\mathcal L$
\State Let $\mathcal R_s'=\{r_1,r_2,\ldots,r_{\frac d2}\}$ contain one representative from each symmetric pair
\State $\mathcal R_s(t)=\mathcal R_s'\cup\mu(\mathcal R_s')$

\For{each robot $r_i\in\mathcal R_s'$}
\State Compute $\Delta_i=d_{\mathcal M}(r_i,V_{\mathcal L})$
\EndFor

\State Let $\Delta_{\min}=\min\{\Delta_i:r_i\in\mathcal R_s'\}$

\If{exactly one robot $r_i\in\mathcal R_s'$ satisfies $\Delta_i=\Delta_{\min}$}
\State Select the symmetric robot pair $(r_i,\mu(r_i))$
\Else
\State Select a unique symmetric robot pair $(r_i,\mu(r_i))$ attaining $\Delta_{\min}$ using the ordering $\mathscr O_4$
\EndIf

\State Let $\mathscr G\in V_{\mathcal L}$ be the unique gathering node determined by the selected pair $(r_i,\mu(r_i))$

\While{$|\mathcal R_s(t)\cap{\mathscr G}|<2$}
\If{one robot of the selected pair $(r_i,\mu(r_i))$ is in a pending move state}
\State The paired robot remains stationary
\State Execute Algorithm~\ref{alg:handle-pending-move} for the pending robot until the reflection symmetry is restored
\Else
\State Robot $r_i$ moves toward $\mathscr G$ using \textsc{MoveToDestination}$(r_i,\mathscr G)$
\State Robot $\mu(r_i)$ moves toward $\mathscr G$ using \textsc{MoveToDestination}$(\mu(r_i),\mathscr G)$
\EndIf
\EndWhile

\State A multiplicity node is created at $\mathscr G$

\While{there exists a surplus robot $r\in\mathcal R_s(t)$ such that $r\neq\mathscr G$}
\State Partition the robots in $\mathcal R_s(t)\setminus(\mathcal R_s(t)\cap{\mathscr G})$ into symmetric pairs
\State Let $\mathcal R_s''$ contain one representative from each remaining symmetric pair

\For{each robot $r_j\in\mathcal R_s''$}
    \State Compute $D_j=d_{\mathcal M}(r_j,\mathscr G)$
\EndFor

\State Let $D_{\min}=\min\{D_j:r_j\in\mathcal R_s''\}$
\State Select a symmetric pair $(r_j,\mu(r_j))$ attaining $D_{\min}$
\State Break ties using the ordering $\mathscr O_4$

\If{one robot of the selected pair $(r_j,\mu(r_j))$ is in a pending move state}
    \State The paired robot remains stationary
    \State Execute Algorithm~\ref{alg:handle-pending-move} for the pending robot until the reflection symmetry is restored
\Else
    \State Robot $r_j$ moves toward $\mathscr G$ using \textsc{MoveToDestination}$(r_j,\mathscr G)$
    \State Robot $\mu(r_j)$ moves toward $\mathscr G$ using \textsc{MoveToDestination}$(\mu(r_j),\mathscr G)$
\EndIf

\State Using strong multiplicity detection, update the number of surplus robots gathered at $\mathscr G$

\EndWhile

\State \Return

\end{algorithmic}
\end{algorithm}

%\FloatBarrier

The symmetric pair $(r_i,\mu(r_i))$ that initiates the \textit{Gathering Phase} moves toward $\mathscr G$ and creates a multiplicity node at $\mathscr G$. Once the multiplicity node at $\mathscr G$ is created, all surplus robots in $\mathcal R_s$ detect this node and move sequentially toward $\mathscr G$ according to the ordering $\mathscr O'$ (ties are broken using ordering $\mathscr O_4)$. Using the \textit{strong multiplicity detection} capability, the robots in $\mathcal R$ can determine the exact number of surplus robots gathered at $\mathscr G$. This phase terminates when all surplus robots reach $\mathscr G$. Note that a configuration $\mathcal C_{222}$ transforms into a configuration of type $\mathcal C_{221}$ after the completion of the \textit{Gathering} phase. However, since the robots have strong multiplicity detection capability, they can identify this transition. Therefore, after the \textit{Gathering} phase, the robots proceed with the \textit{Saturation} phase using the strategy prescribed for configurations of type $\mathcal C_{222}$. The pseudo-code corresponding to this phase is given in Algorithm \ref{alg:gathering-c222}.

\item \textbf{$\mathcal C_{31}$ Configuration:}
The set of parking nodes $\mathcal P$ admits rotational symmetry with center of rotation $c$, while the robot configuration $\mathcal R(t_0)$ is asymmetric. For this configuration, the \textit{Gathering Phase} is executed only after the completion of the \textit{Saturation Phase}. In \textit{Line Formation Phase}, all robots are on the $\mathscr L$ and in \textit{Multiplicity Creation Phase}, the multiplicity node $\mathbf{N}_{c}$ on $\mathscr L$ is created, hence all surplus robots gather at $\mathbf{N}_{c}$.

\begin{algorithm}[ht]
\tiny
\caption{\textsc{: GatheringPhaseC31C321}$(\mathcal C(t))$}
\label{alg:gathering-c31-c321}
\begin{algorithmic}[1]
\Require Current configuration $\mathcal C(t)\in\{\mathcal C_{31},\mathcal C_{321}\}$, formation line $\mathscr L$, multiplicity node $\mathbf N_c$
\Ensure All surplus robots gather at the multiplicity node $\mathbf N_c$

\If{$\mathcal C(t_0)\notin\{\mathcal C_{31},\mathcal C_{321}\}$}
\State \Return
\EndIf

\If{the \textit{Saturation Phase} is not completed}
\State \Return
\EndIf

\If{$\mathcal C(t_0)\in\mathcal C_{321}$}
\State Let $r_c$ be the robot located at the center of rotation $c$
\State During the \textit{Saturation Phase}, robot $r_c$ performs the symmetry-breaking operation
\State The configuration $\mathcal C(t)$ is transformed into a $\mathcal C_{31}$ configuration
\EndIf

\State Let $\mathcal R_s(t)=\{r_1(t),r_2(t),\ldots,r_s(t)\}$ be the set of surplus robots
\State All robots in $\mathcal R_s(t)$ lie on the formation line $\mathscr L$

\While{there exists a surplus robot $r_i\in\mathcal R_s(t)$ such that $r_i\neq \mathbf N_c$}
\State For each surplus robot $r_i\in\mathcal R_s(t)$, compute $\Delta_i=d_{\mathcal M}(r_i,\mathbf N_c)$
\State Let $\Delta_{\min}=\min\{\Delta_i:r_i\in\mathcal R_s(t)\}$

\If{$\Delta_{\min}=0$}
    \State Every surplus robot already located at $\mathbf N_c$ remains stationary
\EndIf

\State Select a surplus robot $r\in\mathcal R_s(t)$ satisfying $d_{\mathcal M}(r,\mathbf N_c)=\Delta_{\min}$
\State If more than one surplus robot attains $\Delta_{\min}$, break ties using the fixed ordering
\State Robot $r$ moves along $\mathscr L$ toward $\mathbf N_c$ using \textsc{MoveToDestination}$(r,\mathbf N_c)$
\State All other surplus robots remain stationary during this move
\State Update the surplus robot set $\mathcal R_s(t)$

\EndWhile

\State \Return

\end{algorithmic}
\end{algorithm}

Consider the surplus robot set $
\mathcal R_s(t)=\{r_1(t),r_2(t),\ldots,r_s(t)\}\subset \mathcal R(t).
$
Here, all robots in $\mathcal R_s$ lie on the line $\mathscr L$ obtained after the completion of the \textit{Line Formation Phase}. For each robot $r_i\in\mathcal R_s(t)$, define
$
\Delta_i=d_{\mathcal M}(r_i,\mathbf{N}_{c}),
$
where $d_{\mathcal M}(r_i,\mathbf{N}_{c})$ denotes the Manhattan distance between the robot $r_i$ and the gathering node $\mathbf{N}_{c}$.

Assume
$\Delta_{min}=\min\{\Delta_i: r_i\in\mathcal R_s(t)\}.
$
The surplus robot whose distance from $\mathbf{N}_{c}$ is $\Delta_{\min}$ is selected to move first. If $\Delta_{\min}=0$, then a robot $r_i$ is already located at $\mathbf{N}_{c}$ and hence remains stationary. If multiple robots attain the same minimum distance $\Delta_{min}$, ties are broken using a fixed ordering~\cite{bose2020arbitrary}, depending on the configuration of the robot positions. Without loss of generality, let the robots be ordered such that $
\Delta_1\le \Delta_2\le \cdots \le \Delta_s.
$
According to this ordering (ties are broken using a fixed ordering~\cite{bose2020arbitrary}, depending on the configuration of the robot positions), all surplus robots sequentially move along $\mathscr L$ towards $\mathbf{N}_{c}$ and gather at the node $\mathbf{N}_{c}$. This phase terminates when all surplus robots reach $\mathbf{N}_{c}$. The pseudo-code corresponding to this phase is given in Algorithm \ref{alg:gathering-c31-c321}.

\item \textbf{$\mathcal C_{321}$ Configuration:}  
In this configuration, a robot $r_c$ is located at the center of rotation $c$. During the \textit{Saturation Phase}, the robot $r_c$ performs the symmetry-breaking operation, thereby transforming the configuration into $\mathcal C_{31}$ configuration. Thereafter, the \textit{Gathering Phase} is executed in the same manner as described for $\mathcal C_{31}$. The pseudo-code corresponding to this phase is given in Algorithm \ref{alg:gathering-c31-c321}.

% %\FloatBarrier

% %\FloatBarrier

\item \textbf{$\mathcal C_{322}$ Configuration:}
In this configuration, both $\mathcal P$ and $\mathcal R$ admit rotational symmetry with center of rotation $c$, and initially no parking node or robot position lies at $c$. For this configuration, the \textit{Gathering Phase} is executed first in order to preserve the symmetry of robot movements so that the configuration $\mathcal C(t)$ remains symmetric throughout the execution. If the configuration $\mathcal C(t)$ becomes asymmetric at any instant of time, collisions may occur, making the problem unsolvable due to the creation of multiple multiplicity nodes. Consider the surplus robot set $\mathcal R_s\subseteq\mathcal R$, where $|\mathcal R_s|=d$ and $d$ is divisible by $q$, $q$ is the order of rotational symmetry. Since the configuration $\mathcal C(t)\in\mathcal C_{322}$ is rotationally symmetric with respect to the center of rotational symmetry $c$, the robots in $\mathcal R_s$ can be partitioned into disjoint rotational orbits under the rotational map $\rho$ (See Definition~\ref{def:rotational-map}).

For each robot $r_i\in\mathcal R_s$, define its rotational orbit as
$
\mathbf{Orbit}_{r_i}=\{r_i,\rho(r_i),\rho^2(r_i),\ldots,\rho^{q-1}(r_i)\},
$
where $q\ge 2$ denotes the order of rotational symmetry. Let
$
\mathcal R_s'=\{r_1,r_2,\ldots,r_{\frac{d}{q}}\}
$
be a set containing exactly one representative robot from each rotational orbit. Then, 
$
\mathcal R_s=\bigcup_{r_i\in\mathcal R_s'}\mathbf{Orbit}_{r_i}.
$

Define $\Delta_i=d_{\mathcal M}(r_i,c)$. Similarly, for every rotational image $\rho^j(r_i)$ of $r_i$ with respect to $c$, and its Manhattan distance is
$
d_{\mathcal M}(\rho^j(r_i),c),
\quad 0\le j\le q-1.
$ Since $\rho^j(r_i)$ is obtained by rotating $r_i$ about the center of rotational symmetry $c$, and since rotation preserves the grid Manhattan distance, we have $d_{\mathcal M}(r_i,c)=d_{\mathcal M}(\rho^j(r_i),c)$ for every $0\le j\le q-1$.

Assume
$
\Delta_{min}=\min\{\Delta_i: r_i\in\mathcal R_s'\}.
$
For the fixed Manhattan distance value $\Delta_{min}$, the number of representative robots in $\mathcal R_s'$ at distance $\Delta_{min}$ may be either one or more than one. If exactly one representative robot, say $r_i$, is at distance $\Delta_{min}$, then it determines a unique rotational orbit $\mathbf{Orbit}_{r_i}$. Otherwise, multiple rotational orbits correspond to the distance value $\Delta_{min}$. In this case, the ordering $\mathscr O_4$ is used to uniquely identify one such rotational orbit, say $\mathbf{Orbit}_{r_i}$. This rotational orbit $\mathbf{Orbit}_{r_i}=\{r_i,\rho(r_i),\rho^2(r_i),\ldots,\rho^{q-1}(r_i)\}$ initiates the \textit{Gathering Phase}. Due to the asynchronous nature of robot activations, one or more robots in the orbit may be in a \textit{pending move} state. Our aim is to preserve the rotational symmetry of the configuration in order to avoid collisions among robots. Using the Algorithm~\ref{alg:handle-pending-move}, a robot in the \textit{pending move} state can always be identified. Whenever such a robot is identified, all other robots belonging to the same rotational orbit remain stationary until the pending robot reaches a position where the rotational symmetry of the configuration is restored and collision is avoided. Thus, the configuration remains rotationally symmetric throughout the execution.

The representative robots in $\mathcal R_s'$ are ordered according to the ordering, say $\mathscr O''$ such that
$
d_{\mathcal M}(r_1,c)\le d_{\mathcal M}(r_2,c)\le \cdots \le d_{\mathcal M}(r_{\frac{d}{q}},c).
$ By symmetry, this induces the same ordering on the corresponding rotational orbits. The rotational orbit
$
\mathbf{Orbit}_{r_i}=\{r_i,\rho(r_i),\rho^2(r_i),\ldots,\rho^{q-1}(r_i)\}
$
that initiates the \textit{Gathering Phase} moves toward the center of rotation $c$ and reaches $c$. All remaining surplus robots in $\mathcal R_s$ detect the center of rotation $c$ and move orbit by orbit toward $c$ according to the ordering $\mathscr O''$, where ties are broken using the ordering $\mathscr O_4$. Using the \textit{strong multiplicity detection} capability, the robots in $\mathcal R$ can determine the exact number of surplus robots gathered at $c$. This phase terminates when all surplus robots in $\mathcal R_s$ reach $c$. Note that a configuration $\mathcal C_{322}$ transforms into a configuration of type $\mathcal C_{321}$ after the completion of the \textit{Gathering} phase. However, since the robots have strong multiplicity detection capability, they can identify this transition. Therefore, after the \textit{Gathering} phase, the robots proceed with the \textit{Saturation} phase using the strategy prescribed for configurations of type $\mathcal C_{322}$. The pseudo-code corresponding to this phase is given in Algorithm \ref{alg:gathering-c322}.

\end{itemize}
Note that the movement of robots for each configuration during the \textit{Gathering Phase} is performed by the algorithm \textit{MoveToDestination()}.

%\FloatBarrier

\begin{algorithm}[ht]
\tiny
\caption{\textsc{: GatheringPhaseC322}$(\mathcal C(t))$}
\label{alg:gathering-c322}
\begin{algorithmic}[1]
\Require Rotationally symmetric configuration $\mathcal C(t)\in\mathcal C_{322}$, center of rotation $c$, rotational map $\rho$, surplus robot set $\mathcal R_s(t)$
\Ensure All surplus robots gather at the center of rotation $c$

\If{$\mathcal C(t_0)\notin\mathcal C_{322}$}
\State \Return
\EndIf

\State Let $q$ be the order of rotational symmetry
\State Let $\mathcal R_s(t)$ be the set of surplus robots, where $|\mathcal R_s(t)|=d$ and $q$ divides $d$
\State Partition $\mathcal R_s(t)$ into disjoint rotational orbits under the rotational map $\rho$
\State Let $\mathcal R_s'=\{r_1,r_2,\ldots,r_{\frac{d}{q}}\}$ contain exactly one representative robot from each rotational orbit

\For{each representative robot $r_i\in\mathcal R_s'$}
\State Compute $\Delta_i=d_{\mathcal M}(r_i,c)$
\EndFor

\State Let $\Delta_{\min}=\min\{\Delta_i:r_i\in\mathcal R_s'\}$

\If{exactly one representative robot $r_i\in\mathcal R_s'$ satisfies $\Delta_i=\Delta_{\min}$}
\State Select the rotational orbit $\mathbf{Orbit}_{r_i}$
\Else
\State Select a unique rotational orbit $\mathbf{Orbit}_{r_i}$ attaining $\Delta_{\min}$ using the ordering $\mathscr O_4$
\EndIf

\While{not all robots in $\mathbf{Orbit}_{r_i}$ have reached $c$}
\If{some robot in $\mathbf{Orbit}_{r_i}$ is in a pending move state}
\State All other robots in $\mathbf{Orbit}_{r_i}$ remain stationary
\State Execute Algorithm~\ref{alg:handle-pending-move} for the pending robot until rotational symmetry is restored
\Else
\For{$a=0$ to $q-1$}
\State Robot $\rho^a(r_i)$ moves toward $c$ using \textsc{MoveToDestination}$(\rho^a(r_i),c)$
\EndFor
\EndIf
\EndWhile

\While{there exists a surplus robot $r\in\mathcal R_s(t)$ such that $r\neq c$}
\State Let $\mathcal R_s^u(t)=\{r\in\mathcal R_s(t):r\neq c\}$ be the set of surplus robots not yet gathered at $c$
\State Partition $\mathcal R_s^u(t)$ into disjoint rotational orbits under $\rho$
\State Let $\widehat{\mathcal R}_s$ contain exactly one representative robot from each remaining rotational orbit

\For{each representative robot $r_j\in\widehat{\mathcal R}_s$}
    \State Compute $D_j=d_{\mathcal M}(r_j,c)$
\EndFor

\State Let $D_{\min}=\min\{D_j:r_j\in\widehat{\mathcal R}_s\}$
\State Select a rotational orbit $\mathbf{Orbit}(r_j)$ attaining $D_{\min}$
\State Break ties using the ordering $\mathscr O_4$

\If{some robot in $\mathbf{Orbit}(r_j)$ is in a pending move state}
    \State All other robots in $\mathbf{Orbit}(r_j)$ remain stationary
    \State Execute Algorithm~\ref{alg:handle-pending-move} for the pending robot until rotational symmetry is restored
\Else
    \For{$a=0$ to $q-1$}
        \State Robot $\rho^a(r_j)$ moves toward $c$ using \textsc{MoveToDestination}$(\rho^a(r_j),c)$
    \EndFor
\EndIf

\State Update the number of surplus robots gathered at $c$

\EndWhile

\State \Return

\end{algorithmic}
\end{algorithm}

%\FloatBarrier

\subsection{Routine \textit{MoveToDestination()}}

The routine \textit{MoveToDestination()} specifies how a selected robot moves toward its designated destination node during the execution of Algorithm \textbf{\textsc{spg()}}. This routine is invoked in the \textit{Line Formation Phase}, \textit{Multiplicity Creation Phase}, \textit{Saturation Phase}, and \textit{Gathering Phase}. In each phase, the admissible robot is selected according to the ordering defined for the corresponding configuration.

During the \textit{Line Formation Phase}, let $\mathcal S$ denote the corresponding reference structure, where $\mathcal S$ is either the line of symmetry $\mathcal L$ or the center of rotational symmetry $c$. For each robot $r_i(t)$, let $v_i$ denote its assigned destination node on the reference structure. Define $
D_i(t)=d_{\mathcal M}(r_i(t),\mathcal S),
$ where $d_{\mathcal M}$ denotes the Manhattan distance. The admissible robot is selected as
$
r^*(t)\in\arg\max_{r_i(t)\in\mathcal R(t)}D_i(t),
$
with ties broken according to the ordering prescribed for the corresponding configuration.

During the \textit{Multiplicity Creation Phase}, \textit{Saturation Phase}, and \textit{Gathering Phase}, let $v_i$ denote the destination node assigned to robot $r_i(t)$. Let $\mathcal R_a(t)$ denote the set of robots eligible to move in the current phase. The admissible robot is selected as
$
r^*(t)\in\arg\min_{r_i(t)\in\mathcal R_a(t)}
d_{\mathcal M}(r_i(t),v_i),
$
where ties are again broken according to the ordering defined for the corresponding configuration. Once the admissible robot $r^*(t)$ is selected, it remains the only moving robot until it reaches its assigned destination node $v^*$. More precisely, for every intermediate time $t'>t$ before $r^*$ reaches $v^*$, all other robots remain stationary. The robot $r^*$ moves along a shortest Manhattan path from its current position to $v^*$. At each move, the robot advances by one grid hop along this shortest path. 
The robot continues moving until $
d_{\mathcal M}(r^*(t),v^*)=0,
$
at which point it reaches its destination and becomes stationary. The ordering is then recomputed over the remaining eligible robots, and the next admissible robot is selected for movement. Since at most one robot moves at any time while all others remain stationary, collisions and non-designated multiplicity nodes are avoided.

If the configuration $\mathcal C(t_0)$ admits a line of symmetry $\mathcal L$, then a symmetric pair of robots, say $(r,\mu(r))$, is selected for movement according to the appropriate ordering, where $\mu(\cdot)$ denotes the reflection mapping with respect to $\mathcal L$. The robots $r$ and $\mu(r)$ move simultaneously toward their respective destination nodes along shortest Manhattan paths that are symmetric with respect to $\mathcal L$. Hence, the symmetry of $\mathcal C(t)$ with respect to $\mathcal L$ is preserved throughout the movement. Moreover, these two paths do not intersect before the robots reach their destination nodes; otherwise, by symmetry, an undesired multiplicity or collision would be created before the designated destination is reached.

Similarly, if the configuration $\mathcal C(t_0)$ admits a rotational symmetry of order $k$, then all robots belonging to the same rotational orbit
$\mathbf{Orbit}(r)={r,\rho(r),\rho^2(r),\ldots,\rho^{k-1}(r)}$
are selected for movement according to the appropriate ordering, where $\rho$ denotes the rotation operator about the center of rotational symmetry. These robots move simultaneously toward their respective destination nodes along shortest Manhattan paths that are rotationally symmetric images of one another. Therefore, the rotational symmetry of $\mathcal C(t)$ is maintained at every time instant during the movement. Furthermore, the paths of robots in the same orbit do not intersect before reaching their corresponding destination nodes. 
% Consequently, collisions and undesired multiplicity nodes are avoided in both reflectional and rotational symmetric cases. The pseudo-code corresponding to this routine is given in Algorithm \ref{alg:move-to-destination}.

\section{Correctness of the Algorithm \textbf{\textsc{spg()}}}
\label{sec:correctness}
We prove the correctness of the algorithm \textbf{\textsc{spg()}} by establishing the correctness of its four phases. We first state a common property of the movement routine, which will be used repeatedly in the phase-wise proofs.

\begin{lemma}
\label{lem:move-correct}
The procedure \textit{MoveToDestination()} guarantees that the designated robot, symmetric pair of robots, or rotational orbit of robots reaches its assigned destination after finitely many hops. Throughout its execution, the required symmetry is preserved, and no unintended collision or multiplicity node is created.
\end{lemma}

\begin{proof}
We prove the statement by considering the possible types of selected entities (robots/ parking nodes) in the procedure \textit{MoveToDestination()}.

\textbf{Case 1: A single designated robot is selected.}
Let $r$ be the designated robot and let $v$ be its assigned destination. The robot $r$ moves along a shortest Manhattan path from its current node to $v$. After each hop, the Manhattan distance from $r$ to $v$ decreases by exactly one. Since this distance is a nonnegative integer, after finitely many hops it becomes zero. Hence, $r$ reaches its destination $v$ in finite time. During this movement, no other robot is allowed to move. Moreover, $r$ is allowed to stop on an occupied node only if that node is explicitly prescribed by the current phase, such as a saturated parking node, the multiplicity node, or the gathering node. Therefore, no unintended collision or multiplicity node is created.

\textbf{Case 2: A symmetric pair of robots is selected.}
Suppose that a pair of robots $(r,\mu(r))$ is selected with respect to the line of symmetry $\mathcal L$, where $\mu$ denotes the reflection map. Let their assigned destinations be $v$ and $\mu(v)$, respectively. If $r\notin \mathcal L$, then $r$ and $\mu(r)$ lie on opposite sides of $\mathcal L$; otherwise, $r=\mu(r)$ lies on $\mathcal L$ and is treated as a single designated robot using an appropriate ordering. Assume first that $r\notin \mathcal L$. The robot $r$ moves along a shortest Manhattan path from $r$ to $v$, while $\mu(r)$ moves along the reflected copy of the same path from $\mu(r)$ to $\mu(v)$. Hence, after every hop, the Manhattan distance of each robot from its own destination decreases by one. Since these distances are nonnegative integers, both robots reach their respective destinations after finitely many hops. Moreover, the two prescribed paths are mirror images of each other with respect to $\mathcal L$. Therefore, the reflectional symmetry of the configuration is preserved after every hop. Since the two robots lie on opposite sides of $\mathcal L$ and move along symmetric paths, their paths do not intersect before reaching their assigned destinations. No robot outside the selected pair is allowed to move during this execution. A collision may occur only when a destination node is explicitly prescribed by the current phase of the algorithm, namely, a multiplicity node, a saturated parking node, or the gathering node. Therefore, any such collision is intentional and occurs only at a prescribed destination. Thus, no unintended collision or multiplicity node is created. If a robot lying on $\mathcal L$ is selected for an intentional symmetry-breaking move, then this movement is not treated as a symmetric-pair movement; it is handled separately as a single designated-robot movement using the prescribed ordering and the procedure \textsc{AllowtoMove}().

\textbf{Case 3: A rotational orbit of robots is selected.}
Suppose that a rotational orbit of robots is selected with respect to the center of rotational symmetry $c$. Let $\rho$ be the rotational map defined in Definition~\ref{def:rotational-map}, and let
$
\mathbf{Orbit}_r=\{r,\rho(r),\rho^2(r),\ldots,\rho^{q-1}(r)\}
$
be the selected orbit. Let $v$ be the assigned destination of the representative robot $r$. Then, for each $i\in\{0,1,\ldots,q-1\}$, the robot $\rho^i(r)$ is assigned the destination $\rho^i(v)$. The representative robot $r$ moves along a shortest Manhattan path from $r$ to $v$. For every $i$, the robot $\rho^i(r)$ moves along the path obtained by applying the rotational map $\rho^i$ to the path of $r$. Thus, all robots in the orbit move along rotated copies of the same prescribed path toward their respective destinations. After every hop, the Manhattan distance between each robot $\rho^i(r)$ and its assigned destination $\rho^i(v)$ decreases by one. Since these distances are nonnegative integers, all robots in the selected rotational orbit reach their assigned destinations after finitely many hops. Moreover, the destinations and the prescribed paths are closed under the rotational map $\rho$. Hence, after every hop, the image of the position of $\rho^i(r)$ under $\rho$ coincides with the position of $\rho^{i+1}(r)$, for all $i$ modulo $q$. Therefore, the rotational symmetry of the configuration is preserved throughout the execution. No robot outside the selected rotational orbit is allowed to move during this procedure. Also, by the construction of the prescribed paths, the moving robots do not create any unintended collision before reaching their destinations. A robot may terminate on an occupied node only when that node is explicitly designated by the current phase of the algorithm, namely, a multiplicity node, a saturated parking node, or the gathering node. Therefore, no unintended collision or multiplicity node is created.

Combining the three cases, the procedure \textit{MoveToDestination()} guarantees that the selected robot, symmetric pair, or rotational orbit reaches its assigned destination after finitely many hops, preserves the required symmetry, and avoids every unintended collision or multiplicity node.
\end{proof}

\begin{lemma}
\label{lem:line-formation-correct}
For every initial configuration $\mathcal C(t_0)\in\{\mathcal C_{21},\mathcal C_{31}\}$, the \textit{Line Formation Phase} terminates in finite time with all robots occupying distinct nodes of a uniquely identifiable formation line $\mathscr L$.
\end{lemma}

\begin{proof}
We first show that the formation line is uniquely determined.

\smallskip
\textbf{Case 1.} $\mathcal C(t_0)\in\mathcal C_{21}$. Since the parking-node set $\mathcal P$ admits a unique line of reflectional symmetry $\mathcal L$, every robot can independently compute $\mathcal L$, the rectangle $\mathcal{MER}_{\mathcal P}$, and the Manhattan distance of every robot from $\mathcal L$. The robot $r_g$ is then selected from the set of robots farthest from $\mathcal L$, with ties broken according to the ordering $\mathscr O_2$. As the robot configuration is asymmetric, the ordering $\mathscr O_2$ yields a unique choice of $r_g$. Let $u$ denote the point where the line through $r_g$ perpendicular to $\mathcal L$ intersects the nearest side of $\mathcal{MER}_{\mathcal P}$. The node $w$ is then uniquely determined by the condition $d_{\mathcal M}(u,w)=2$ and by requiring $w$ to lie on the ray from $u$ directed away from $\mathcal{MER}_{\mathcal P}$. Hence, the formation line $\mathscr L$, defined as the line passing through $w$ and parallel to $\mathcal L$, is uniquely determined.

\smallskip
\textbf{Case 2.} $\mathcal C(t_0)\in\mathcal C_{31}$. The center of rotational symmetry $c$ is uniquely determined. Every robot can therefore compute its Manhattan distance from $c$. The robot $r_g$ is selected from the set of robots farthest from $c$, with ties broken according to the ordering $\mathscr O_2$. Since the ordering is deterministic, the selected robot $r_g$ is unique. Consequently, the grid line passing through $c$ and $r_g$ is uniquely determined. Let $u$ denote the point where this line meets the boundary of $\mathcal{MER}_{\mathcal P}$ in the direction of $r_g$. The node $w$ is then uniquely determined by the condition $d_{\mathcal M}(u,w)=2$ and by requiring $w$ to lie beyond $\mathcal{MER}_{\mathcal P}$ along the same line. Hence, the formation line $\mathscr L$, defined as the line passing through $w$ and perpendicular to the line joining $c$ and $r_g$, is uniquely determined.

Thus, the formation line $\mathscr L$ is uniquely determined in both configuration classes. If the designated destination node on $\mathscr L$ is already occupied, the auxiliary-line rule is applied to locate another free node on $\mathscr L$. Since only finitely many nodes of $\mathscr L$ are occupied by robots, whereas $\mathscr L$ contains infinitely many grid nodes, such a free node always exists. Hence, every robot can be assigned a distinct destination node on $\mathscr L$.

By Lemma~\ref{lem:move-correct}, each selected robot reaches its assigned destination after finitely many hops without creating any collision or unintended multiplicity node. Once a robot reaches its destination on $\mathscr L$, it remains stationary and is never selected again. Therefore, the number of robots outside $\mathscr L$ decreases by one after each successful movement and never increases. Since the number of robots is finite, after finitely many movements every robot occupies a distinct node on $\mathscr L$. Hence, the \textit{Line Formation Phase} terminates with all robots positioned at the distinct nodes of the uniquely determined formation line $\mathscr L$.
\end{proof}

\begin{lemma}
\label{lem:multiplicity-correct}
After the \textit{Line Formation Phase}, the \textit{Multiplicity Creation Phase} correctly creates a unique and stable multiplicity node in finite time.
\end{lemma}

\begin{proof}
By Lemma~\ref{lem:line-formation-correct}, after the \textit{Line Formation Phase}, all robots occupy distinct nodes on the uniquely determined formation line $\mathscr L$.

\smallskip
\textbf{Case 1.} $\mathcal C(t_0)\in\mathcal C_{21}$. The terminal robots $r_1$ and $r_n$ on $\mathscr L$ uniquely determine the segment $[r_1,r_n]$. If this segment contains an odd number of nodes, then its middle node is unique. Otherwise, the segment contains exactly two middle nodes. Since the parking-node set $\mathcal P$ has a unique line of symmetry, among these two nodes, the one uniquely closer, say  $\mathbf N_{\mathscr L}$ to the two leading corners of $\mathcal{MER}_{\mathcal P}$ is selected. Hence, every robot identifies the same node $\mathbf N_{\mathscr L}$. If $\mathbf N_{\mathscr L}$ is already occupied by a robot, then that robot remains stationary, and another designated robot is moved to $\mathbf N_{\mathscr L}$. The designated robot is chosen as the closest robot on $\mathscr L$ to $\mathbf N_{\mathscr L}$. If two such closest robots exist, then they lie on opposite sides of $\mathbf N_{\mathscr L}$ along $\mathscr L$. Otherwise, a designated robot(s) first moves to $\mathbf N_{\mathscr L}$ and remains stationary there. Then, another designated robot(s) moves to the same node. Consequently, $\mathbf N_{\mathscr L}$ becomes the unique multiplicity node.

\smallskip
\textbf{Case 2.} $\mathcal C(t_0)\in\mathcal C_{31}$. The center of rotational symmetry $c$ and the formation line $\mathscr L$ are uniquely determined in a configuration $\mathcal C_{31}$. After the line-formation phase is completed, all robots lie on $\mathscr L$, and hence every robot identifies the same node $\mathbf N_c$ on $\mathscr L$. Once $\mathbf N_c$ is identified, the procedure for creating a unique and stable multiplicity node is the same as that used for configuration $\mathcal C_{21}$.
\smallskip

By Lemma~\ref{lem:move-correct}, every designated robot reaches its assigned destination after finitely many hops without creating any undesired collision or unintended multiplicity node. Once a robot reaches the designated node, it remains stationary. Therefore, after finitely many movements, a unique and stable multiplicity node is created.
\end{proof}

\begin{lemma}
\label{lem:phase-identification}
For configurations $\mathcal C(t_0)\in\{\mathcal C_{21},\mathcal C_{31}\}$, the multiplicity node created in the \textit{Multiplicity Creation Phase} enables every robot to distinguish the completion of the \textit{Line Formation Phase} from the beginning of the \textit{Saturation Phase}.
\end{lemma}

\begin{proof}
By Lemma~\ref{lem:line-formation-correct}, at the completion of the \textit{Line Formation Phase}, all robots occupy distinct nodes on the uniquely determined formation line $\mathscr L$. The following cases are to be considered.

\smallskip
\smallskip
\textbf{Case 1.} $\mathcal C(t_0)\in\mathcal C_{21}$. By Lemma~\ref{lem:line-formation-correct}, the terminal configuration of the \textit{Line Formation Phase} consists of all robots occupying distinct nodes on the formation line $\mathscr L$. Furthermore, by Lemma~\ref{lem:multiplicity-correct}, the \textit{Multiplicity Creation Phase} creates a unique and stable multiplicity node $\mathbf N_{\mathscr L}$. Since the robots have global strong multiplicity detection, every robot can determine whether a multiplicity node exists. Therefore, a configuration in which all robots lie on $\mathscr L$ and no multiplicity node is present is uniquely recognized as the terminal configuration of the \textit{Line Formation Phase}. After the unique multiplicity node $\mathbf N_{\mathscr L}$ is created, every robot recognizes the resulting configuration as the initial configuration of the \textit{Saturation Phase}. Hence, the completion of the \textit{Line Formation Phase} is unambiguously distinguishable from the beginning of the \textit{Saturation Phase}.

\smallskip
\smallskip
\textbf{Case 2.} $\mathcal C(t_0)\in\mathcal C_{31}$. By Lemma~\ref{lem:line-formation-correct}, the terminal configuration of the \textit{Line Formation Phase} consists of all robots occupying distinct nodes on the formation line $\mathscr L$. Furthermore, by Lemma~\ref{lem:multiplicity-correct}, the \textit{Multiplicity Creation Phase} creates a unique and stable multiplicity node $\mathbf N_c$. Since the robots have global strong multiplicity detection, every robot can determine whether a multiplicity node exists. Therefore, a configuration in which all robots lie on $\mathscr L$ and no multiplicity node is present is uniquely recognized as the terminal configuration of the \textit{Line Formation Phase}. After the unique multiplicity node $\mathbf N_c$ is created, every robot recognizes the resulting configuration as the initial configuration of the \textit{Saturation Phase}. Hence, the completion of the \textit{Line Formation Phase} is unambiguously distinguishable from the beginning of the \textit{Saturation Phase}.

\smallskip

Hence, for both configuration classes, the unique multiplicity node enables every robot to distinguish the end of the \textit{Line Formation Phase} from the beginning of the \textit{Saturation Phase}.
\end{proof}

\begin{lemma}
\label{lem:saturation-correct}
For every configuration in which the \textit{Saturation Phase} is executed, the procedure saturates every parking node to its prescribed capacity in finite time, without disturbing previously saturated parking nodes or creating unintended multiplicity nodes.
\end{lemma}

\begin{proof}
We first establish that, in every case, the next parking node, symmetric pair of parking nodes, or rotational orbit of parking nodes to be saturated is uniquely determined.

\smallskip
\textbf{Case 1.} $\mathcal C(t_0)\in\mathcal C_1$.
Since $\mathcal P$ is asymmetric, all parking nodes are uniquely ordered according to $\mathscr O_1$. Hence, the next parking node to be saturated is uniquely determined.

\smallskip
\textbf{Case 2.} $\mathcal C(t_0)\in\mathcal C_{21}$. The parking nodes lying on the line of symmetry $\mathcal L$ are saturated first according to $\mathscr O_2$. The remaining parking nodes are ordered by their Manhattan distance from the multiplicity node $\mathbf N_{\mathscr L}$, with ties broken according to $\mathscr O_2$. Hence, the next parking node or symmetric pair of parking nodes to be saturated is uniquely determined.

\smallskip
\textbf{Case 3.} $\mathcal C(t_0)\in\mathcal C_{31}$. If a parking node is located at the center of rotational symmetry $c$, it is saturated first. The remaining parking nodes are ordered by their Manhattan distance from the multiplicity node $\mathbf N_c$, with ties broken according to $\mathscr O_2$. Hence, the next parking node or rotational orbit of parking nodes to be saturated is uniquely determined.

\smallskip
\textbf{Case 4.} $\mathcal C(t_0)\in\{\mathcal C_{221},\mathcal C_{321}\}$.
The prescribed symmetry-breaking operation transforms the configuration into an asymmetric configuration of type $\mathcal C_{21}$ or $\mathcal C_{31}$, respectively. Therefore, the corresponding ordering established in Cases 2 and 3 applies uniquely.

\smallskip
\textbf{Case 5.} $\mathcal C(t_0)\in\mathcal C_{222}$. The \textit{Gathering Phase} first creates the gathering node $\mathscr G$ on the line of symmetry $\mathcal L$. The parking nodes are then considered in reflected pairs, and the ordering $\mathscr O_4$ uniquely selects the next pair to be saturated.

\smallskip
\textbf{Case 6.} $\mathcal C(t_0)\in\mathcal C_{322}$. The \textit{Gathering Phase} first creates the gathering node $\mathscr G$ on $\mathcal c$. The parking nodes are grouped into rotational orbits about the center $c$, and the ordering $\mathscr O_4$ uniquely selects the next orbit to be saturated.

Thus, in every configuration class, the next parking node, symmetric pair of parking nodes, or rotational orbit of parking nodes is uniquely determined.

For each selected parking node $p$ with prescribed capacity $\kappa(p)$, the unsaturated robots are ordered by their Manhattan distance from $p$, with ties broken according to the prescribed ordering. Hence, the robots assigned to saturate $p$ are uniquely determined. In symmetric configurations, the same rule is applied pairwise or orbit-wise, thereby preserving the required reflectional or rotational symmetry. Strong multiplicity detection allows every robot to determine when the occupancy of a parking node reaches its prescribed capacity. Once a parking node becomes saturated, it is never selected again, and the robots occupying it remain stationary. Therefore, previously saturated parking nodes remain unchanged.

Define
$
\Phi(t)=\sum_{p\in\mathcal P}\max\{0,\kappa(p)-\operatorname{occ}_t(p)\},
$
where $\operatorname{occ}_t(p)$ denotes the number of robots occupying the parking node $p$ at time $t$. Since the number of parking nodes is finite, $\Phi(t)$ is a finite nonnegative integer. Whenever a selected robot reaches an unsaturated parking node, the occupancy of that node increases by one, and hence $\Phi(t)$ decreases by one. By Lemma~\ref{lem:move-correct}, every selected robot reaches its assigned destination after finitely many hops without creating any undesired collision or unintended multiplicity node. Therefore, $\Phi(t)$ eventually reaches zero. Consequently, every parking node attains its prescribed capacity after finitely many movements, previously saturated parking nodes remain unchanged, and no unintended multiplicity node is created.
\end{proof}

\begin{lemma}
\label{lem:gathering-correct}
For every solvable configuration, the \textit{Gathering Phase} gathers all surplus robots at a uniquely identifiable gathering node in finite time, without disturbing saturated parking nodes.
\end{lemma}

\begin{proof}
We first establish that the gathering node is uniquely determined.

\smallskip
\textbf{Case 1.} $\mathcal C(t_0)\in\mathcal C_1$.
Since $\mathcal P$ is asymmetric, the leading corner $\mathcal K_p$ of $\mathcal{MER}_{\mathcal P}$ and the corresponding string direction are uniquely determined. Hence, the node $\mathscr G$ adjacent to $\mathcal K_p$ and lying outside $\mathcal{MER}_{\mathcal P}$ is uniquely determined.

\smallskip
\textbf{Case 2.} $\mathcal C(t_0)\in\mathcal C_{21}$. The gathering node is the unique multiplicity node $\mathbf N_{\mathscr L}$ created during the \textit{Multiplicity Creation Phase}.

\smallskip
\textbf{Case 3.} $\mathcal C(t_0)\in\mathcal C_{31}$. The gathering node is the unique multiplicity node $\mathbf N_c$ created during the \textit{Multiplicity Creation Phase}.

\smallskip
\textbf{Case 4.} $\mathcal C(t_0)\in\{\mathcal C_{221},\mathcal C_{321}\}$. The prescribed symmetry-breaking operation transforms the configuration into one of type $\mathcal C_{21}$ or $\mathcal C_{31}$, respectively. Hence, by Cases~2 and~3, the gathering node is uniquely determined.

\smallskip
\textbf{Case 5.} $\mathcal C(t_0)\in\mathcal C_{222}$. The gathering node $\mathscr G$ is selected on the line of symmetry $\mathcal L$ using the closest symmetric pair of surplus robots, with ties broken according to $\mathscr O_4$. Hence, $\mathscr G$ is uniquely determined.

\smallskip
\textbf{Case 6.} $\mathcal C(t_0)\in\mathcal C_{322}$. The gathering node is the uniquely determined as the center of rotational symmetry $c$. Thus, in every configuration class, the gathering node is uniquely determined.

Only surplus robots participate in the \textit{Gathering Phase}. Robots already occupying saturated parking nodes are never selected and therefore remain stationary. Hence, previously saturated parking nodes remain unchanged. The surplus robots are ordered by their Manhattan distance from the gathering node, with ties broken according to the ordering $\mathscr O_4$. Consequently, the next surplus robot, symmetric pair of surplus robots, or rotational orbit of surplus robots is uniquely determined. In symmetric configurations, the robots move in reflected pairs or rotational orbits, while pending movements are handled by allowing the corresponding paired or orbit robots to wait until symmetry is restored. Therefore, the required symmetry is preserved throughout the gathering process.

Let $\mathscr G$ denote the gathering node and define
$
\Psi(t)=\left|\left\{r_i\in\mathcal R_s(t):r_i\neq\mathscr G\right\}\right|.
$
Since the number of surplus robots is finite, $\Psi(t)$ is a finite nonnegative integer. Whenever a selected surplus robot reaches $\mathscr G$, it remains stationary, and $\Psi(t)$ decreases by one. In reflection-symmetric configurations, $\Psi(t)$ decreases by two, whereas in rotationally symmetric configurations it decreases by the size of the selected rotational orbit. By Lemma~\ref{lem:move-correct}, every selected surplus robot reaches $\mathscr G$ after finitely many hops without creating any undesired collision or unintended multiplicity node. Therefore, $\Psi(t)$ eventually reaches zero. Consequently, after finitely many movements, all surplus robots are gathered at the uniquely determined gathering node $\mathscr G$, while previously saturated parking nodes remain unchanged.
\end{proof}

\begin{theorem}
\label{thm:spg-correct}
For every solvable initial configuration, the algorithm \textbf{\textsc{spg()}} correctly solves the surplus parking gathering problem.
\end{theorem}

\begin{proof}
We consider each solvable configuration class separately. The following cases are to be considered.

\smallskip
\textbf{Case 1.} $\mathcal C(t_0)\in\mathcal C_1$.
The algorithm first executes the \textit{Saturation Phase}. By Lemma~\ref{lem:saturation-correct}, every parking node is saturated to its prescribed capacity. The algorithm then executes the \textit{Gathering Phase}. By Lemma~\ref{lem:gathering-correct}, all surplus robots gather at the uniquely determined gathering node.

\smallskip
\textbf{Case 2.} $\mathcal C(t_0)\in\mathcal C_{21}$. The algorithm first executes the \textit{Line Formation Phase}. By Lemma~\ref{lem:line-formation-correct}, all robots occupy distinct nodes on the uniquely determined formation line. Next, the \textit{Multiplicity Creation Phase} creates a unique and stable multiplicity node by Lemma~\ref{lem:multiplicity-correct}. The algorithm then executes the \textit{Saturation Phase}, which correctly saturates every parking node by Lemma~\ref{lem:saturation-correct}. Finally, the \textit{Gathering Phase} gathers all surplus robots at the gathering node by Lemma~\ref{lem:gathering-correct}.

\smallskip
\textbf{Case 3.} $\mathcal C(t_0)\in\mathcal C_{31}$. The proof is identical to Case~2. The \textit{Line Formation Phase}, \textit{Multiplicity Creation Phase}, \textit{Saturation Phase}, and \textit{Gathering Phase} are executed sequentially, and their correctness follows from Lemmas~\ref{lem:line-formation-correct}, \ref{lem:multiplicity-correct}, \ref{lem:saturation-correct}, and \ref{lem:gathering-correct}, respectively.

\smallskip
\textbf{Case 4.} $\mathcal C(t_0)\in\{\mathcal C_{221},\mathcal C_{321}\}$.
The prescribed symmetry-breaking operation first transforms the configuration into one of type $\mathcal C_{21}$ or $\mathcal C_{31}$, respectively. Therefore, the correctness follows directly from Cases~2 and~3.

\smallskip
\textbf{Case 5.} $\mathcal C(t_0)\in\mathcal C_{222}$. The algorithm first executes the \textit{Gathering Phase}. By Lemma~\ref{lem:gathering-correct}, all surplus robots gather at the selected gathering node while preserving the required reflectional symmetry. The algorithm then executes the \textit{Saturation Phase}, which saturates every parking node to its prescribed capacity by Lemma~\ref{lem:saturation-correct}.

\smallskip
\textbf{Case 6.} $\mathcal C(t_0)\in\mathcal C_{322}$. The algorithm first executes the \textit{Gathering Phase}. By Lemma~\ref{lem:gathering-correct}, all surplus robots gather at the selected gathering node while preserving the required rotational symmetry. The algorithm then executes the \textit{Saturation Phase}, which saturates every parking node to its prescribed capacity by Lemma~\ref{lem:saturation-correct}.

\smallskip

Since each phase decreases its associated progress measure and terminates after finitely many robot movements, the algorithm cannot deadlock before reaching the final configuration. Therefore, for every solvable configuration class, the algorithm terminates after finitely many movements and reaches a configuration in which every parking node is saturated to its prescribed capacity and all surplus robots are gathered at the uniquely determined gathering node. Moreover, by Lemma~\ref{lem:move-correct}, the required symmetry is preserved throughout the execution, and no undesired collision or unintended multiplicity node is created. Hence, \textbf{\textsc{spg()}} correctly solves the surplus parking gathering problem for every solvable initial configuration.
\end{proof}
\section{Move Complexity of the Algorithm}
\label{sec:complexity}

In this section we analyze the efficiency of the algorithm \textsc{spg}() in terms of
the total number of hop-moves performed by all robots (\emph{move complexity}).
Throughout, let $a,b$ denote the side lengths of $\mathcal{MER}_\mathcal{C}(t_0)$, and
write $D = a+b$. Recall that $n$ is the total number of robots, $m$ the number of
parking nodes, and $s = n - \sum_{i=1}^m \kappa_i$ the number of surplus robots.

\subsection{Upper Bound}
\begin{theorem}
\label{thm:upper}
The algorithm \textsc{spg}() has move complexity
$O\bigl(n(a+b)+n^2\bigr)$,
where $a$ and $b$ denote the side lengths of the initial minimum enclosing
rectangle $\mathcal{MER}_{\mathcal C}(t_0)$.
\end{theorem}

\begin{proof}
Let $D=a+b$, which is an upper bound on the Manhattan diameter of
$\mathcal{MER}_{\mathcal C}(t_0)$. We analyze the move complexity of each phase
separately.

Throughout the execution, the procedure
\textsc{MoveToDestination}() routes every selected robot (or symmetric pair or
rotational orbit) along a shortest Manhattan path.

\smallskip
\noindent
\textbf{Line Formation Phase.}
Each robot first reaches the formation line $\mathscr{L}$ from its initial
position. Since $\mathscr{L}$ is constructed inside
$\mathcal{MER}_{\mathcal C}(t_0)$ or at a constant offset from it, this
requires at most $O(D)$ moves per robot, contributing $O(nD)$ moves in total.
To place all robots on distinct consecutive nodes of $\mathscr{L}$, the occupied
segment of the line may grow to length $O(n)$. Consequently, the cumulative
displacement of the robots along $\mathscr{L}$ is bounded by
$\sum_{i=0}^{n-1} O(i)=O(n^2)$. Hence,
$M_{\mathrm{LF}}=O(nD+n^2)$.

\smallskip
\noindent
\textbf{Multiplicity Creation Phase.}
Only $O(1)$ robots move to create the multiplicity node, each over a shortest
path of length at most $O(D+n)$. Therefore,
$M_{\mathrm{MC}}=O(D+n)$.

\smallskip
\noindent
\textbf{Saturation Phase.}
Exactly $n-s$ robots move from the formation line to their assigned parking
nodes. Since every parking node lies inside
$\mathcal{MER}_{\mathcal P}\subseteq
\mathcal{MER}_{\mathcal C}(t_0)$, each robot traverses at most $O(D+n)$
edges. Therefore,
$M_{\mathrm{Sat}}=O((n-s)(D+n))=O(nD+n^2)$.

\smallskip
\noindent
\textbf{Gathering Phase.}
The remaining $s$ surplus robots move from the formation line to the gathering
node $\mathscr{G}$. Each movement has length at most $O(D+n)$, yielding
$M_{\mathrm{Gath}}=O(s(D+n))=O(nD+n^2)$.

\smallskip
\noindent
Summing the contributions of the four phases,
$M=M_{\mathrm{LF}}+M_{\mathrm{MC}}+M_{\mathrm{Sat}}+M_{\mathrm{Gath}}
=O(nD+n^2)+O(D+n)+O(nD+n^2)+O(nD+n^2)
=O(nD+n^2).$ Since $D=a+b$, the overall move complexity is
$O\bigl(n(a+b)+n^2\bigr)$.
\end{proof}

\begin{figure}[htbp]
    \centering
    \includegraphics[width=0.5\textwidth]{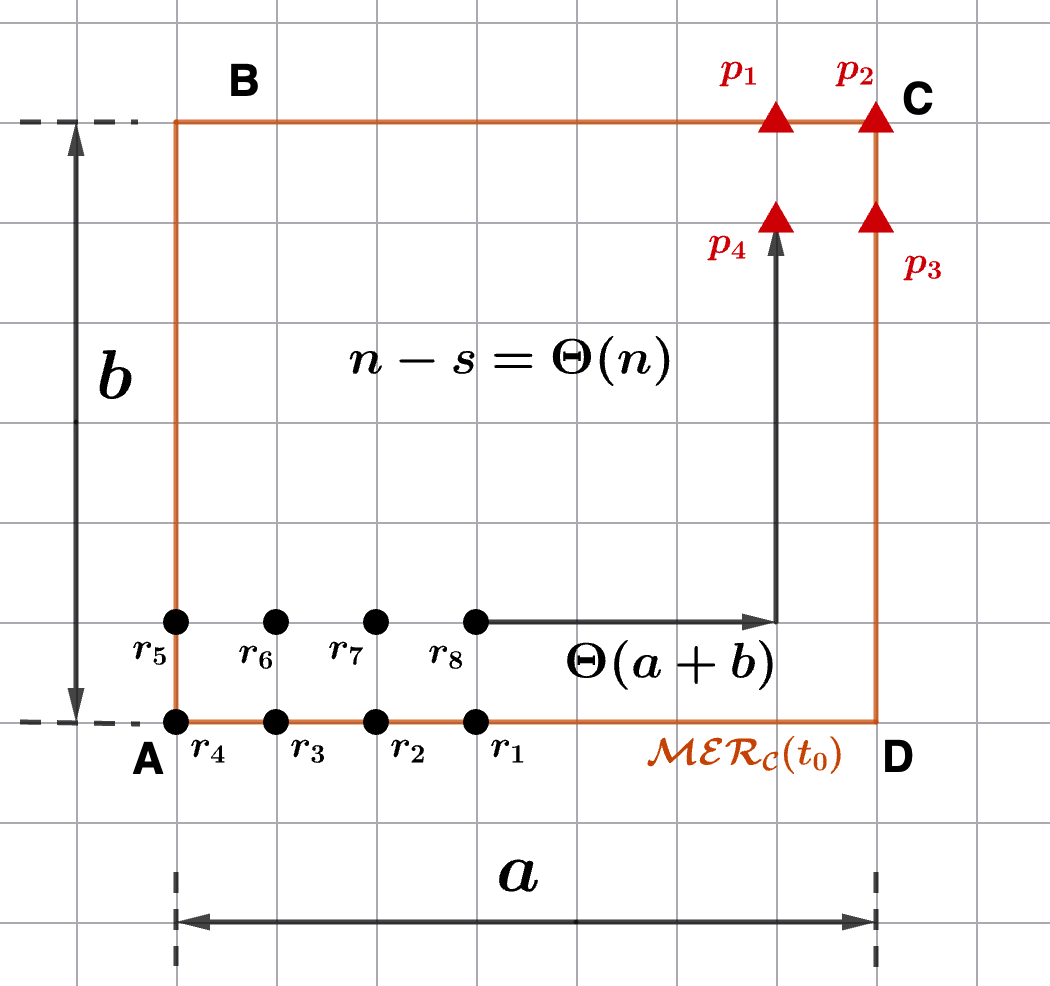}

\caption{An illustration of the worst-case instance used in the proof of
Theorem~\ref{thm:lower}. The robots are initially clustered near one corner of
the minimum enclosing rectangle, while all parking nodes are located near the
opposite corner. The capacities of the parking nodes $p_1,p_2,p_3$, and $p_4$
are $2,1,1$, and $1$, respectively. Since $\Theta(n)$ robots must each travel
a Manhattan distance of $\Theta(a+b)$ to reach their assigned parking nodes,
every correct algorithm requires $\Omega(n(a+b))$ robot moves.}

    % \caption{Leading-corner based ordering of robots and parking nodes with respect to $\mathcal{MER}_{\mathcal P}$. The capacities of $p_1,p_2,\ldots,p_6$ are $2,2,4,1,2,3$ and $p_6$ and $p_3$ contain $2$ and $3$ robots and a multiplicity node contains $4$ robots referred robot $r_9$. The lexicographic string associated with the key corner $A$ is  $\mathcal{STR}^{AD}(t)=$ ((1,0), (1,0), (0,0), (0,0), (0,0), (0,0), (1,0), (1,0), (0,0), (1,0), (0,0), (1,0), (1,0), (0,0), (0,0), (0,0), (1,0), (0,0), (0,0), (1,0), (0,0), (0,0), (0,0), (0,0), (0,0), (0,0), (1,0), (0,0), (0,0), (1,0), (0,0), (0,0), (0,2), (0,2), (0,0), (0,0), (4,0), (0,0), (2,3), (0,0), (0,0), (0,0), (0,0), (0,0), (1,0), (0,0), (0,0), (0,0), (0,0), (1,0), (0,0), (0,0), (0,0), (0,0), (1,0), (0,0), (1,0), (0,0), (0,0), (3,4), (0,0), (0,0), (0,1), (0,0), (1,0), (0,0), (0,2), (0,0), (1,0), (0,0), (0,0), (0,0), (0,0), (0,0), (0,0), (0,0), (1,0), (0,0)). The leading corner $D'$ of $\mathcal P$ is $\rho_{D'}=22000000000000300000000020001004$. }
    \label{fig:upperbound}
\end{figure}

% \begin{figure}[t]
% \centering
% \begin{tikzpicture}[scale=0.6]

% % MER
% \draw[very thick] (0,0) rectangle (12,7);

% \node[below left] at (0,0) {$A$};
% \node[above left] at (0,7) {$B$};
% \node[above right] at (12,7) {$C$};
% \node[below right] at (12,0) {$D$};

% \node at (6,7.7)
% {$\mathcal{MER}_{\mathcal C}(t_0)$};

% % Robot cluster
% \foreach \x in {0.7,1.3,1.9}{
%     \foreach \y in {0.7,1.3,1.9,2.5}{
%         \fill[blue] (\x,\y) circle (2.5pt);
%     }
% }

% \node[align=center] at (2,-0.8)
% {$n-s=\Theta(n)$\\robots};

% % Parking nodes
% \foreach \x/\y/\c in {
% 10.3/6.2/2,
% 11.0/6.2/3,
% 11.7/6.2/2,
% 10.6/5.4/1,
% 11.4/5.4/2}{
% \node[draw,fill=orange!30,minimum size=6mm]
% at (\x,\y){\tiny\c};
% }

% \node[align=center] at (10.9,7.8)
% {Parking nodes};

% % Manhattan path
% \draw[red,very thick,-latex]
% (2.3,2.5)--(10.8,2.5);

% \draw[red,very thick,-latex]
% (10.8,2.5)--(10.8,5.8);

% \node[red] at (6.2,2.9)
% {$\Theta(a+b)$};

% % Legend
% \fill[blue] (1,-2) circle (2.5pt);
% \node[right] at (1.2,-2){Robot};

% \node[draw,fill=orange!30,
% minimum size=5mm] at (5,-2){};

% \node[right] at (5.3,-2){Parking node};

% \end{tikzpicture}

% \label{fig:upperbound}
% \end{figure}
\begin{theorem}
\label{thm:lower}
In the worst case, any algorithm solving the
$\mathcal {SPG}$ problem requires
$\Omega(n(a+b))$ robot moves.
\end{theorem}

\begin{proof}
Let $D=a+b$ denote the Manhattan diameter of the initial minimum enclosing
rectangle $\mathcal{MER}_{\mathcal C}(t_0)$.
Consider the family of instances illustrated in
Figure~\ref{fig:upperbound}, in which all $n$ robots are initially clustered in a
small neighborhood of one corner of
$\mathcal{MER}_{\mathcal C}(t_0)$, whereas all parking nodes are clustered near
the opposite corner. The parking capacities are chosen such that
$n-s=\Theta(n)$ robots must eventually occupy parking nodes. Such an instance is
valid since the parking-node locations and capacities are specified as part of
the input. The Manhattan distance between the two clusters is
$\Theta(D)=\Theta(a+b)$. Consequently, every one of the
$n-s=\Theta(n)$ parking-bound robots must traverse at least
$\Theta(D)$ edges before reaching its assigned parking node, regardless of the
algorithm used. Hence every correct algorithm performs at least $
(n-s)\cdot\Theta(D)
=
\Theta(n)\cdot\Theta(D)
=
\Omega(nD)
$
robot moves. Substituting $D=a+b$ yields the desired lower bound
$\Omega(n(a+b))$.
\end{proof}

\begin{remark}
The above results establish a worst-case upper bound of
$O(n(a+b)+n^2)$ and a general lower bound of
$\Omega(n(a+b))$. The gap between these bounds remains open.
The additional $n^2$ term arises from the line-formation strategy employed by
\textsc{spg}(). Whether this term can be eliminated by designing a more efficient
distributed algorithm, or whether a stronger lower bound can be established for
the SPG problem, remains an interesting direction for future research.
\end{remark}
\section{Conclusion and Future Work}
\label{sec:conclusion}

In this paper, we introduced the $\ mathcal {SPG}$ problem, a new coordination problem for mobile robots deployed on the nodes of an infinite grid. Unlike the classical parking and gathering problems, $\mathcal{SPG}$ simultaneously requires every designated parking node to be saturated according to its prescribed capacity while all remaining surplus robots gather at a common grid node determined autonomously during the execution. We considered anonymous, oblivious, disoriented, and silent robots operating under the asynchronous model with global visibility and global strong multiplicity detection.

We first characterized the initial configurations from which $\mathcal{SPG}$ is deterministically unsolvable. For every remaining solvable configuration, we presented a deterministic distributed algorithm that correctly solves the problem. We proved that the algorithm is collision-free and terminates in finite time. Upon termination, every parking node is occupied according to its prescribed capacity, and all surplus robots are gathered at a uniquely determined gathering node.

Several interesting directions remain for future research. One natural extension is to study the case where the total number of robots is smaller than the total parking capacity. In this setting, an interesting objective is to maximize the number of saturated parking nodes. Another direction is to design algorithms that optimize movement complexity, for example by minimizing the total distance traveled by all robots or the maximum distance traveled by any individual robot. Another interesting direction for future work is to investigate whether the current upper bound can be improved, or whether a matching lower bound can be established for the $\mathcal {SPG}$ problem. Finally, extending the problem to other graph classes, such as trees, rings, and general graphs, or investigating fault-tolerant variants in the presence of robot failures, constitutes promising directions for future research.

%
% ---- Bibliography ----
%
% BibTeX users should specify bibliography style 'splncs04'.
% References will then be sorted and formatted in the correct style.
%
% \bibliographystyle{splncs04}
% \bibliography{mybibliography}

\bibliographystyle{plain} 
\bibliography{Biblo}

\end{document}